\documentclass[12pt,a4paper]{article}
\usepackage[english]{babel}
\usepackage[whole]{bxcjkjatype}

\usepackage[T1]{fontenc}
\usepackage[utf8]{inputenc}
\usepackage{lmodern}
\usepackage{parskip}

\usepackage{geometry}
\usepackage{verbatim}
\usepackage{float}
\usepackage{bm}
\usepackage{amsmath}
\usepackage{amssymb}
\usepackage{amsfonts}
\usepackage{amsthm}
\usepackage{mathtools}
\usepackage{graphicx}
\usepackage{dsfont}
\usepackage{array}
\usepackage{color}
\usepackage{enumitem}
\usepackage{authblk}
\usepackage{booktabs}

\numberwithin{equation}{section}
\mathtoolsset{showonlyrefs=true}

\usepackage{multirow}
\usepackage{makecell}
\usepackage{arydshln}
\usepackage[font=small]{caption}
\usepackage{subcaption}

\usepackage{algorithm}
\usepackage{algpseudocode}

\RequirePackage[authoryear]{natbib}

\RequirePackage[colorlinks,citecolor=blue,urlcolor=blue]{hyperref}
\usepackage[dvipsnames]{xcolor}
\usepackage{cleveref}

\usepackage{pgfplots}
\pgfplotsset{compat=1.18}
\usepackage{appendix}

\allowdisplaybreaks

\theoremstyle{plain}

\newtheorem{theorem}{Theorem}[section]
\newtheorem{lemma}[theorem]{Lemma}
\newtheorem{proposition}[theorem]{Proposition}
\newtheorem{corollary}[theorem]{Corollary}

\theoremstyle{definition}
\newtheorem{definition}[theorem]{Definition}
\newtheorem{assumption}[theorem]{Assumption}
\newtheorem{remark}[theorem]{Remark}
\newtheorem{example}[theorem]{Example}

\newcommand{\EE}{\mathbb{E}}
\newcommand{\NN}{\mathbb{N}}

\newcommand{\RR}{\mathbb{R}}

\newcommand{\cC}{\mathcal{C}}
\newcommand{\cD}{\mathcal{D}}
\newcommand{\cE}{\mathcal{E}}
\newcommand{\cF}{\mathcal{F}}

\newcommand{\cH}{\mathcal{H}}

\newcommand{\cK}{\mathcal{K}}

\newcommand{\cN}{\mathcal{N}}

\newcommand{\cR}{\mathcal{R}}

\newcommand{\cU}{\mathcal{U}}

\newcommand{\cW}{\mathcal{W}}

\newcommand{\tr}{\mathop{\mathrm{tr}}}

\newcommand{\diag}{\mathrm{diag}}
\newcommand{\dd}{\mathop{}\!\mathrm{d}}

\DeclareMathOperator{\TV}{\mathrm{TV}}

\DeclareMathOperator{\Bern}{\mathrm{Bern}}
\DeclareMathOperator{\var}{var}
\DeclareMathOperator{\cov}{cov}

\newcommand{\1}{\mathbf{1}}

\title{Location--Scale Calibration for Generalized Posteriors}
\author[1,2,$\ast$]{Shu Tamano}
\author[2]{Yui Tomo}
\affil[1]{\small
Department of Multidisciplinary Sciences, Graduate School of Arts and Sciences, The University of Tokyo, 3-8-1 Komaba, Meguro-ku, Tokyo 153-8902, Japan
}
\affil[2]{
\small
Department of Epidemiology, National Institute of Infectious Diseases, Japan Institute for Health Security, 1-23-1 Toyama, Shinjuku-ku, Tokyo 162-0052, Japan
}
\affil[$\ast$]{\small
Email:
\href{mailto:tamano-shu212@g.ecc.u-tokyo.ac.jp}
{\texttt{tamano-shu212@g.ecc.u-tokyo.ac.jp}}
}
\date{}

\begin{document}
\maketitle

\begin{abstract}
    General Bayesian updating replaces the likelihood with a loss scaled by a learning rate, but posterior uncertainty can depend sharply on that scale.
    We propose a simple post-processing that aligns generalized posterior draws with the asymptotic sampling distribution of a penalized estimator, yielding uncertainty quantification that is asymptotically invariant to the learning rate.
    We prove a total-variation normal approximation at an effective scale determined by a score central limit theorem, allowing sample-size-dependent prior factors, non-i.i.d.\ observations, and locally smooth convex penalties under model misspecification.
    Within this framework, we justify the open-faced sandwich adjustment \citep{shaby2014open-faced} and extend it from covariance rescaling to a location--scale calibration whose draws converge in total variation to the target for any learning rate.
    In our empirical illustration, calibrated draws maintain stable coverage, interval width, and bias over orders of magnitude in the learning rate and closely track frequentist benchmarks, whereas uncalibrated posteriors vary markedly.
\end{abstract}
\noindent
\textit{
Keywords:
Bayesian inference;
Bernstein--von Mises;
Generalized Bayes;
Learning rate;
Open-faced sandwich;
Penalized estimating equation;
Sandwich variance.
}

\section{Introduction}

Bayesian inference combines prior information with a likelihood to quantify uncertainty about an unknown parameter.
When the likelihood is misspecified, or learning is based on a quasi-likelihood, composite likelihood, or other empirical objective, posterior uncertainty can be miscalibrated \citep{kleijn2012bernstein,syring2019calibrating,miller2021asymptotic}.
General Bayesian updating replaces the likelihood with the exponential of a negative loss \citep{bissiri2016general}, allowing inference without specifying the entire data distribution.
A positive scaling parameter $\eta$, called the learning rate or inverse temperature, multiplies the loss and controls the concentration of the resulting generalized posterior.
Ordinary Bayesian updating is recovered when the loss is the negative log-likelihood and $\eta=1$.
With a fixed proper prior, decreasing $\eta$ attenuates the contribution of the data; under suitable integrability conditions the update reverts to the prior as $\eta\downarrow0$.
For a fixed data set, increasing $\eta$ favors minimizers of the empirical loss, and even as the sample size grows, the choice of $\eta$ continues to govern the concentration of the posterior around the population target.
Thus, the learning rate governs the trade-off between prior information and loss-based evidence and can substantially affect reported uncertainty.

Several approaches select the learning rate to achieve a particular inferential goal.
Bootstrap methods calibrate uncertainty to frequentist targets \citep{lyddon2019general,syring2019calibrating}, and \citet{matsubara2024generalized} use bootstrap minimizers in a score-matching criterion.
Other proposals adapt the learning rate for robustness \citep{grunwald2017inconsistency}, match prior-expected information gains \citep{holmes2017assigning}, or compare posterior distributions on an initial data window \citep{altamirano2023robust}.
Predictive performance provides another criterion \citep{mclatchie2025predictive}.
For parameter uncertainty, a scalar learning rate can match the posterior and sampling covariance matrices only when the score covariance is proportional to the curvature of the penalized loss.
Matrix adjustment addresses this distinction directly.
Sandwich corrections for quasi-posterior inference were studied by \citet{chernozhukov2003mcmc}, and the open-faced sandwich (OFS) adjustment of \citet{shaby2014open-faced} transforms posterior draws to match the estimator's sampling covariance.
Related normal approximations under misspecification and generalized updating were established by \citet{kleijn2012bernstein} and \citet{miller2021asymptotic}.

In this paper, we revisit the open-faced sandwich adjustment \citep{shaby2014open-faced} and develop an asymptotic framework for location--scale calibration of generalized posteriors with penalties, sample-size-dependent prior factors, and potentially dependent observations.
First, we show that, under sample-size-dependent priors, the generalized posterior admits a normal limit in total variation, and we show how scaling the penalty with the loss determines the target curvature and variability.
Second, building on this justification, we extend the open-faced sandwich adjustment from covariance rescaling to a location--scale calibration that is implementable via plug-in sandwich estimators, requiring only posterior draws and empirical moments, with no bootstrap or learning-rate tuning.
Third, we prove that the calibrated draws converge to the target distribution for any learning rate, so that the asymptotic uncertainty quantification is invariant to the choice of learning rate.
Our empirical illustration for ridge-penalized Huber regression with clustered observations shows that the calibrated draws maintain stable coverage and interval width over a wide range of learning rates, whereas the uncalibrated posterior varies markedly.

\section{Problem setup}
\label{sec:setting}

We first introduce the data, the loss, and the learning rates.
Let $\cD_n$ denote the observed data with sampling law $P^{\ast}$, where $n$ is the sample size.
Let $\boldsymbol{\theta}$ denote the parameter, taking values in an open parameter space $\Theta\subset\RR^p$ with $p$ fixed.
Fix a learning-rate interval $\cE=[\eta_L,\eta_U]$, where $0<\eta_L\le\eta_U<\infty$.
Let $M_n:\Theta\to(-\infty,\infty]$ be a jointly measurable empirical loss and $s_n$ a positive deterministic normalization with $s_n\to\infty$ as $n\to\infty$.
Write
\begin{equation*}
    \bar{M}_n(\boldsymbol{\theta})
    =
    s_n^{-1}M_n(\boldsymbol{\theta})
    ,\quad
    \boldsymbol{U}_n(\boldsymbol{\theta})
    =
    \nabla\bar{M}_n(\boldsymbol{\theta})
    ,\quad
    \boldsymbol{J}_n(\boldsymbol{\theta})
    =
    \nabla^2\bar{M}_n(\boldsymbol{\theta})
    ,
\end{equation*}
where the derivatives exist, and let $\Psi$ denote the population limit of the score $\boldsymbol{U}_n$.
We call $s_n$ the effective scale because the score satisfies a central limit theorem at rate $\sqrt{s_n}$, stated in \eqref{eq:score-clt} below.
For i.i.d.\ additive losses, the usual choice is $s_n=n$, and $\Psi$ is the expected single-observation loss gradient whenever differentiation under expectation is justified.
For other sampling designs, $s_n$ must be chosen so that the central limit theorem in \eqref{eq:score-clt} holds.

We next define the generalized posterior.
For a convex penalty $\rho$ and a deterministic regularization parameter $\lambda_n\ge0$ with $\lambda_n\to\lambda\ge0$ as $n\to\infty$, use the composite loss $M_n+\lambda_ns_n\rho$.
The generalized posterior is
\begin{equation}
\label{eq:posterior}
    \Pi_n^\eta(\dd\boldsymbol{\theta})
    =
    Z_n(\eta)^{-1}\exp\{-\eta s_nG_n(\boldsymbol{\theta})+r_n(\boldsymbol{\theta})\}\dd\boldsymbol{\theta}
    ,\quad
    G_n
    =
    \bar{M}_n+\lambda_n\rho
    ,\quad
    \eta\in\cE
    .
\end{equation}
Here $\exp(r_n)$ is a prior factor independent of the data and $\eta$, and $Z_n(\eta)$ is the normalizing constant.
A constant, possibly improper prior factor is allowed if the posterior is proper.
Equivalently,
\begin{equation}
\label{eq:eta-prior}
    \Pi_n^\eta(\dd\boldsymbol{\theta})
    \propto
    \exp\{-\eta M_n(\boldsymbol{\theta})\}\pi_n^\eta(\boldsymbol{\theta})\dd\boldsymbol{\theta}
    ,\quad
    \pi_n^\eta(\boldsymbol{\theta})
    \propto
    \exp\{-\eta\lambda_ns_n\rho(\boldsymbol{\theta})+r_n(\boldsymbol{\theta})\}
    ,
\end{equation}
where $\pi_n^\eta$ denotes the prior density that combines the penalty with the prior factor.
Thus the penalty acts as a sample-size-dependent prior, and the learning rate scales the entire composite loss, while the prior information in $r_n$ is retained \citep{bissiri2016general}.
Because the loss and the penalty are scaled together, the target defined next does not depend on $\eta$.

We now define the target parameter and the target covariance.
The target $\boldsymbol{\theta}^{\lambda}$ is an interior solution of the penalized population equation
\begin{equation}
\label{eq:pop-penalized}
    \boldsymbol{0}
    =
    \Psi(\boldsymbol{\theta})
    +
    \lambda\nabla\rho(\boldsymbol{\theta})
    .
\end{equation}
Define $\boldsymbol{J}^{\ast}=\nabla\Psi(\boldsymbol{\theta}^{\lambda})$ and $\boldsymbol{H}_{\rho}(\boldsymbol{\theta})=\nabla^2\rho(\boldsymbol{\theta})$, and let $\boldsymbol{K}^{\ast}$ be the covariance in the score limit
\begin{equation}
\label{eq:score-clt}
    \sqrt{s_n}\{\boldsymbol{U}_n(\boldsymbol{\theta}^{\lambda})-\Psi(\boldsymbol{\theta}^{\lambda})\}
    \to_d
    N(\boldsymbol{0},\boldsymbol{K}^{\ast})
    .
\end{equation}
The penalized curvature and target covariance are
\begin{equation}
\label{eq:target-objects}
    \boldsymbol{J}_{\lambda}^{\ast}
    =
    \boldsymbol{J}^{\ast}+\lambda\boldsymbol{H}_{\rho}(\boldsymbol{\theta}^{\lambda})
    ,\quad
    \boldsymbol{V}_{\mathrm{target}}^{\ast}
    =
    (\boldsymbol{J}_{\lambda}^{\ast})^{-1}\boldsymbol{K}^{\ast}(\boldsymbol{J}_{\lambda}^{\ast})^{-1}
    .
\end{equation}
Thus the penalty, or equivalently the prior in \eqref{eq:eta-prior}, adds $\lambda\boldsymbol{H}_{\rho}(\boldsymbol{\theta}^{\lambda})$ to the curvature, whereas $\boldsymbol{K}^{\ast}$ is determined by the loss alone.
Specifically, for i.i.d.\ summands with $s_n=n$, $\boldsymbol{K}^{\ast}$ is the covariance of the single-observation loss gradient evaluated at $\boldsymbol{\theta}^{\lambda}$, and for dependent observations it is determined by the full score limit in \eqref{eq:score-clt}.
We take $\boldsymbol{V}_{\mathrm{target}}^{\ast}$ as the target covariance because it is the first-order sampling covariance of a penalized estimating-equation solution \citep[Theorem~5.21]{vandervaart1998asymptotic}.
Accordingly, coverage concerns the population solution $\boldsymbol{\theta}^{\lambda}$, which includes shrinkage when $\lambda>0$.
For $\lambda=0$, the rate condition in Assumption~\ref{ass:regularity} gives $\sqrt{s_n}\lambda_n\to0$ and inference for the unpenalized target.

We next fix the basic notation.
All stochastic limits are taken as $n\to\infty$ under the true law $P^{\ast}$, and $\to_p$ and $\to_d$ denote convergence in probability and in distribution, respectively.
When simulation variables occur, these limits are taken under the joint law of the data and simulation.
We use $\|\cdot\|$ for the Euclidean vector norm and spectral matrix norm, let $\boldsymbol{I}_d$ denote the $d\times d$ identity matrix, and write $\boldsymbol{A}\succ0$ when a symmetric matrix $\boldsymbol{A}$ is positive definite.
For deterministic $a_n>0$, $X_n=O_p(a_n)$ means that $\|X_n\|/a_n$ is bounded in probability, and $X_n=o_p(a_n)$ means that $\|X_n\|/a_n\to_p0$.
For a probability law $P$, $\EE_P[X]$ and $\var_P(X)$ denote expectation and covariance matrix under $P$, with variance in the scalar case.
Without a subscript, the law is $P^{\ast}$, and moments with subscript $\Pi_n^\eta$ integrate over $\boldsymbol{\theta}$ conditionally on $\cD_n$, even when this conditioning is suppressed.
We write $N(\boldsymbol{\mu},\boldsymbol{\Sigma})$ for the normal distribution with mean $\boldsymbol{\mu}$ and covariance matrix $\boldsymbol{\Sigma}$, $X\sim P$ when $X$ has law $P$, and $\operatorname{Law}(X\mid\cF)$ for the conditional distribution of $X$ given a sigma-field $\cF$.
Write $B_\delta=\{\boldsymbol{\theta}:\|\boldsymbol{\theta}-\boldsymbol{\theta}^{\lambda}\|<\delta\}$, and let $\overline{B}$ and $\1\{E\}$ denote the closure of a set $B$ and the indicator of an event $E$.
For probability measures $P,Q$, let $\|P-Q\|_{\TV}=\sup_B|P(B)-Q(B)|$, where $B$ ranges over measurable sets.
Conditional total-variation convergence means convergence in probability of this random distance.
All suprema and estimator selections are assumed measurable.
Quantities undefined on events whose probabilities tend to zero are assigned fixed measurable values there.

To state our theoretical results rigorously, we impose the following regularity conditions.
The first concerns the loss and the penalty near the target.
\begin{assumption}[Local loss and penalty regularity]
\label{ass:regularity}
    There is a bounded open convex neighborhood $\cN$ of $\boldsymbol{\theta}^{\lambda}$ with closure in $\Theta$ such that:
    \begin{enumerate}[label=(\roman*),leftmargin=*]
        \item\label{ass:regularity:population}
        The penalty $\rho$ is $C^2$ on $\cN$ and $\sqrt{s_n}(\lambda_n-\lambda)\to0$.
        The map $\Psi$ is $C^1$ on $\cN$, its derivative $\boldsymbol{J}=\nabla\Psi$ is symmetric, and \eqref{eq:pop-penalized} has the unique solution $\boldsymbol{\theta}^{\lambda}$ there.

        \item\label{ass:regularity:empirical}
        With probability tending to one, $\bar{M}_n$ is $C^2$ on $\cN$, and, for every compact $\cK\subset\cN$,
        \begin{equation*}
            \sup_{\boldsymbol{\theta}\in\cK}\|\boldsymbol{U}_n(\boldsymbol{\theta})-\Psi(\boldsymbol{\theta})\|
            \to_p
            0
            ,\quad
            \sup_{\boldsymbol{\theta}\in\cK}\|\boldsymbol{J}_n(\boldsymbol{\theta})-\boldsymbol{J}(\boldsymbol{\theta})\|
            \to_p
            0
            .
        \end{equation*}

        \item\label{ass:regularity:clt}
        The score limit \eqref{eq:score-clt} holds, and $\boldsymbol{J}_{\lambda}^{\ast}$ and $\boldsymbol{K}^{\ast}$ are positive definite.
    \end{enumerate}
\end{assumption}
The derivative conditions in \ref{ass:regularity:population} and \ref{ass:regularity:empirical} provide the local quadratic approximation of the composite loss, and the score limit in \ref{ass:regularity:clt} determines the sampling variability of estimators based on it.

The second condition concerns the prior factor and the posterior mass away from the target.
\begin{assumption}[Prior and posterior concentration]
\label{ass:posterior}
    \leavevmode
    \begin{enumerate}[label=(\roman*),leftmargin=*]
        \item\label{ass:posterior:proper}
        Almost surely, $0<Z_n(\eta)<\infty$ for every $\eta\in\cE$, and $G_n$ and $r_n$ are finite near $\boldsymbol{\theta}^{\lambda}$ with probability tending to one.

        \item\label{ass:posterior:local}
        For some $\delta_0>0$ with $\overline{B}_{\delta_0}\subset\cN$, the function $r_n$ is Lipschitz on $B_{\delta_0}$ with constant $L_n=o(\sqrt{s_n})$.

        \item\label{ass:posterior:separation}
        For every $\varepsilon>0$, there is $a_\varepsilon>0$ such that
        \begin{equation*}
            \begin{split}
                I_{n,\varepsilon}(\eta)
                &=
                \int_{\Theta\setminus B_\varepsilon}
                \exp[-\eta s_n\{G_n(\boldsymbol{\theta})-G_n(\boldsymbol{\theta}^{\lambda})\}
                +
                r_n(\boldsymbol{\theta})
                -
                r_n(\boldsymbol{\theta}^{\lambda})]\dd\boldsymbol{\theta}
                ,\\
                \sup_{\eta\in\cE} I_{n,\varepsilon}(\eta)
                &=
                O_p\{\exp(-a_\varepsilon s_n)\}
                .
            \end{split}
        \end{equation*}
    \end{enumerate}
\end{assumption}
Condition~\ref{ass:posterior:local} holds, for example, for fixed positive continuously differentiable prior densities with bounded log-gradient near the target.
Condition~\ref{ass:posterior:separation} controls the posterior mass outside a neighborhood of the target, and sufficient conditions are given in Lemmas~\ref{lem:separation-verification} and~\ref{hb:convex-verification}.

The third condition requires an estimator that approximately solves the penalized estimating equation.
\begin{assumption}[Baseline estimator]
\label{ass:baseline}
    There is a measurable $\check{\boldsymbol{\theta}}_n\in\Theta$ such that $\check{\boldsymbol{\theta}}_n\to_p\boldsymbol{\theta}^{\lambda}$ and
    \begin{equation*}
        \boldsymbol{F}_n(\check{\boldsymbol{\theta}}_n)
        =
        o_p(s_n^{-1/2})
        ,\quad
        \boldsymbol{F}_n
        =
        \boldsymbol{U}_n+\lambda_n\nabla\rho
        .
    \end{equation*}
\end{assumption}
We call $\check{\boldsymbol{\theta}}_n$ the baseline estimator and center the calibration at estimators that are asymptotically equivalent to it, in the following sense.
\begin{definition}[Admissible center]
\label{def:admissible}
    Under Assumption~\ref{ass:baseline}, a measurable center $\tilde{\boldsymbol{\theta}}_n$ is admissible if
    \begin{equation*}
        \|\tilde{\boldsymbol{\theta}}_n-\check{\boldsymbol{\theta}}_n\|
        =
        o_p(s_n^{-1/2})
        .
    \end{equation*}
    The center may depend on simulation as well as the data, in which case the condition is interpreted under their joint law.
\end{definition}
Admissible centers share the first-order sampling distribution of the baseline estimator.
A consistent interior maximum a posteriori (MAP) estimator is admissible under Assumptions~\ref{ass:regularity}, \ref{ass:posterior}~\ref{ass:posterior:local}, and~\ref{ass:baseline}, as is a one-step Newton update from a $\sqrt{s_n}$-consistent preliminary estimator.
The posterior mean is admissible under the additional moment condition in Assumption~\ref{ass:UI}.
Section~\ref{supp:ex-admissible} gives the details, including conditions for a global MAP estimator.

\section{Main results}
\label{sec:main}

This section shows that an affine transformation of the posterior draws reproduces the first-order sampling distribution of an admissible center, uniformly over the learning rates in $\cE$.
Proofs are given in Section~\ref{supp:pf-main}.

We first compare the sampling distribution of an admissible center with the posterior distribution around it, because their covariance difference determines the calibration.
Using Assumptions~\ref{ass:regularity} and~\ref{ass:baseline}, we obtain the following expansion, which gives the sampling distribution.
\begin{proposition}[Baseline sampling law]
\label{prop:baseline-clt}
    Under Assumptions~\ref{ass:regularity} and~\ref{ass:baseline},
    \begin{equation*}
        \sqrt{s_n}(\check{\boldsymbol{\theta}}_n-\boldsymbol{\theta}^{\lambda})
        =
        -(\boldsymbol{J}_{\lambda}^{\ast})^{-1}\sqrt{s_n}\boldsymbol{F}_n(\boldsymbol{\theta}^{\lambda})
        +
        o_p(1)
        \to_d
        N(\boldsymbol{0},\boldsymbol{V}_{\mathrm{target}}^{\ast})
        .
    \end{equation*}
    Any two baseline estimators differ by $o_p(s_n^{-1/2})$, and every admissible center has the same expansion.
\end{proposition}

We next describe the posterior distribution.
For a data-measurable center $\boldsymbol{a}_n$, write $Q_n^{\boldsymbol{a}_n,\eta}$ for the conditional law of $\sqrt{s_n}(\boldsymbol{\theta}-\boldsymbol{a}_n)$ under $\Pi_n^\eta$.
Then we obtain the following normal approximation.
\begin{proposition}[Total-variation normal approximation]
\label{prop:bvm-baseline}
    Under Assumptions~\ref{ass:regularity}--\ref{ass:baseline}, every data-measurable admissible center $\tilde{\boldsymbol{\theta}}_n$ common to all $\eta\in\cE$ satisfies
    \begin{equation*}
        \sup_{\eta\in\cE}
        \bigl\|
            Q_n^{\tilde{\boldsymbol{\theta}}_n,\eta}
            -
            N\{\boldsymbol{0},(\eta\boldsymbol{J}_{\lambda}^{\ast})^{-1}\}
        \bigr\|_{\TV}
        \to_p
        0
        .
    \end{equation*}
\end{proposition}
This Bernstein--von Mises approximation combines a local quadratic expansion with global concentration, as in \citet{miller2021asymptotic}.
Its covariance depends on $\eta$, whereas the sampling covariance in Proposition~\ref{prop:baseline-clt} does not.

The following transformation removes that dependence.
For a symmetric positive definite matrix $\boldsymbol{A}$, use its symmetric positive definite square root $\boldsymbol{A}^{1/2}$ and write $\boldsymbol{A}^{-1/2}=(\boldsymbol{A}^{1/2})^{-1}$.
Choose two data-measurable admissible centers, an input center $\boldsymbol{a}_n$ and an output center $\tilde{\boldsymbol{\theta}}_n$, both common to all $\eta\in\cE$, and define
\begin{equation}
\label{eq:def-lsc-map}
    \begin{gathered}
        \boldsymbol{H}_0
        =
        \eta\boldsymbol{J}_{\lambda}^{\ast}
        ,\quad
        \boldsymbol{\Omega}
        =
        (\boldsymbol{V}_{\mathrm{target}}^{\ast})^{1/2}\boldsymbol{H}_0^{1/2}
        ,\\
        \boldsymbol{\theta}_{\mathrm{calib}}^{(\eta)}
        =
        \tilde{\boldsymbol{\theta}}_n+\boldsymbol{\Omega}(\boldsymbol{\theta}-\boldsymbol{a}_n)
        ,\quad
        \boldsymbol{\theta}\sim\Pi_n^\eta(\cdot\mid\cD_n)
        .
    \end{gathered}
\end{equation}
The dependence of $\boldsymbol{H}_0$ and $\boldsymbol{\Omega}$ on $\eta$ is suppressed in the notation.
The input center specifies the origin for rescaling, and the output center specifies the location of the calibrated distribution.
They may coincide, so no posterior-moment assumption is needed for this construction.
We call the conditional law of $\boldsymbol{\theta}_{\mathrm{calib}}^{(\eta)}$ given $\cD_n$ the calibrated posterior.
Using Proposition~\ref{prop:bvm-baseline}, we obtain the following theorem, which is our main result.
\begin{theorem}[Location--scale calibration]
\label{thm:ls-calib}
    Under Assumptions~\ref{ass:regularity}, \ref{ass:posterior}, and~\ref{ass:baseline}, the calibrated conditional laws in \eqref{eq:def-lsc-map} satisfy
    \begin{equation}
    \label{eq:calib-tv}
        \sup_{\eta\in\cE}
        \Bigl\|
            \operatorname{Law}\{\sqrt{s_n}(\boldsymbol{\theta}_{\mathrm{calib}}^{(\eta)}-\tilde{\boldsymbol{\theta}}_n)\mid\cD_n\}
            -
            N(\boldsymbol{0},\boldsymbol{V}_{\mathrm{target}}^{\ast})
        \Bigr\|_{\TV}
        \to_p
        0
        .
    \end{equation}
\end{theorem}
The covariance identity $\boldsymbol{\Omega}\boldsymbol{H}_0^{-1}\boldsymbol{\Omega}^{\top}=\boldsymbol{V}_{\mathrm{target}}^{\ast}$ gives the result by invariance of total variation under nonsingular affine maps.
Thus, the calibrated posterior has the same first-order uncertainty as the admissible estimator.
Theorem~\ref{thm:ls-calib} also bounds the total-variation distance between the calibrated posteriors at any two learning rates in $\cE$ by twice the supremum in \eqref{eq:calib-tv}.
This distance therefore tends to zero in probability, uniformly over the two rates.
Moreover, the supremum bounds the distance to the normal limit at any measurable data-selected rate $\eta_n=\eta_n(\cD_n)\in\cE$, so the learning rate may be selected from the data within the fixed interval $\cE$.

We next turn to the frequentist coverage of calibrated credible regions.
Let $\hat{\boldsymbol{V}}_n$ be a symmetric data-measurable estimator of $\boldsymbol{V}_{\mathrm{target}}^{\ast}$, such as the one given in Section~\ref{sec:plugin}, fix $0<\alpha<1$, and define
\begin{equation}
\label{eq:ellipsoid}
    C_n(\alpha)
    =
    \Bigl\{
        \boldsymbol{\vartheta}:
        s_n(\boldsymbol{\vartheta}-\tilde{\boldsymbol{\theta}}_n)^{\top}\hat{\boldsymbol{V}}_n^{-1}
        (\boldsymbol{\vartheta}-\tilde{\boldsymbol{\theta}}_n)
        \le
        \chi^2_{p,1-\alpha}
    \Bigr\}
    .
\end{equation}
Here $\chi^2_{p,1-\alpha}$ is the $(1-\alpha)$-quantile of a chi-square distribution with $p$ degrees of freedom.
Then we obtain the following result.
\begin{corollary}[Calibrated regions]
\label{cor:coverage}
    Under Assumptions~\ref{ass:regularity}, \ref{ass:posterior}, and~\ref{ass:baseline}, if $\hat{\boldsymbol{V}}_n\to_p\boldsymbol{V}_{\mathrm{target}}^{\ast}$, then $C_n(\alpha)$ has calibrated posterior probability tending to $1-\alpha$ in probability uniformly over $\eta\in\cE$, and $P^{\ast}\{\boldsymbol{\theta}^{\lambda}\in C_n(\alpha)\}\to1-\alpha$.
\end{corollary}
Corollary~\ref{cor:fixed-shape} in Section~\ref{supp:pf-main} extends this result to regions of a general fixed shape.
The same conclusion holds for smooth functionals.
Let $g:\Theta\to\RR^k$ be $C^1$ near $\boldsymbol{\theta}^{\lambda}$ with full-row-rank Jacobian $\boldsymbol{\Gamma}^{\ast}=\nabla g(\boldsymbol{\theta}^{\lambda})$, extended measurably to $\RR^p$ if necessary, and define
\begin{equation*}
    \boldsymbol{W}^{\ast}
    =
    \boldsymbol{\Gamma}^{\ast}\boldsymbol{V}_{\mathrm{target}}^{\ast}(\boldsymbol{\Gamma}^{\ast})^{\top}
    ,\quad
    \hat{\boldsymbol{W}}_n
    =
    \nabla g(\tilde{\boldsymbol{\theta}}_n)\hat{\boldsymbol{V}}_n\nabla g(\tilde{\boldsymbol{\theta}}_n)^{\top}
    .
\end{equation*}
The functional ellipsoid is obtained from \eqref{eq:ellipsoid} by replacing the parameter, center, covariance, and dimension by $\boldsymbol{w}\in\RR^k$, $g(\tilde{\boldsymbol{\theta}}_n)$, $\hat{\boldsymbol{W}}_n$, and $k$, respectively.
We measure conditional weak convergence by the bounded-Lipschitz distance
\begin{equation*}
    d_{\mathrm{BL}}(P,Q)
    =
    \sup_{\|f\|_\infty\le1,\operatorname{Lip}(f)\le1}
    \Biggl|
        \int f\dd(P-Q)
    \Biggr|
    ,
\end{equation*}
where $\|f\|_\infty=\sup_{\boldsymbol{x}}|f(\boldsymbol{x})|$ and $\operatorname{Lip}(f)$ is the Lipschitz constant of $f$.
Uniform convergence means that the supremum of this distance over $\eta\in\cE$ tends to zero in probability.
With this notation, we obtain the following result.
\begin{corollary}[Smooth functionals]
\label{cor:functionals}
    Under Assumptions~\ref{ass:regularity}, \ref{ass:posterior}, and~\ref{ass:baseline}, the conditional laws of $\sqrt{s_n}\{g(\boldsymbol{\theta}_{\mathrm{calib}}^{(\eta)})-g(\tilde{\boldsymbol{\theta}}_n)\}$ converge weakly in probability, uniformly over $\eta\in\cE$, to $N(\boldsymbol{0},\boldsymbol{W}^{\ast})$.
    If $\hat{\boldsymbol{V}}_n\to_p\boldsymbol{V}_{\mathrm{target}}^{\ast}$, the functional ellipsoid has limiting calibrated posterior probability and frequentist coverage $1-\alpha$ for $g(\boldsymbol{\theta}^{\lambda})$.
    For $k=1$, equal-tailed calibrated intervals also have limiting coverage $1-\alpha$, including at any measurable data-selected rate in $\cE$.
\end{corollary}
Corollaries~\ref{cor:coverage} and~\ref{cor:functionals} guarantee that credible regions of the calibrated posterior are asymptotically valid confidence regions for $\boldsymbol{\theta}^{\lambda}$ and $g(\boldsymbol{\theta}^{\lambda})$ at every learning rate in $\cE$.
In particular, the quantiles of $g(\boldsymbol{\theta}_{\mathrm{calib}}^{(\eta)})$ give a confidence interval for a scalar functional without computing the derivative of $g$.

\section{Practical plug-in calibration}
\label{sec:plugin}

This section gives a feasible version of the calibration, in which the unknown quantities in \eqref{eq:def-lsc-map} are estimated from the posterior draws and the data.
Proofs appear in Section~\ref{supp:pf-plugin}.

Implementation requires posterior draws, loss curvature, and a score covariance estimator appropriate to the sampling design, which we introduce in turn.
Let $\eta_n=\eta_n(\cD_n)\in\cE$ be a measurable choice, and write $\eta=\eta_n$ throughout this section.
All posterior moments and draws refer to this selected posterior.
For posterior draws $\boldsymbol{\theta}^{(1)},\ldots,\boldsymbol{\theta}^{(D)}$, with $D=D_n\to\infty$ chosen before simulation, define
\begin{equation*}
    \hat{\boldsymbol{\theta}}_{\mathrm{GB}}=D^{-1}\sum_{d=1}^D\boldsymbol{\theta}^{(d)}
    ,\quad
    \hat{\boldsymbol{\Sigma}}_{\mathrm{post}}
    =
    D^{-1}\sum_{d=1}^D(\boldsymbol{\theta}^{(d)}-\hat{\boldsymbol{\theta}}_{\mathrm{GB}})
    (\boldsymbol{\theta}^{(d)}-\hat{\boldsymbol{\theta}}_{\mathrm{GB}})^{\top}
    .
\end{equation*}
Let $\boldsymbol{\theta}_{\mathrm{GB}}=\EE_{\Pi_n^\eta}[\boldsymbol{\theta}\mid\cD_n]$ and $\boldsymbol{\Sigma}_{\mathrm{post},n}=\var_{\Pi_n^\eta}(\boldsymbol{\theta}\mid\cD_n)$ denote the exact posterior mean and covariance.

To use these Monte Carlo summaries, we impose the following two conditions, which separate posterior tail control from simulation accuracy.
\begin{assumption}[Posterior second moments]
\label{ass:UI}
    The selected posterior has finite second moments almost surely and, writing $\boldsymbol{X}_n=\sqrt{s_n}(\boldsymbol{\theta}-\check{\boldsymbol{\theta}}_n)$, satisfies, for every $\varepsilon>0$,
    \begin{equation*}
        \lim_{M\to\infty}
        \limsup_{n\to\infty}
        P^{\ast}\bigl\{
            \EE_{\Pi_n^\eta}\bigl[
                \|\boldsymbol{X}_n\|^2\1\{\|\boldsymbol{X}_n\|>M\}
                \mid
                \cD_n
            \bigr]>\varepsilon
        \bigr\}
        =
        0
        .
    \end{equation*}
\end{assumption}
Assumption~\ref{ass:UI} is a uniform integrability condition, under which the total-variation approximation in Proposition~\ref{prop:bvm-baseline} yields convergence of the first and second moments.
It holds under Assumptions~\ref{ass:regularity}--\ref{ass:baseline} on bounded parameter spaces, and Lemmas~\ref{lem:weighted-tail} and~\ref{hb:convex-verification} give sufficient conditions for unbounded parameter spaces.

The second condition requires the Monte Carlo errors of the sample mean and covariance to be asymptotically negligible.
\begin{assumption}[Monte Carlo accuracy]
\label{ass:MC}
    The simulation satisfies
    \begin{equation*}
        \sqrt{s_n}\|\hat{\boldsymbol{\theta}}_{\mathrm{GB}}-\boldsymbol{\theta}_{\mathrm{GB}}\|
        \to_p
        0
        ,\quad
        s_n\|\hat{\boldsymbol{\Sigma}}_{\mathrm{post}}-\boldsymbol{\Sigma}_{\mathrm{post},n}\|
        \to_p
        0
        .
    \end{equation*}
\end{assumption}
Lemma~\ref{lem:mc-sufficient} in Section~\ref{supp:pf-plugin} gives a sufficient condition, which requires the effective number of draws to diverge.
Under these two assumptions, we obtain the following result.
\begin{lemma}[Working covariance from posterior draws]
\label{lem:postvar-consistency}
    Under Assumptions~\ref{ass:regularity}, \ref{ass:posterior}, \ref{ass:baseline}, and~\ref{ass:UI}, the posterior mean is admissible and $\|s_n\boldsymbol{\Sigma}_{\mathrm{post},n}-\boldsymbol{H}_0^{-1}\|\to_p0$.
    If Assumption~\ref{ass:MC} also holds, the Monte Carlo mean is admissible under the joint law, and $\|s_n\hat{\boldsymbol{\Sigma}}_{\mathrm{post}}-\boldsymbol{H}_0^{-1}\|\to_p0$.
\end{lemma}

We next estimate the target covariance.
At any consistent evaluation point $\bar{\boldsymbol{a}}_n\to_p\boldsymbol{\theta}^{\lambda}$ and for a symmetric estimator $\hat{\boldsymbol{K}}$ of $\boldsymbol{K}^{\ast}$, define
\begin{equation}
\label{eq:target-plugin}
    \hat{\boldsymbol{J}}_{\lambda}
    =
    \boldsymbol{J}_n(\bar{\boldsymbol{a}}_n)+\lambda_n\boldsymbol{H}_{\rho}(\bar{\boldsymbol{a}}_n)
    ,\quad
    \hat{\boldsymbol{V}}_{\mathrm{target}}
    =
    \hat{\boldsymbol{J}}_{\lambda}^{-1}\hat{\boldsymbol{K}}\hat{\boldsymbol{J}}_{\lambda}^{-1}
    .
\end{equation}
The estimator $\hat{\boldsymbol{K}}$ must suit the sampling design.
For i.i.d.\ summands, suppose $s_n=N_n\in\NN$ and let $\boldsymbol{\varphi}_i$ be i.i.d.\ random maps such that
\begin{equation*}
    \boldsymbol{U}_n(\boldsymbol{\theta})
    =
    N_n^{-1}\sum_{i=1}^{N_n}\boldsymbol{\varphi}_i(\boldsymbol{\theta})
    .
\end{equation*}
At an evaluation point $\bar{\boldsymbol{a}}_n$, define the centered score covariance by
\begin{equation}
\label{eq:K-centered}
    \hat{\boldsymbol{K}}
    =
    N_n^{-1}\sum_i
    (\boldsymbol{\varphi}_i(\bar{\boldsymbol{a}}_n)-\bar{\boldsymbol{\varphi}})
    (\boldsymbol{\varphi}_i(\bar{\boldsymbol{a}}_n)-\bar{\boldsymbol{\varphi}})^{\top}
    ,\quad
    \bar{\boldsymbol{\varphi}}
    =
    N_n^{-1}\sum_i\boldsymbol{\varphi}_i(\bar{\boldsymbol{a}}_n)
    .
\end{equation}
The sample mean $\bar{\boldsymbol{\varphi}}$ is subtracted because the summands have mean $\Psi(\boldsymbol{\theta}^{\lambda})=-\lambda\nabla\rho(\boldsymbol{\theta}^{\lambda})$ at the target, which is nonzero in general.
The uncentered second moment would instead estimate $\boldsymbol{K}^{\ast}+\Psi(\boldsymbol{\theta}^{\lambda})\Psi(\boldsymbol{\theta}^{\lambda})^{\top}$.
Lemma~\ref{lem:target-plugin} in Section~\ref{supp:pf-plugin} shows that $\hat{\boldsymbol{V}}_{\mathrm{target}}$ is consistent whenever $\hat{\boldsymbol{K}}$ is.
By Lemma~\ref{lem:K-consistency}, the estimator in \eqref{eq:K-centered} is consistent if each $\boldsymbol{\varphi}_i$ is almost surely continuous at $\boldsymbol{\theta}^{\lambda}$ and bounded in norm near $\boldsymbol{\theta}^{\lambda}$ by a square-integrable random variable.
For independent clusters, the corresponding estimator is given in \eqref{hb:eq:cluster-variance}.

We now define the feasible calibration.
For an admissible output center $\bar{\boldsymbol{\theta}}_n$, calibrate each draw by
\begin{equation}
\label{eq:empirical-lsc}
    \hat{\boldsymbol{H}}_0^{-1}
    =
    s_n\hat{\boldsymbol{\Sigma}}_{\mathrm{post}}
    ,\quad
    \hat{\boldsymbol{\Omega}}
    =
    \hat{\boldsymbol{V}}_{\mathrm{target}}^{1/2}\hat{\boldsymbol{H}}_0^{1/2}
    ,\quad
    \hat{\boldsymbol{\theta}}_{\mathrm{calib}}^{(d)}
    =
    \bar{\boldsymbol{\theta}}_n
    +
    \hat{\boldsymbol{\Omega}}(\boldsymbol{\theta}^{(d)}
    -
    \hat{\boldsymbol{\theta}}_{\mathrm{GB}})
    .
\end{equation}
Consistency makes the estimated matrices positive definite with probability tending to one, and we take $\hat{\boldsymbol{\Omega}}$ to be the identity otherwise.
The evaluation point $\bar{\boldsymbol{a}}_n$ needs only consistency, whereas the output center requires admissibility for coverage.

To state the result, we distinguish the draws used for estimation from a new draw.
Call the draws used to estimate the transformation pilot draws, and let $\cW_n$ be the sigma-field generated by the data, these draws, and any auxiliary randomness used to compute the estimates.
All estimated matrices and centers are $\cW_n$-measurable.
For a new draw $\boldsymbol{\theta}\sim\Pi_n^\eta(\cdot\mid\cD_n)$ conditionally independent of the pilot draws and auxiliary randomness given the data, set $\hat{\boldsymbol{\theta}}_{\mathrm{calib}}=\bar{\boldsymbol{\theta}}_n+\hat{\boldsymbol{\Omega}}(\boldsymbol{\theta}-\hat{\boldsymbol{\theta}}_{\mathrm{GB}})$.
Then we obtain the following feasible counterpart of Theorem~\ref{thm:ls-calib}.
\begin{proposition}[Estimated location--scale calibration]
\label{prop:estimated-calib}
    Suppose Assumptions~\ref{ass:regularity}, \ref{ass:posterior}, \ref{ass:baseline}, \ref{ass:UI}, and~\ref{ass:MC} hold, $\bar{\boldsymbol{a}}_n\to_p\boldsymbol{\theta}^{\lambda}$, and the symmetric score covariance estimator satisfies $\hat{\boldsymbol{K}}\to_p\boldsymbol{K}^{\ast}$.
    Then
    \begin{equation*}
        \Bigl\|
            \operatorname{Law}\{\sqrt{s_n}(\hat{\boldsymbol{\theta}}_{\mathrm{calib}}-\bar{\boldsymbol{\theta}}_n)\mid\cW_n\}
            -
            N(\boldsymbol{0},\boldsymbol{V}_{\mathrm{target}}^{\ast})
        \Bigr\|_{\TV}
        \to_p
        0
        .
    \end{equation*}
\end{proposition}
Consequently, with an admissible output center, the regions and functional intervals in Corollaries~\ref{cor:coverage} and~\ref{cor:functionals}, formed with $\hat{\boldsymbol{V}}_{\mathrm{target}}$, have limiting calibrated posterior probability and frequentist coverage $1-\alpha$ under this conditional law.
In practice, the pilot draws themselves are transformed.
Lemma~\ref{lem:empirical-output} and Corollary~\ref{cor:empirical-quantiles} in Section~\ref{supp:pf-plugin} show that the resulting empirical distribution has the same normal limit and yields equal-tailed intervals with nominal limiting coverage, under sampler conditions stated there.

\begin{remark}[Relation to OFS and matrix choices]
\label{rem:ofs}
    The OFS matrix of \citet[Section~3, equation~(4)]{shaby2014open-faced}, $\boldsymbol{Q}^{-1}\boldsymbol{P}^{1/2}\boldsymbol{Q}^{1/2}$ in that paper's notation, becomes
    \begin{equation*}
        \boldsymbol{\Omega}_{\mathrm{OFS}}
        =
        \sqrt{\eta}\,(\boldsymbol{J}_{\lambda}^{\ast})^{-1}(\boldsymbol{K}^{\ast})^{1/2}(\boldsymbol{J}_{\lambda}^{\ast})^{1/2}
    \end{equation*}
    here.
    It satisfies the covariance identity $\boldsymbol{\Omega}\boldsymbol{H}_0^{-1}\boldsymbol{\Omega}^{\top}=\boldsymbol{V}_{\mathrm{target}}^{\ast}$, as does the matrix in \eqref{eq:def-lsc-map}, although the two matrices generally differ when their factors do not commute.
    Every nonsingular matrix satisfying this identity has the form $(\boldsymbol{V}_{\mathrm{target}}^{\ast})^{1/2}\boldsymbol{O}\boldsymbol{H}_0^{1/2}$ for an orthogonal $\boldsymbol{O}$.
    Consistently estimating any fixed matrix of this form gives the same first-order normal limit, while different choices may change finite-sample shape.
    The symmetric-root construction in \eqref{eq:def-lsc-map} fixes one such choice.
\end{remark}

\section{Empirical illustration}
\label{hb:application}

We conducted a simulation study of ridge-penalized Huber regression with independent clusters to demonstrate the practical behavior of the location--scale calibration across learning rates.

\subsection{Huber regression for clustered observations}
\label{subsec:ri-lmm}
We used a random-intercept working model to standardize residuals and a Huber loss to reduce the contribution of large residuals.
The sampling units were independent identically distributed groups $i=1,\ldots,G$ of fixed size $m$, with $n=mG$ and $s_n=n$.
For observations $(y_{ij},\boldsymbol{x}_{ij})\in\RR\times\RR^p$, the working model was
\begin{equation*}
    y_{ij}
    =
    \boldsymbol{x}_{ij}^{\top}\boldsymbol{\beta}+b_i+\varepsilon_{ij}
    ,\quad
    b_i
    \sim
    N(0,\tau^2)
    ,\quad
    \varepsilon_{ij}
    \sim
    N(0,\sigma^2)
    .
\end{equation*}
We held the working variance components fixed, with $\sigma^2>0$, because estimating them from the same observations would require accounting for their estimation error.
Writing $\boldsymbol{y}_i=(y_{i1},\ldots,y_{im})^{\top}$, letting $\boldsymbol{X}_i$ have rows $\boldsymbol{x}_{ij}^{\top}$, and letting $\boldsymbol{1}_m$ be the $m$-vector of ones, we set
\begin{equation*}
    \boldsymbol{\Sigma}=\tau^2\boldsymbol{1}_m\boldsymbol{1}_m^{\top}+\sigma^2\boldsymbol{I}_m
    ,\quad
    \tilde{\boldsymbol{y}}_i=(\tilde{y}_{i1},\ldots,\tilde{y}_{im})^{\top}=\boldsymbol{\Sigma}^{-1/2}\boldsymbol{y}_i
    ,\quad
    \tilde{\boldsymbol{X}}_i=\boldsymbol{\Sigma}^{-1/2}\boldsymbol{X}_i
    ,
\end{equation*}
writing $\tilde{\boldsymbol{x}}_{ij}^{\top}$ for the $j$th row of $\tilde{\boldsymbol{X}}_i$.
The generalized posterior was defined through the standardized marginal residuals $\tilde{\boldsymbol{r}}_i(\boldsymbol{\beta})=(\tilde{r}_{i1}(\boldsymbol{\beta}),\ldots,\tilde{r}_{im}(\boldsymbol{\beta}))^{\top}=\tilde{\boldsymbol{y}}_i-\tilde{\boldsymbol{X}}_i\boldsymbol{\beta}$, with no latent random effects sampled.

We next define the loss.
For the Huber loss \citep{huber1964robust}, we used
\begin{equation*}
    \rho_c(u)
    =
    \begin{cases}
        u^2/2,&|u|\le c,\\
        c|u|-c^2/2,&|u|>c,
    \end{cases}
    \quad
    M_n(\boldsymbol{\beta})
    =
    \sum_{i=1}^G\sum_{j=1}^m\rho_c\{\tilde{r}_{ij}(\boldsymbol{\beta})\}
    .
\end{equation*}
With $\psi_c=\rho_c'$ applied coordinatewise to residual vectors, the cluster score and empirical curvature are
\begin{equation*}
    \begin{gathered}
        \boldsymbol{S}_i(\boldsymbol{\beta})
        =
        -\tilde{\boldsymbol{X}}_i^{\top}\psi_c\{\tilde{\boldsymbol{r}}_i(\boldsymbol{\beta})\}
        ,\quad
        \boldsymbol{U}_n(\boldsymbol{\beta})
        =
        n^{-1}\sum_{i=1}^G\boldsymbol{S}_i(\boldsymbol{\beta})
        ,\\
        \boldsymbol{W}_i(\boldsymbol{\beta})
        =
        \diag\{\1\{|\tilde{r}_{ij}(\boldsymbol{\beta})|<c\}\}_{j=1}^m
        ,\quad
        \boldsymbol{J}_n^{\dagger}(\boldsymbol{\beta})
        =
        n^{-1}\sum_{i=1}^G\tilde{\boldsymbol{X}}_i^{\top}\boldsymbol{W}_i(\boldsymbol{\beta})\tilde{\boldsymbol{X}}_i
        .
    \end{gathered}
\end{equation*}
The dagger indicates that this expression is defined even at a residual kink, where the classical Hessian need not exist.
Although the Huber loss is not twice differentiable at the kinks, Section~\ref{hb:theory} establishes the local expansions and consistency needed for inference.
Under the normalization $n=mG$, the population score and score covariance are
\begin{equation*}
    \Psi(\boldsymbol{\beta})
    =
    m^{-1}\EE[\boldsymbol{S}_1(\boldsymbol{\beta})]
    ,\quad
    \boldsymbol{K}^{\ast}
    =
    m^{-1}\var\{\boldsymbol{S}_1(\boldsymbol{\beta}^{\lambda})\}
    .
\end{equation*}
The factor $1/m$ accounts for the distinction between the number of observations used to normalize the loss and the number of independent clusters.

We included a ridge penalty with fixed center $\boldsymbol{\mu}$ and $\boldsymbol{Q}\succ0$, for which the penalty and the population equation \eqref{eq:pop-penalized} are
\begin{equation*}
    \rho(\boldsymbol{\beta})
    =
    \frac{1}{2}(\boldsymbol{\beta}-\boldsymbol{\mu})^{\top}\boldsymbol{Q}(\boldsymbol{\beta}-\boldsymbol{\mu})
    ,\quad
    \boldsymbol{0}
    =
    \Psi(\boldsymbol{\beta}^{\lambda})+\lambda\boldsymbol{Q}(\boldsymbol{\beta}^{\lambda}-\boldsymbol{\mu})
    .
\end{equation*}
With $r_n\equiv0$ and $\lambda_n>0$, the generalized posterior was
\begin{equation}
\label{hb:eq:posterior}
    \Pi_n^{\eta}(\dd\boldsymbol{\beta})
    \propto
    \exp\{-\eta[M_n(\boldsymbol{\beta})+\lambda_n n\rho(\boldsymbol{\beta})]\}\dd\boldsymbol{\beta}
    .
\end{equation}
The Gaussian factor $N\{\boldsymbol{\mu},(\eta\lambda_n n\boldsymbol{Q})^{-1}\}$ therefore represented the ridge contribution to the loss.
Scaling this contribution together with the loss kept the penalized population target fixed across learning rates.

We then applied the plug-in calibration of Section~\ref{sec:plugin}.
We centered the calibrated draws at the penalized minimizer $\check{\boldsymbol{\beta}}_n$, which was also the MAP estimator because $r_n\equiv0$, and evaluated the target matrices there by
\begin{equation}
\label{hb:eq:cluster-variance}
    \begin{gathered}
        \hat{\boldsymbol{J}}_{\lambda}
        =
        \boldsymbol{J}_n^{\dagger}(\check{\boldsymbol{\beta}}_n)+\lambda_n\boldsymbol{Q}
        ,\quad
        \hat{\boldsymbol{V}}_{\mathrm{target}}
        =
        \hat{\boldsymbol{J}}_{\lambda}^{-1}\hat{\boldsymbol{K}}\hat{\boldsymbol{J}}_{\lambda}^{-1}
        ,\\
        \hat{\boldsymbol{K}}
        =
        \frac{1}{m(G-1)}\sum_{i=1}^G(\hat{\boldsymbol{S}}_i-\bar{\boldsymbol{S}})(\hat{\boldsymbol{S}}_i-\bar{\boldsymbol{S}})^{\top}
        ,\quad
        \hat{\boldsymbol{S}}_i
        =
        \boldsymbol{S}_i(\check{\boldsymbol{\beta}}_n)
        ,\quad
        \bar{\boldsymbol{S}}
        =
        G^{-1}\sum_i\hat{\boldsymbol{S}}_i
        .
    \end{gathered}
\end{equation}
Centering of the cluster scores was needed because $\EE[\boldsymbol{S}_1(\boldsymbol{\beta}^{\lambda})]=-m\lambda\boldsymbol{Q}(\boldsymbol{\beta}^{\lambda}-\boldsymbol{\mu})$ generally differed from zero.
The uncentered second moment would instead converge to $\boldsymbol{K}^{\ast}+m\Psi(\boldsymbol{\beta}^{\lambda})\Psi(\boldsymbol{\beta}^{\lambda})^{\top}$.
Using the posterior mean and covariance estimated from draws, we calibrated each draw by
\begin{equation*}
    \hat{\boldsymbol{\beta}}_{\mathrm{calib}}^{(d)}
    =
    \check{\boldsymbol{\beta}}_n
    +
    \hat{\boldsymbol{V}}_{\mathrm{target}}^{1/2}(n\hat{\boldsymbol{\Sigma}}_{\mathrm{post}})^{-1/2}
    (\boldsymbol{\beta}^{(d)}-\hat{\boldsymbol{\beta}}_{\mathrm{GB}})
    .
\end{equation*}
Lemma~\ref{hb:huber-verification} in Section~\ref{hb:theory} verifies the statistical assumptions for this model with finite second design moments and no atoms at the Huber kinks.
The feasible result additionally requires Assumption~\ref{ass:MC}.

\subsection{Experiment}
\label{hb:experiment}
This subsection describes the simulation design and the results.
We generated $200$ independent data sets with $G=100$, $m=5$, $n=500$, and $p=1$.
The covariates were independent $N(0,1)$ variables, and the observations followed
\begin{gather*}
    y_{ij}=2x_{ij}+b_i+e_{ij}+z_{ij}\xi_{ij}
    ,\\
    b_i\sim N(0,2)
    ,\quad
    e_{ij}\sim N(0,1)
    ,\quad
    z_{ij}\sim\Bern(0.1)
    ,\quad
    \xi_{ij}\sim N(0,10^2)
    ,
\end{gather*}
with all generating variables independent except for the shared group intercept.
Here $\Bern(a)$ denotes the Bernoulli distribution with success probability $a$.
We used working covariance $2\boldsymbol{1}_5\boldsymbol{1}_5^{\top}+\boldsymbol{I}_5$, $c=1$, $\rho(\beta)=\beta^2/2$, and $\lambda_n=\lambda=0.5$.
These choices gave the penalized population target $\beta^{\lambda}\approx0.8163$, obtained from the finite Gaussian mixture over the $2^5$ contamination patterns as described in Section~\ref{hb:population-numerics}.
All coverage and bias calculations concerned this target, whereas the data-generating coefficient was $2$.

We compared the penalized Huber estimator with its sandwich Wald interval, the generalized posterior in \eqref{hb:eq:posterior}, and the calibrated posterior.
All intervals had nominal level $95\%$, with equal-tailed intervals for the two posterior procedures.
The uncalibrated point estimate was the posterior mean.
Each data set was analyzed at the twenty rates $\eta_k=10^{-2+4(k-1)/19}$, $k=1,\ldots,20$, so comparisons across rates used common data.
For every data set and rate, two independently seeded Gibbs chains each discarded $1000$ iterations and retained $10000$ draws, giving $D=20000$ draws.
The sampler in Section~\ref{supp-subsec:sample-algo} used shifted exponential and uniform full conditionals, with interval endpoints obtained analytically from the piecewise quadratic loss.
Section~\ref{hb:population-numerics} reports the convergence diagnostics.

Figure~\ref{fig:random-intercept-exp} summarizes the results.
Calibrated coverage ranged from $96.5\%$ to $97.0\%$, compared with $97.0\%$ for the sandwich benchmark.
Mean calibrated width ranged from $0.1931$ to $0.1933$, compared with $0.1933$ for the benchmark.
In contrast, uncalibrated coverage decreased from $100.0\%$ at $\eta=0.01$ to $10.5\%$ at $\eta=100$, while mean width decreased from $2.0688$ to $0.0208$.
The calibrated estimator and benchmark both had mean bias $0.0025$ and error standard deviation $0.0488$.
These equalities followed from the centering identity $D^{-1}\sum_d\hat{\beta}_{\mathrm{calib}}^{(d)}=\check{\beta}_n$.
The uncalibrated posterior mean had bias $-0.0339$ at the smallest rate and approached the penalized minimizer as the rate increased.
Thus the calibrated intervals remained stable across the grid and closely tracked sandwich inference for the fixed penalized target.
At the observed benchmark coverage, the binomial Monte Carlo standard error was $1.2$ percentage points.

\begin{figure}[tbp]
\centering
\pgfplotsset{compat=1.18}
\begin{subfigure}{0.245\linewidth}
\centering
\begin{tikzpicture}
\begin{axis}[
    title={(a)},
    xlabel={$\log_{10}\eta$},
    xmode=log,
    log basis x=10,
    xmin=0.01,
    xmax=100,
    enlarge x limits=0.03,
    xtick={0.01,0.1,1,10,100},
    xticklabels={$-2$, $-1$, $0$, $1$, $2$},
    scaled y ticks=false,
    scale only axis,
    width=0.65\linewidth,
    height=3.0cm,
    mark size=1.0pt,
    font=\footnotesize,
    xlabel style={yshift=1ex},
    title style={yshift=-0.5ex},
    ymin=0,
    ymax=1.05,
    ytick={0,0.25,0.5,0.75,1},
    yticklabel style={/pgf/number format/fixed,/pgf/number format/precision=2},
]
\addplot[gray, dashdotted, mark=none, line width=0.4pt, forget plot]
coordinates
{(0.01,0.95) (100,0.95)}
;
\addplot[black, solid, mark=none]
coordinates
{(0.01,0.97) (0.016237767391887217,0.97) (0.02636650898730358,0.97) (0.04281332398719394,0.97) (0.06951927961775606,0.97) (0.11288378916846889,0.97) (0.18329807108324356,0.97) (0.29763514416313175,0.97) (0.4832930238571752,0.97) (0.7847599703514611,0.97) (1.2742749857031335,0.97) (2.06913808111479,0.97) (3.359818286283781,0.97) (5.455594781168514,0.97) (8.858667904100823,0.97) (14.38449888287663,0.97) (23.357214690901213,0.97) (37.92690190732246,0.97) (61.584821106602604,0.97) (100.0,0.97)}
;
\addplot[black, dashed, mark=square*, mark size=1.0pt, mark options={fill=white}]
coordinates
{(0.01,0.965) (0.016237767391887217,0.965) (0.02636650898730358,0.965) (0.04281332398719394,0.965) (0.06951927961775606,0.965) (0.11288378916846889,0.965) (0.18329807108324356,0.965) (0.29763514416313175,0.965) (0.4832930238571752,0.97) (0.7847599703514611,0.965) (1.2742749857031335,0.965) (2.06913808111479,0.965) (3.359818286283781,0.965) (5.455594781168514,0.965) (8.858667904100823,0.97) (14.38449888287663,0.965) (23.357214690901213,0.97) (37.92690190732246,0.965) (61.584821106602604,0.97) (100.0,0.97)}
;
\addplot[black, dotted, mark=diamond*, mark size=1.0pt]
coordinates
{(0.01,1.0) (0.016237767391887217,1.0) (0.02636650898730358,1.0) (0.04281332398719394,1.0) (0.06951927961775606,1.0) (0.11288378916846889,1.0) (0.18329807108324356,1.0) (0.29763514416313175,1.0) (0.4832930238571752,1.0) (0.7847599703514611,0.995) (1.2742749857031335,0.95) (2.06913808111479,0.885) (3.359818286283781,0.725) (5.455594781168514,0.605) (8.858667904100823,0.47) (14.38449888287663,0.38) (23.357214690901213,0.3) (37.92690190732246,0.2) (61.584821106602604,0.16) (100.0,0.105)}
;
\end{axis}
\end{tikzpicture}
\end{subfigure}%
\hfill
\begin{subfigure}{0.245\linewidth}
\centering
\begin{tikzpicture}
\begin{axis}[
    title={(b)},
    xlabel={$\log_{10}\eta$},
    xmode=log,
    log basis x=10,
    xmin=0.01,
    xmax=100,
    enlarge x limits=0.03,
    xtick={0.01,0.1,1,10,100},
    xticklabels={$-2$, $-1$, $0$, $1$, $2$},
    scaled y ticks=false,
    scale only axis,
    width=0.65\linewidth,
    height=3.0cm,
    mark size=1.0pt,
    font=\footnotesize,
    xlabel style={yshift=1ex},
    title style={yshift=-0.5ex},
    ymin=0,
    enlarge y limits={upper,value=0.04},
    ytick distance=0.5,
    yticklabel style={/pgf/number format/fixed,/pgf/number format/precision=1},
]
\addplot[black, solid, mark=none]
coordinates
{(0.01,0.19331614393643315) (0.016237767391887217,0.19331614393643315) (0.02636650898730358,0.19331614393643315) (0.04281332398719394,0.19331614393643315) (0.06951927961775606,0.19331614393643315) (0.11288378916846889,0.19331614393643315) (0.18329807108324356,0.19331614393643315) (0.29763514416313175,0.19331614393643315) (0.4832930238571752,0.19331614393643315) (0.7847599703514611,0.19331614393643315) (1.2742749857031335,0.19331614393643315) (2.06913808111479,0.19331614393643315) (3.359818286283781,0.19331614393643315) (5.455594781168514,0.19331614393643315) (8.858667904100823,0.19331614393643315) (14.38449888287663,0.19331614393643315) (23.357214690901213,0.19331614393643315) (37.92690190732246,0.19331614393643315) (61.584821106602604,0.19331614393643315) (100.0,0.19331614393643315)}
;
\addplot[black, dashed, mark=square*, mark size=1.0pt, mark options={fill=white}]
coordinates
{(0.01,0.19310675255083773) (0.016237767391887217,0.19311091814928155) (0.02636650898730358,0.19315959867936067) (0.04281332398719394,0.19324076150055472) (0.06951927961775606,0.19327057397640268) (0.11288378916846889,0.19328969435018364) (0.18329807108324356,0.19317614000373962) (0.29763514416313175,0.1932695415415661) (0.4832930238571752,0.19332383031593595) (0.7847599703514611,0.19306820047011372) (1.2742749857031335,0.19333418569435884) (2.06913808111479,0.1933170766826308) (3.359818286283781,0.1932582708612432) (5.455594781168514,0.1932427280378164) (8.858667904100823,0.1931860827903764) (14.38449888287663,0.19331095967647471) (23.357214690901213,0.19326337511128472) (37.92690190732246,0.19322336206195959) (61.584821106602604,0.19332879829150568) (100.0,0.19329282673954473)}
;
\addplot[black, dotted, mark=diamond*, mark size=1.0pt]
coordinates
{(0.01,2.068813050551103) (0.016237767391887217,1.6260394447663418) (0.02636650898730358,1.2765895062545416) (0.04281332398719394,1.001904850120538) (0.06951927961775606,0.786885502959725) (0.11288378916846889,0.6178115969952295) (0.18329807108324356,0.4844986116487134) (0.29763514416313175,0.38051341248460213) (0.4832930238571752,0.2988618957137781) (0.7847599703514611,0.23406607064367135) (1.2742749857031335,0.18397605994995106) (2.06913808111479,0.14428999591533143) (3.359818286283781,0.11323057170661446) (5.455594781168514,0.08889660424553318) (8.858667904100823,0.06970779488910298) (14.38449888287663,0.05472973127596521) (23.357214690901213,0.04296168101215775) (37.92690190732246,0.033709760694856504) (61.584821106602604,0.0264781313752965) (100.0,0.02076771876595599)}
;
\end{axis}
\end{tikzpicture}
\end{subfigure}%
\hfill
\begin{subfigure}{0.245\linewidth}
\centering
\begin{tikzpicture}
\begin{axis}[
    title={(c)},
    xlabel={$\log_{10}\eta$},
    xmode=log,
    log basis x=10,
    xmin=0.01,
    xmax=100,
    enlarge x limits=0.03,
    xtick={0.01,0.1,1,10,100},
    xticklabels={$-2$, $-1$, $0$, $1$, $2$},
    scaled y ticks=false,
    scale only axis,
    width=0.65\linewidth,
    height=3.0cm,
    mark size=1.0pt,
    font=\footnotesize,
    xlabel style={yshift=1ex},
    title style={yshift=-0.5ex},
    enlarge y limits=0.08,
    ytick distance=0.01,
    yticklabel style={/pgf/number format/fixed,/pgf/number format/precision=2},
]
\addplot[gray!60, solid, mark=none, line width=0.3pt, forget plot]
coordinates
{(0.01,0) (100,0)}
;
\addplot[black, solid, mark=none]
coordinates
{(0.01,0.0025387683360829274) (0.016237767391887217,0.0025387683360829274) (0.02636650898730358,0.0025387683360829274) (0.04281332398719394,0.0025387683360829274) (0.06951927961775606,0.0025387683360829274) (0.11288378916846889,0.0025387683360829274) (0.18329807108324356,0.0025387683360829274) (0.29763514416313175,0.0025387683360829274) (0.4832930238571752,0.0025387683360829274) (0.7847599703514611,0.0025387683360829274) (1.2742749857031335,0.0025387683360829274) (2.06913808111479,0.0025387683360829274) (3.359818286283781,0.0025387683360829274) (5.455594781168514,0.0025387683360829274) (8.858667904100823,0.0025387683360829274) (14.38449888287663,0.0025387683360829274) (23.357214690901213,0.0025387683360829274) (37.92690190732246,0.0025387683360829274) (61.584821106602604,0.0025387683360829274) (100.0,0.0025387683360829274)}
;
\addplot[black, dashed, mark=square*, mark size=1.0pt, mark options={fill=white}]
coordinates
{(0.01,0.0025387683360829226) (0.016237767391887217,0.0025387683360829296) (0.02636650898730358,0.0025387683360829244) (0.04281332398719394,0.0025387683360829382) (0.06951927961775606,0.0025387683360829257) (0.11288378916846889,0.002538768336082924) (0.18329807108324356,0.002538768336082932) (0.29763514416313175,0.0025387683360829283) (0.4832930238571752,0.0025387683360829265) (0.7847599703514611,0.0025387683360829343) (1.2742749857031335,0.002538768336082931) (2.06913808111479,0.0025387683360829296) (3.359818286283781,0.0025387683360829204) (5.455594781168514,0.002538768336082942) (8.858667904100823,0.002538768336082908) (14.38449888287663,0.0025387683360829265) (23.357214690901213,0.002538768336082886) (37.92690190732246,0.0025387683360829512) (61.584821106602604,0.0025387683360828615) (100.0,0.0025387683360829677)}
;
\addplot[black, dotted, mark=diamond*, mark size=1.0pt]
coordinates
{(0.01,-0.033894855145227776) (0.016237767391887217,-0.020416644326044127) (0.02636650898730358,-0.012033862678607043) (0.04281332398719394,-0.006483842664538637) (0.06951927961775606,-0.00303629162573234) (0.11288378916846889,-0.0010485377536195734) (0.18329807108324356,0.000507718224124557) (0.29763514416313175,0.0012024153402614296) (0.4832930238571752,0.0016755178704159612) (0.7847599703514611,0.0021317680763389478) (1.2742749857031335,0.002247408967840357) (2.06913808111479,0.0023250344341732444) (3.359818286283781,0.0024288332208536687) (5.455594781168514,0.0024665834046469063) (8.858667904100823,0.0025015796673632396) (14.38449888287663,0.0025213151387127) (23.357214690901213,0.0025229111032133607) (37.92690190732246,0.0025281811326712342) (61.584821106602604,0.0025285121458186654) (100.0,0.002533334203796942)}
;
\end{axis}
\end{tikzpicture}
\end{subfigure}%
\hfill
\begin{subfigure}{0.245\linewidth}
\centering
\begin{tikzpicture}
\begin{axis}[
    title={(d)},
    xlabel={$\log_{10}\eta$},
    xmode=log,
    log basis x=10,
    xmin=0.01,
    xmax=100,
    enlarge x limits=0.03,
    xtick={0.01,0.1,1,10,100},
    xticklabels={$-2$, $-1$, $0$, $1$, $2$},
    scaled y ticks=false,
    scale only axis,
    width=0.65\linewidth,
    height=3.0cm,
    mark size=1.0pt,
    font=\footnotesize,
    xlabel style={yshift=1ex},
    title style={yshift=-0.5ex},
    enlarge y limits=0.08,
    ytick distance=0.001,
    yticklabel style={/pgf/number format/fixed,/pgf/number format/precision=3},
]
\addplot[black, solid, mark=none]
coordinates
{(0.01,0.04876645583855118) (0.016237767391887217,0.04876645583855118) (0.02636650898730358,0.04876645583855118) (0.04281332398719394,0.04876645583855118) (0.06951927961775606,0.04876645583855118) (0.11288378916846889,0.04876645583855118) (0.18329807108324356,0.04876645583855118) (0.29763514416313175,0.04876645583855118) (0.4832930238571752,0.04876645583855118) (0.7847599703514611,0.04876645583855118) (1.2742749857031335,0.04876645583855118) (2.06913808111479,0.04876645583855118) (3.359818286283781,0.04876645583855118) (5.455594781168514,0.04876645583855118) (8.858667904100823,0.04876645583855118) (14.38449888287663,0.04876645583855118) (23.357214690901213,0.04876645583855118) (37.92690190732246,0.04876645583855118) (61.584821106602604,0.04876645583855118) (100.0,0.04876645583855118)}
;
\addplot[black, dashed, mark=square*, mark size=1.0pt, mark options={fill=white}]
coordinates
{(0.01,0.04876645583855118) (0.016237767391887217,0.048766455838551194) (0.02636650898730358,0.04876645583855117) (0.04281332398719394,0.048766455838551194) (0.06951927961775606,0.04876645583855117) (0.11288378916846889,0.04876645583855117) (0.18329807108324356,0.048766455838551194) (0.29763514416313175,0.04876645583855117) (0.4832930238571752,0.04876645583855118) (0.7847599703514611,0.048766455838551194) (1.2742749857031335,0.0487664558385512) (2.06913808111479,0.04876645583855117) (3.359818286283781,0.04876645583855117) (5.455594781168514,0.048766455838551166) (8.858667904100823,0.048766455838551166) (14.38449888287663,0.048766455838551166) (23.357214690901213,0.04876645583855113) (37.92690190732246,0.0487664558385512) (61.584821106602604,0.04876645583855117) (100.0,0.048766455838551194)}
;
\addplot[black, dotted, mark=diamond*, mark size=1.0pt]
coordinates
{(0.01,0.04522628871824022) (0.016237767391887217,0.0463486454942992) (0.02636650898730358,0.047480831834152984) (0.04281332398719394,0.04779207087248951) (0.06951927961775606,0.04802044328738652) (0.11288378916846889,0.04842103523171031) (0.18329807108324356,0.0484556297299152) (0.29763514416313175,0.04863420657715559) (0.4832930238571752,0.04869194546628785) (0.7847599703514611,0.048809526686420024) (1.2742749857031335,0.04872986492094675) (2.06913808111479,0.048757462527530805) (3.359818286283781,0.04873211974650815) (5.455594781168514,0.04875220090890722) (8.858667904100823,0.04876577762599834) (14.38449888287663,0.04875808305892399) (23.357214690901213,0.04875987156872306) (37.92690190732246,0.04875658678927084) (61.584821106602604,0.04877379669723203) (100.0,0.04876350040005571)}
;
\end{axis}
\end{tikzpicture}
\end{subfigure}%
\caption{Simulation results for $\beta^{\lambda}$ over the learning-rate grid, using $200$ common data sets with $G=100$, $m=5$, and $n=500$.
Panels show (a) coverage of nominal $95\%$ intervals, (b) mean interval width, (c) mean bias, and (d) standard deviation of the estimation error.
Solid line: sandwich Wald benchmark; dashed line with squares: calibrated posterior; dotted line with diamonds: uncalibrated posterior.
The horizontal reference in panel~(a) is $0.95$.
Posterior intervals used $20000$ retained Gibbs draws per data set and rate.
}
\label{fig:random-intercept-exp}
\end{figure}

\section{Discussion}
\label{sec:discussion}

Location--scale calibration aligns generalized posterior draws with the first-order sampling distribution of a penalized estimating-equation estimator.
Whereas a scalar learning rate only rescales the posterior, this matrix transformation corrects both scale and directional discrepancies between the posterior and sampling covariances.
The calibrated posterior therefore achieves valid frequentist coverage for every learning rate in a fixed compact subset of $(0,\infty)$.
This invariance relies on scaling the loss and penalty together, which keeps the penalized population target and its shrinkage fixed.

Although calibrated credible regions agree to first order with sandwich Wald regions at the same center, representing uncertainty through posterior draws offers practical advantages.
Once a single posterior sample is transformed, its empirical quantiles directly yield asymptotically valid intervals for smooth functionals without computing derivatives.
The plug-in procedure requires only empirical score moments and accounts for Monte Carlo error without bootstrap resampling or learning-rate tuning.
As shown in our Huber regression study, calibrated draws therefore reproduce frequentist sandwich inference reliably across a wide range of learning rates.

The smoothness conditions of the main results can be relaxed.
Section~\ref{hb:theory} verifies the required local expansions directly for the piecewise smooth Huber loss, and Section~\ref{supp:non-smooth} treats locally separable non-smooth penalties.
Under strict complementarity, the posterior concentrates at rate $s_n^{-1/2}$ with a Gaussian limit on the active coordinates and at the faster rate $s_n^{-1}$ on the inactive coordinates.
The marginal posterior of the active coordinates can therefore be calibrated under analogous moment and simulation conditions.
Because this penalty result assumes a smooth loss, combining Huber losses with non-smooth penalties requires a joint expansion across both local scales.

Several limitations stem from relying on a first-order local Gaussian approximation.
Because the theory is first-order, different admissible centers and matrix square roots share the same Gaussian limit but can alter the finite-sample distribution.
Higher-order expansions could clarify how these choices affect finite-sample accuracy relative to sandwich Wald intervals.
Global affine rescaling may also push draws outside a constrained parameter space in finite samples, making reparameterization to an unconstrained space useful before calibration.
Finally, our local expansion requires a fixed dimension, nonsingular curvature at an interior target, and consistent score covariance estimation under the sampling design.
Extending the calibration to diverging dimensions or singular models remains an important direction for future work.

\section*{Declaration of the use of generative AI and AI-assisted technologies}
During the preparation of this work the authors used ChatGPT and Claude in order to assist with writing and refactoring the simulation code.
After using these tools the authors reviewed and edited the content as necessary and take full responsibility for the content of the publication.

\section*{Code Availability}
\label{sec:code-availability}
The Python code for reproducing the experiments is available at \url{https://github.com/shutech2001/ls-calib-gp}.

\section*{Funding}
Shu Tamano was supported by JSPS KAKENHI Grant Numbers 25K24203.

\bibliographystyle{apalike}
\bibliography{bibliography}

\clearpage
\setcounter{section}{0}
\renewcommand{\thesection}{\Alph{section}}
\makeatletter
\@addtoreset{equation}{section}
\@addtoreset{table}{section}
\@addtoreset{figure}{section}
\setcounter{equation}{0}
\setcounter{table}{0}
\setcounter{figure}{0}
\renewcommand{\theequation}{\thesection.\arabic{equation}}
\renewcommand{\thetable}{\thesection.\arabic{table}}
\renewcommand{\thefigure}{\thesection.\arabic{figure}}
\def\supp@numbercounter#1{%
    \@ifundefined{c@#1}{}{%
        \@addtoreset{#1}{section}%
        \setcounter{#1}{0}%
        \expandafter\gdef\csname the#1\endcsname{\thesection.\arabic{#1}}%
    }%
}
\supp@numbercounter{theorem}
\let\supp@numbercounter\relax
\makeatother

\begin{center}
{\Large\bfseries Supplement to ``Location--Scale Calibration for Generalized Posteriors''}
\end{center}

\section{Examples of admissible centers}
\label{supp:ex-admissible}

This section shows that the MAP estimator, a one-step Newton update, and the posterior mean are admissible centers in the sense of Definition~\ref{def:admissible}.

\subsection{MAP centering}
We first consider maximizers of the generalized posterior density.
Write the generalized log-posterior, up to an additive constant, as
\begin{equation*}
    \ell_n(\boldsymbol{\theta})
    =
    -\eta s_nG_n(\boldsymbol{\theta})
    +
    r_n(\boldsymbol{\theta})
    .
\end{equation*}
Here $\eta$ may be any data-measurable selection from $\cE$.

For a global maximizer, consider the separation condition that, for every $\varepsilon>0$,
\begin{equation}
\label{eq:map-global}
    P^\ast\Biggl\{
        \sup_{\boldsymbol{\theta}\notin B_\varepsilon}
        [\ell_n(\boldsymbol{\theta})-\ell_n(\boldsymbol{\theta}^\lambda)]<0
    \Biggr\}
    \longrightarrow
    1
    .
\end{equation}
The conditions in Lemma~\ref{lem:separation-verification} imply \eqref{eq:map-global}.
Then we obtain the following result.
\begin{lemma}[MAP centers]
\label{prop:centers}
    Under Assumptions~\ref{ass:regularity}, \ref{ass:posterior} \ref{ass:posterior:local}, and~\ref{ass:baseline}, every consistent measurable interior local maximizer of $\ell_n$ is admissible.
    A measurable global maximizer is consistent and admissible if it exists with probability tending to one and \eqref{eq:map-global} holds.
\end{lemma}

\begin{proof}[Proof of Lemma~\ref{prop:centers}]
    Suppose $\boldsymbol{a}_n$ maximizes $\ell_n$ over $\Theta$ on an event $E_n$ with $P^\ast(E_n)\to1$.
    On $E_n\cap\{\boldsymbol{a}_n\notin B_\varepsilon\}$, global optimality gives $\sup_{\boldsymbol{\theta}\notin B_\varepsilon}\ell_n(\boldsymbol{\theta})\ge\ell_n(\boldsymbol{a}_n)\ge\ell_n(\boldsymbol{\theta}^\lambda)$, and hence \eqref{eq:map-global} implies, for every $\varepsilon>0$,
    \begin{equation*}
        \begin{split}
            P^\ast(\|\boldsymbol{a}_n-\boldsymbol{\theta}^\lambda\|\ge\varepsilon)
            &\le
            P^\ast\Biggl\{
                \sup_{\boldsymbol{\theta}\notin B_\varepsilon}
                [\ell_n(\boldsymbol{\theta})-\ell_n(\boldsymbol{\theta}^\lambda)]\ge0
            \Biggr\}
            +
            P^\ast(E_n^c)
            \\
            &\to
            0
            .
        \end{split}
    \end{equation*}
    Thus $\boldsymbol{a}_n\to_p\boldsymbol{\theta}^\lambda$, and since $\boldsymbol{\theta}^\lambda$ is an interior point of $\Theta$, $\boldsymbol{a}_n$ is an interior local maximizer with probability tending to one.
    Now let $\boldsymbol{a}_n$ be any consistent interior local maximizer of $\ell_n$.
    By consistency, Assumption~\ref{ass:regularity}~\ref{ass:regularity:population}--\ref{ass:regularity:empirical}, and Assumption~\ref{ass:posterior}~\ref{ass:posterior:local}, the event on which $\boldsymbol{a}_n\in B_{\delta_0}$ and $G_n$ is differentiable on $B_{\delta_0}\subset\cN$ has probability tending to one.
    On this event, for any unit vector $\boldsymbol{e}$ and all sufficiently small $t>0$, local optimality $\ell_n(\boldsymbol{a}_n+t\boldsymbol{e})\le\ell_n(\boldsymbol{a}_n)$ and the Lipschitz bound for $r_n$ on $B_{\delta_0}$ give
    \begin{equation*}
        \frac{G_n(\boldsymbol{a}_n+t\boldsymbol{e})-G_n(\boldsymbol{a}_n)}{t}
        \ge
        \frac{r_n(\boldsymbol{a}_n+t\boldsymbol{e})-r_n(\boldsymbol{a}_n)}{\eta s_nt}
        \ge
        -\frac{L_n}{\eta s_n}
        .
    \end{equation*}
    Letting $t\downarrow0$ and applying the same argument to $-\boldsymbol{e}$, we obtain, since $\eta\ge\eta_L>0$ and $L_n=o(\sqrt{s_n})$ by Assumption~\ref{ass:posterior}~\ref{ass:posterior:local},
    \begin{equation*}
        \|\boldsymbol{F}_n(\boldsymbol{a}_n)\|
        =
        \sup_{\|\boldsymbol{e}\|=1}
        |\boldsymbol{F}_n(\boldsymbol{a}_n)^\top\boldsymbol{e}|
        \le
        \frac{L_n}{\eta s_n}
        \le
        \frac{L_n}{\eta_Ls_n}
        =
        o(s_n^{-1/2})
        .
    \end{equation*}
    Therefore $\boldsymbol{F}_n(\boldsymbol{a}_n)=o_p(s_n^{-1/2})$, so $\boldsymbol{a}_n$ is a baseline estimator and Proposition~\ref{prop:baseline-clt} gives $\|\boldsymbol{a}_n-\check{\boldsymbol{\theta}}_n\|=o_p(s_n^{-1/2})$.
\end{proof}

\subsection{One-step centering}
We next consider a one-step Newton update.
For a measurable preliminary estimator $\boldsymbol{\theta}_n^{(0)}$, define the one-step update
\begin{equation}
\label{supp-eq:onestep-def}
    \tilde{\boldsymbol{\theta}}_n^{(1)}
    =
    \boldsymbol{\theta}_n^{(0)}
    -
    \boldsymbol{A}_n(\boldsymbol{\theta}_n^{(0)})^{-1}\boldsymbol{F}_n(\boldsymbol{\theta}_n^{(0)})
    ,\quad
    \boldsymbol{A}_n
    =
    \boldsymbol{J}_n+\lambda_n\boldsymbol{H}_\rho
    ,
\end{equation}
with the convention in Section~\ref{sec:setting} when the inverse does not exist.
Then we obtain the following result.
\begin{lemma}[One-step centers]
\label{lem:one-step}
    Under Assumptions~\ref{ass:regularity} and~\ref{ass:baseline}, the update in \eqref{supp-eq:onestep-def} is admissible whenever $\sqrt{s_n}(\boldsymbol{\theta}_n^{(0)}-\boldsymbol{\theta}^\lambda)=O_p(1)$.
\end{lemma}

\begin{proof}[Proof of Lemma~\ref{lem:one-step}]
    Since $\sqrt{s_n}(\boldsymbol{\theta}_n^{(0)}-\boldsymbol{\theta}^\lambda)=O_p(1)$ by assumption and $\sqrt{s_n}(\check{\boldsymbol{\theta}}_n-\boldsymbol{\theta}^\lambda)=O_p(1)$ by Proposition~\ref{prop:baseline-clt}, the triangle inequality gives $\boldsymbol{\theta}_n^{(0)}-\check{\boldsymbol{\theta}}_n=O_p(s_n^{-1/2})$.
    Because both estimators converge in probability to $\boldsymbol{\theta}^\lambda$, the line segment joining $\check{\boldsymbol{\theta}}_n$ and $\boldsymbol{\theta}_n^{(0)}$ lies in a fixed compact subset of $\cN$ with probability tending to one.
    On this event, the fundamental theorem of calculus and Assumption~\ref{ass:regularity}~\ref{ass:regularity:population}--\ref{ass:regularity:empirical} give
    \begin{equation*}
        \begin{gathered}
            \boldsymbol{F}_n(\boldsymbol{\theta}_n^{(0)})
            -\boldsymbol{F}_n(\check{\boldsymbol{\theta}}_n)
            =
            \boldsymbol{B}_n(\boldsymbol{\theta}_n^{(0)}-\check{\boldsymbol{\theta}}_n)
            ,\\
            \boldsymbol{B}_n
            =
            \int_0^1\boldsymbol{A}_n
            \{\check{\boldsymbol{\theta}}_n+t(\boldsymbol{\theta}_n^{(0)}-\check{\boldsymbol{\theta}}_n)\}\dd t
            \to_p
            \boldsymbol{J}_\lambda^\ast
            .
        \end{gathered}
    \end{equation*}
    Since $\boldsymbol{\theta}_n^{(0)}\to_p\boldsymbol{\theta}^\lambda$, Assumption~\ref{ass:regularity}~\ref{ass:regularity:population}--\ref{ass:regularity:empirical} also gives $\boldsymbol{A}_n(\boldsymbol{\theta}_n^{(0)})\to_p\boldsymbol{J}_\lambda^\ast$, which is positive definite by Assumption~\ref{ass:regularity}~\ref{ass:regularity:clt}.
    Consequently,
    \begin{equation*}
        \boldsymbol{A}_n(\boldsymbol{\theta}_n^{(0)})^{-1}
        =
        O_p(1)
        ,\quad
        \boldsymbol{I}_p
        -\boldsymbol{A}_n(\boldsymbol{\theta}_n^{(0)})^{-1}\boldsymbol{B}_n
        =
        o_p(1)
        .
    \end{equation*}
    Substituting the segment expansion into \eqref{supp-eq:onestep-def} and using $\boldsymbol{F}_n(\check{\boldsymbol{\theta}}_n)=o_p(s_n^{-1/2})$ from Assumption~\ref{ass:baseline} yields
    \begin{equation*}
        \begin{split}
            \tilde{\boldsymbol{\theta}}_n^{(1)}-\check{\boldsymbol{\theta}}_n
            &=
            \boldsymbol{\theta}_n^{(0)}-\check{\boldsymbol{\theta}}_n
            -\boldsymbol{A}_n(\boldsymbol{\theta}_n^{(0)})^{-1}
            \boldsymbol{F}_n(\boldsymbol{\theta}_n^{(0)})
            \\
            &=
            \{\boldsymbol{I}_p-\boldsymbol{A}_n(\boldsymbol{\theta}_n^{(0)})^{-1}
            \boldsymbol{B}_n\}(\boldsymbol{\theta}_n^{(0)}-\check{\boldsymbol{\theta}}_n)
            -\boldsymbol{A}_n(\boldsymbol{\theta}_n^{(0)})^{-1}
            \boldsymbol{F}_n(\check{\boldsymbol{\theta}}_n)
            \\
            &=
            o_p(1)O_p(s_n^{-1/2})+O_p(1)o_p(s_n^{-1/2})
            \\
            &=
            o_p(s_n^{-1/2})
            .
        \end{split}
    \end{equation*}
    This verifies Definition~\ref{def:admissible}, following the standard one-step argument for estimating equations \citep[Section~5.7]{vandervaart1998asymptotic}.
\end{proof}

\subsection{Generalized posterior mean centering}
Finally, we consider the generalized posterior mean.
Under the assumptions of Lemma~\ref{lem:postvar-consistency},
\begin{equation*}
    \sqrt{s_n}(\boldsymbol{\theta}_{\mathrm{GB}}-\check{\boldsymbol{\theta}}_n)
    =
    \EE_{\Pi_n^\eta}\Bigl[
        \sqrt{s_n}(\boldsymbol{\theta}-\check{\boldsymbol{\theta}}_n)\mid\cD_n
    \Bigr]
    \to_p
    \boldsymbol{0}
    .
\end{equation*}
Consequently, the posterior mean is admissible and can be the starting point of \eqref{supp-eq:onestep-def}.
The Monte Carlo mean is admissible under the joint data and simulation law when Assumption~\ref{ass:MC} also holds.

For a constrained parameter space, these conclusions hold on the events where the centers lie in the interior neighborhood of the target, with the convention in Section~\ref{sec:setting} elsewhere.

\section{Proofs for Section~\ref{sec:main}}
\label{supp:pf-main}

This section proves the results of Section~\ref{sec:main}, after establishing posterior concentration and collecting elementary properties of the total-variation distance.

\subsection{Global concentration}
This subsection gives sufficient conditions for Assumption~\ref{ass:posterior} \ref{ass:posterior:separation} and proves posterior concentration.
For each fixed $\varepsilon>0$, consider the following two separation conditions, where $|\Theta|$ denotes Lebesgue volume.
\begin{enumerate}[label=(\alph*),leftmargin=*]
    \item\label{cond:sep-a}
    $|\Theta|<\infty$, $\inf_{\Theta\setminus B_\varepsilon}\{G_n-G_n(\boldsymbol{\theta}^\lambda)\}\ge\Delta_\varepsilon>0$ with probability tending to one, and $\sup_\Theta |r_n-r_n(\boldsymbol{\theta}^\lambda)|/s_n=o_p(1)$.
    \item\label{cond:sep-b}
    For fixed $R\ge1$, $\kappa,c_0>0$, and $c_1\ge0$ with $\eta_Lc_0>c_1$, the lower bound on $G_n-G_n(\boldsymbol{\theta}^\lambda)$ in \ref{cond:sep-a} holds on $(\Theta\setminus B_\varepsilon)\cap\{\|\boldsymbol{\theta}\|\le R\}$, the difference $r_n-r_n(\boldsymbol{\theta}^\lambda)$ is uniformly $o_p(s_n)$ on $\Theta\cap\{\|\boldsymbol{\theta}\|\le R\}$, and, for $\|\boldsymbol{\theta}\|>R$,
    \begin{equation*}
        G_n(\boldsymbol{\theta})
        -
        G_n(\boldsymbol{\theta}^\lambda)
        \ge
        c_0\|\boldsymbol{\theta}\|^\kappa
        ,\quad
        r_n(\boldsymbol{\theta})
        -
        r_n(\boldsymbol{\theta}^\lambda)
        \le
        c_1s_n\|\boldsymbol{\theta}\|^\kappa
        +
        s_n\zeta_n,
        \quad
        \zeta_n
        =
        o_p(1)
        ,
    \end{equation*}
    with probability tending to one.
\end{enumerate}
The constants are chosen at the lower endpoint $\eta_L$ of $\cE$.
Then we obtain the following result.
\begin{lemma}[Separation from primitive bounds]
\label{lem:separation-verification}
    If, for every $\varepsilon>0$, either condition~\ref{cond:sep-a} or condition~\ref{cond:sep-b} holds, then Assumption~\ref{ass:posterior} \ref{ass:posterior:separation} and \eqref{eq:map-global} hold.
\end{lemma}
The lower bound on $G_n-G_n(\boldsymbol{\theta}^\lambda)$ required in \ref{cond:sep-a} and \ref{cond:sep-b} holds on any region on which $G_n$ converges uniformly to a continuous population criterion $G$ whose values outside $B_\varepsilon$ exceed $G(\boldsymbol{\theta}^\lambda)$ by a positive constant.

\begin{proof}[Proof of Lemma~\ref{lem:separation-verification}]
    In condition~\ref{cond:sep-a}, put $\xi_n=s_n^{-1}\sup_\Theta|r_n-r_n(\boldsymbol{\theta}^\lambda)|=o_p(1)$ and fix $0<a_\varepsilon<\eta_L\Delta_\varepsilon$.
    On the event where $\inf_{\Theta\setminus B_\varepsilon}\{G_n-G_n(\boldsymbol{\theta}^\lambda)\}\ge\Delta_\varepsilon$ and $\xi_n\le\eta_L\Delta_\varepsilon-a_\varepsilon$, whose probability tends to one by condition~\ref{cond:sep-a}, $\eta\ge\eta_L$ gives
    \begin{equation*}
        \begin{gathered}
            \sup_{\eta\in\cE}I_{n,\varepsilon}(\eta)
            \le
            |\Theta|\exp\{-s_n(\eta_L\Delta_\varepsilon-\xi_n)\}
            \le
            |\Theta|\exp(-a_\varepsilon s_n)
            ,\\
            \sup_{\boldsymbol{\theta}\notin B_\varepsilon}
            [\ell_n(\boldsymbol{\theta})-\ell_n(\boldsymbol{\theta}^\lambda)]
            \le
            -s_n(\eta_L\Delta_\varepsilon-\xi_n)
            \le
            -a_\varepsilon s_n
            <
            0
            ,
        \end{gathered}
    \end{equation*}
    which proves Assumption~\ref{ass:posterior}~\ref{ass:posterior:separation} and \eqref{eq:map-global} under condition~\ref{cond:sep-a}.

    In condition~\ref{cond:sep-b}, put $d=\eta_Lc_0-c_1>0$.
    The contribution to $I_{n,\varepsilon}(\eta)$ from $(\Theta\setminus B_\varepsilon)\cap\{\|\boldsymbol{\theta}\|\le R\}$ is bounded in the same way, with $|\Theta|$ replaced by the volume of the radius-$R$ ball and $\xi_n$ replaced by $s_n^{-1}\sup_{\Theta\cap\{\|\boldsymbol{\theta}\|\le R\}}|r_n-r_n(\boldsymbol{\theta}^\lambda)|=o_p(1)$.
    For $\|\boldsymbol{\theta}\|>R$, since $\|\boldsymbol{\theta}\|^\kappa\ge R^\kappa/2+\|\boldsymbol{\theta}\|^\kappa/2$ and $\eta\ge\eta_L$, the bounds in condition~\ref{cond:sep-b} and the change of variables $\boldsymbol{v}=(ds_n/2)^{1/\kappa}\boldsymbol{\theta}$ give, with probability tending to one,
    \begin{equation*}
        \begin{split}
            &\sup_{\eta\in\cE}\int_{\|\boldsymbol{\theta}\|>R}
            \exp[-\eta s_n\{G_n(\boldsymbol{\theta})-G_n(\boldsymbol{\theta}^\lambda)\}
            +r_n(\boldsymbol{\theta})-r_n(\boldsymbol{\theta}^\lambda)]
            \dd\boldsymbol{\theta}
            \\
            &\quad\le
            \int_{\|\boldsymbol{\theta}\|>R}
            \exp(s_n\zeta_n-ds_n\|\boldsymbol{\theta}\|^\kappa)
            \dd\boldsymbol{\theta}
            \\
            &\quad\le
            \exp(s_n\zeta_n-ds_nR^\kappa/2)
            \int_{\RR^p}\exp(-ds_n\|\boldsymbol{\theta}\|^\kappa/2)
            \dd\boldsymbol{\theta}
            \\
            &\quad=
            C_{p,\kappa}(ds_n/2)^{-p/\kappa}
            \exp\{-s_n(dR^\kappa/2-\zeta_n)\}
            ,
        \end{split}
    \end{equation*}
    where $C_{p,\kappa}=\int_{\RR^p}\exp(-\|\boldsymbol{v}\|^\kappa)\dd\boldsymbol{v}<\infty$.
    Choosing $0<a_\varepsilon<\min(\eta_L\Delta_\varepsilon,dR^\kappa/2)$ and using $\zeta_n=o_p(1)$, the two integral bounds yield $\sup_{\eta\in\cE}I_{n,\varepsilon}(\eta)=O_p\{\exp(-a_\varepsilon s_n)\}$.
    Moreover, on the intersection of the events where the bounds in condition~\ref{cond:sep-b} hold, $\sup_{\Theta\cap\{\|\boldsymbol{\theta}\|\le R\}}|r_n-r_n(\boldsymbol{\theta}^\lambda)|\le s_n\eta_L\Delta_\varepsilon/2$, and $\zeta_n\le dR^\kappa/2$, we have
    \begin{equation*}
        \sup_{\boldsymbol{\theta}\notin B_\varepsilon}
        [\ell_n(\boldsymbol{\theta})-\ell_n(\boldsymbol{\theta}^\lambda)]
        \le
        -\frac{s_n}{2}\min(\eta_L\Delta_\varepsilon,dR^\kappa)
        <
        0
        ,
    \end{equation*}
    proving \eqref{eq:map-global}.
    Finally, on any region $\cR$ containing $\boldsymbol{\theta}^\lambda$, the triangle inequality gives
    \begin{equation*}
        G_n(\boldsymbol{\theta})-G_n(\boldsymbol{\theta}^\lambda)
        \ge
        G(\boldsymbol{\theta})-G(\boldsymbol{\theta}^\lambda)
        -2\sup_{\boldsymbol{v}\in\cR}|G_n(\boldsymbol{v})-G(\boldsymbol{v})|
        ,
    \end{equation*}
    so uniform convergence $G_n\to_p G$ on $\cR$ and a positive lower bound on $G-G(\boldsymbol{\theta}^\lambda)$ on $\cR\setminus B_\varepsilon$ yield the lower bound on $G_n-G_n(\boldsymbol{\theta}^\lambda)$ in \ref{cond:sep-a} and \ref{cond:sep-b}.
\end{proof}

The next lemma gives a lower bound for the normalizing constant and posterior concentration, uniformly over $\cE$.
\begin{lemma}[Small-ball lower bound and concentration]
\label{lem:concentration}
    Under Assumptions~\ref{ass:regularity} and~\ref{ass:posterior} \ref{ass:posterior:local}, there are $A_n\ge0$, $A_n=O_p(1)$ and $B_n=O_p(1)$ such that, on one event of probability tending to one, simultaneously for $\eta\in\cE$,
    \begin{equation}
    \label{eq:small-ball}
        \hat{Z}_n(\eta)
        \coloneq
        \int_\Theta\exp[-\eta s_n\{G_n(\boldsymbol{\theta})-G_n(\boldsymbol{\theta}^\lambda)\}
        +
        r_n(\boldsymbol{\theta})-r_n(\boldsymbol{\theta}^\lambda)]\dd\boldsymbol{\theta}
        \ge
        s_n^{-p/2}\exp(-\eta A_n-B_n)
        .
    \end{equation}
    If Assumption~\ref{ass:posterior} \ref{ass:posterior:proper}, \ref{ass:posterior:separation} also holds, then, for every fixed $\varepsilon>0$,
    \begin{equation*}
        \sup_{\eta\in\cE}\Pi_n^\eta(B_\varepsilon^c\mid\cD_n)
        \to_p
        0
        ,
    \end{equation*}
    where $B_\varepsilon^c$ denotes the complement of $B_\varepsilon$.
\end{lemma}

\begin{proof}[Proof of Lemma~\ref{lem:concentration}]
    By the population equation \eqref{eq:pop-penalized} and Assumption~\ref{ass:regularity}~\ref{ass:regularity:population}, \ref{ass:regularity:clt},
    \begin{equation*}
        \sqrt{s_n}\boldsymbol{F}_n(\boldsymbol{\theta}^\lambda)
        =
        \sqrt{s_n}\{\boldsymbol{U}_n(\boldsymbol{\theta}^\lambda)
        -\Psi(\boldsymbol{\theta}^\lambda)\}
        +
        \sqrt{s_n}(\lambda_n-\lambda)\nabla\rho(\boldsymbol{\theta}^\lambda)
        =
        O_p(1)
        .
    \end{equation*}
    Choose a fixed closed ball $\cK$ centered at $\boldsymbol{\theta}^\lambda$ with $\cK\subset\cN\cap B_{\delta_0}$.
    On the event where $\bar{M}_n$ is $C^2$ on $\cN$, for all $\|\boldsymbol{u}\|\le1$ and all sufficiently large $n$ such that $\boldsymbol{\theta}^\lambda+s_n^{-1/2}\boldsymbol{u}\in\cK$, Taylor expansion along the line segment from $\boldsymbol{\theta}^\lambda$ to $\boldsymbol{\theta}^\lambda+s_n^{-1/2}\boldsymbol{u}$ gives
    \begin{equation*}
        \begin{split}
            s_n\{G_n(\boldsymbol{\theta}^\lambda+s_n^{-1/2}\boldsymbol{u})
            -G_n(\boldsymbol{\theta}^\lambda)\}
            &\le
            \sqrt{s_n}\boldsymbol{F}_n(\boldsymbol{\theta}^\lambda)^\top\boldsymbol{u}
            +\frac{1}{2}\sup_{\boldsymbol{\theta}\in\cK}
            \|\boldsymbol{J}_n(\boldsymbol{\theta})+\lambda_n\boldsymbol{H}_\rho(\boldsymbol{\theta})\|\|\boldsymbol{u}\|^2
            \\
            &\le
            \sqrt{s_n}\|\boldsymbol{F}_n(\boldsymbol{\theta}^\lambda)\|
            +\frac{1}{2}\sup_{\boldsymbol{\theta}\in\cK}
            \|\boldsymbol{J}_n(\boldsymbol{\theta})+\lambda_n\boldsymbol{H}_\rho(\boldsymbol{\theta})\|
            \eqqcolon
            A_n
            ,
        \end{split}
    \end{equation*}
    where $A_n=O_p(1)$ by Assumption~\ref{ass:regularity}~\ref{ass:regularity:population}--\ref{ass:regularity:empirical}.
    For $\|\boldsymbol{u}\|\le1$ and all sufficiently large $n$, Assumption~\ref{ass:posterior}~\ref{ass:posterior:local} also gives $|r_n(\boldsymbol{\theta}^\lambda+s_n^{-1/2}\boldsymbol{u})-r_n(\boldsymbol{\theta}^\lambda)|\le L_n/\sqrt{s_n}$.
    Letting $c_p$ denote the volume of the unit ball in $\RR^p$, restricting the integral in \eqref{eq:small-ball} to $\{\boldsymbol{\theta}^\lambda+s_n^{-1/2}\boldsymbol{u}:\|\boldsymbol{u}\|\le1\}$, and changing variables to $\boldsymbol{u}$ with Jacobian $s_n^{-p/2}$ yield, simultaneously for all $\eta\in\cE$ on one event of probability tending to one,
    \begin{equation*}
        \begin{split}
            \hat{Z}_n(\eta)
            &\ge
            c_ps_n^{-p/2}\exp(-\eta A_n-L_n/\sqrt{s_n})
            =
            s_n^{-p/2}\exp(-\eta A_n-B_n)
            ,\\
            B_n
            &\coloneq
            L_n/\sqrt{s_n}-\log c_p
            =
            O_p(1)
            ,
        \end{split}
    \end{equation*}
    which proves \eqref{eq:small-ball}.

    Under Assumption~\ref{ass:posterior}~\ref{ass:posterior:proper} and~\ref{ass:posterior:separation}, \eqref{eq:small-ball} gives
    \begin{equation}
    \label{eq:concentration-bound}
        \Pi_n^\eta(B_\varepsilon^c\mid\cD_n)
        =
        \frac{I_{n,\varepsilon}(\eta)}{\hat{Z}_n(\eta)}
        \le
        s_n^{p/2}\exp(\eta A_n+B_n)I_{n,\varepsilon}(\eta)
        .
    \end{equation}
    Since $A_n\ge0$, taking the supremum over $\eta\in\cE$ and applying Assumption~\ref{ass:posterior}~\ref{ass:posterior:separation} yields
    \begin{equation*}
        \begin{split}
            \sup_{\eta\in\cE}\Pi_n^\eta(B_\varepsilon^c\mid\cD_n)
            &\le
            s_n^{p/2}\exp(\eta_U A_n+B_n)
            \sup_{\eta\in\cE}I_{n,\varepsilon}(\eta)
            \\
            &=
            O_p\{s_n^{p/2}\exp(-a_\varepsilon s_n)\}
            \to_p
            0
            .
        \end{split}
    \end{equation*}
\end{proof}

\subsection{Elementary total-variation tools}
We use the following elementary properties of total variation and Gaussian laws.
Related convergence results are given in \citet[Lemma~2.2, Theorem~2.3, Lemma~2.8 and Corollary~2.30]{vandervaart1998asymptotic}.

The first property concerns measurable maps.
For a measurable map $T$, let $T_\#P$ denote the distribution of $T(X)$ when $X$ has law $P$.

\begin{lemma}[Total variation under measurable maps]
\label{supp-lem:tv-tools}
    For probability measures $P,Q$ and a measurable map $T$, $\|T_\#P-T_\#Q\|_{\TV}\le\|P-Q\|_{\TV}$, with equality when $T$ is a bijection with measurable inverse.
\end{lemma}

\begin{proof}[Proof of Lemma~\ref{supp-lem:tv-tools}]
    Since $T^{-1}(B)$ is measurable in the domain for every measurable set $B$ in the image space, taking suprema over measurable sets $B$ in the image space and $A$ in the domain gives
    \begin{equation*}
        \begin{split}
            \|T_\#P-T_\#Q\|_{\TV}
            &=
            \sup_B|P\{T^{-1}(B)\}-Q\{T^{-1}(B)\}|
            \\
            &\le
            \sup_A|P(A)-Q(A)|
            =
            \|P-Q\|_{\TV}
            .
        \end{split}
    \end{equation*}
    If $T$ has a measurable inverse, applying the same inequality to $T^{-1}$ gives the reverse bound and hence equality.
\end{proof}

The second property is the continuity of Gaussian laws in their parameters.
Write $\Phi$ for the standard normal distribution function.

\begin{lemma}[Gaussian continuity]
\label{supp-lem:gaussian}
    For vectors $\boldsymbol{m}_n,\boldsymbol{m}\in\RR^p$ and symmetric matrices $\boldsymbol{\Sigma}_n,\boldsymbol{\Sigma}\in\RR^{p\times p}$, if $\boldsymbol{m}_n\to\boldsymbol{0}$ and $\boldsymbol{\Sigma}_n\to\boldsymbol{\Sigma}\succ0$, the normal laws $N(\boldsymbol{m}_n,\boldsymbol{\Sigma}_n)$ converge in total variation to $N(\boldsymbol{0},\boldsymbol{\Sigma})$.
    The corresponding statement holds in probability for random parameters.
    For fixed $\boldsymbol{\Sigma}\succ0$,
    \begin{equation}
    \label{eq:gaussian-shift}
        \|N(\boldsymbol{m},\boldsymbol{\Sigma})-N(\boldsymbol{0},\boldsymbol{\Sigma})\|_{\TV}
        =
        2\Phi(\|\boldsymbol{\Sigma}^{-1/2}\boldsymbol{m}\|/2)-1
        \le
        \|\boldsymbol{\Sigma}^{-1/2}\boldsymbol{m}\|/\sqrt{2\pi}
        .
    \end{equation}
\end{lemma}

\begin{proof}[Proof of Lemma~\ref{supp-lem:gaussian}]
    Let $f_n$ and $f$ be the densities of $N(\boldsymbol{m}_n,\boldsymbol{\Sigma}_n)$ and $N(\boldsymbol{0},\boldsymbol{\Sigma})$.
    Eventually $\boldsymbol{\Sigma}_n\succ0$, and continuity of determinants and matrix inverses gives $f_n(\boldsymbol{x})\to f(\boldsymbol{x})$ for every $\boldsymbol{x}\in\RR^p$.
    Writing $a^+=\max(a,0)$ and using $0\le\{f(\boldsymbol{x})-f_n(\boldsymbol{x})\}^+\le f(\boldsymbol{x})$, dominated convergence gives
    \begin{equation*}
        \begin{split}
            \|N(\boldsymbol{m}_n,\boldsymbol{\Sigma}_n)
            -N(\boldsymbol{0},\boldsymbol{\Sigma})\|_{\TV}
            &=
            \frac{1}{2}\int_{\RR^p}|f_n(\boldsymbol{x})-f(\boldsymbol{x})|\dd\boldsymbol{x}
            \\
            &=
            \int_{\RR^p}\{f(\boldsymbol{x})-f_n(\boldsymbol{x})\}^+\dd\boldsymbol{x}
            \longrightarrow
            0
            .
        \end{split}
    \end{equation*}
    When $\boldsymbol{m}_n\to_p\boldsymbol{0}$ and $\boldsymbol{\Sigma}_n\to_p\boldsymbol{\Sigma}$, every subsequence has a further subsequence along which convergence holds almost surely, so the same conclusion holds in probability.

    For \eqref{eq:gaussian-shift}, put $\delta=\|\boldsymbol{\Sigma}^{-1/2}\boldsymbol{m}\|$ and $\boldsymbol{e}_1=(1,0,\ldots,0)^\top\in\RR^p$.
    If $\delta=0$ the identity is immediate.
    Otherwise, whitening by $\boldsymbol{\Sigma}^{-1/2}$ and rotating $\boldsymbol{\Sigma}^{-1/2}\boldsymbol{m}$ to $\delta\boldsymbol{e}_1$ preserves total variation by Lemma~\ref{supp-lem:tv-tools}.
    Writing $\boldsymbol{x}=(x_1,\ldots,x_p)^\top\in\RR^p$, the density of $N(\delta\boldsymbol{e}_1,\boldsymbol{I}_p)$ exceeds that of $N(\boldsymbol{0},\boldsymbol{I}_p)$ on $\{\boldsymbol{x}:x_1>\delta/2\}$.
    Hence
    \begin{equation*}
        \begin{split}
            \|N(\boldsymbol{m},\boldsymbol{\Sigma})
            -N(\boldsymbol{0},\boldsymbol{\Sigma})\|_{\TV}
            &=
            \|N(\delta\boldsymbol{e}_1,\boldsymbol{I}_p)
            -N(\boldsymbol{0},\boldsymbol{I}_p)\|_{\TV}
            \\
            &=
            \Phi(\delta/2)-\{1-\Phi(\delta/2)\}
            \\
            &=
            2\Phi(\delta/2)-1
            =
            2\int_0^{\delta/2}\frac{\exp(-t^2/2)}{\sqrt{2\pi}}\dd t
            \le
            \frac{\delta}{\sqrt{2\pi}}
            .
        \end{split}
    \end{equation*}
\end{proof}

The third property is a localization argument for normalized random densities.
Let $g_n\ge0$ be random functions on $\RR^p$ with $0<Z_n^u=\int g_n<\infty$, and let $\mu_n$ have density $g_n/Z_n^u$.
Suppose $\psi:\RR^p\to\RR$ is continuous and $0<C=\int\exp(-\psi)<\infty$.
Let $\cU_n\subset\RR^p$ be random measurable sets.
Then we obtain the following result.

\begin{lemma}[Localization]
\label{supp-lem:localization}
    Assume that, for each finite $M$, $\sup_{\|\boldsymbol{u}\|\le M}|\log g_n(\boldsymbol{u})+\psi(\boldsymbol{u})|\to_p0$ and the ball $\{\|\boldsymbol{u}\|\le M\}$ is contained in $\cU_n$ with probability tending to one, and that, on events of probability tending to one, $g_n\le b$ on $\cU_n$, where $b$ is a fixed integrable function.
    If $\mu_n(\cU_n^c)\to_p0$, then $\|\mu_n-C^{-1}\exp(-\psi)\dd\boldsymbol{u}\|_{\TV}\to_p0$.
    The conclusion is uniform over an index set when these hypotheses are uniform, $b$ is common, $C$ is uniformly bounded above and away from zero, and the limiting unnormalized densities have uniformly small tails.
\end{lemma}

\begin{proof}[Proof of Lemma~\ref{supp-lem:localization}]
    For fixed $M<\infty$, since $|\exp(t)-1|\le\exp(|t|)-1$, the uniform convergence $\sup_{\|\boldsymbol{u}\|\le M}|\log g_n(\boldsymbol{u})+\psi(\boldsymbol{u})|\to_p0$ gives
    \begin{equation*}
        \begin{split}
            \int_{\|\boldsymbol{u}\|\le M}|g_n-\exp(-\psi)|\dd\boldsymbol{u}
            &=
            \int_{\|\boldsymbol{u}\|\le M}\exp(-\psi)
            \bigl|
                \exp(\log g_n+\psi)-1
            \bigr|\dd\boldsymbol{u}
            \\
            &\le
            C\Bigl\{
                \exp\Bigl(
                    \sup_{\|\boldsymbol{u}\|\le M}|\log g_n+\psi|
                \Bigr)-1
            \Bigr\}
            \to_p
            0
            .
        \end{split}
    \end{equation*}
    On the event where $\{\|\boldsymbol{u}\|\le M\}\subset\cU_n$ and $g_n\le b$ on $\cU_n$,
    \begin{equation*}
        \begin{split}
            &\int_{\RR^p}\bigl|
                g_n\1\{\boldsymbol{u}\in\cU_n\}-\exp(-\psi)
            \bigr|\dd\boldsymbol{u}
            \\
            &\quad\le
            \int_{\|\boldsymbol{u}\|\le M}|g_n-\exp(-\psi)|\dd\boldsymbol{u}
            +
            \int_{\|\boldsymbol{u}\|>M}\{b+\exp(-\psi)\}\dd\boldsymbol{u}
            .
        \end{split}
    \end{equation*}
    Taking $n\to\infty$ and then $M\to\infty$ gives $\int_{\RR^p}|g_n\1\{\boldsymbol{u}\in\cU_n\}-\exp(-\psi)|\dd\boldsymbol{u}\to_p0$ and $\int_{\cU_n}g_n\dd\boldsymbol{u}\to_p C$.
    Since $\int_{\cU_n}g_n\dd\boldsymbol{u}=Z_n^u\{1-\mu_n(\cU_n^c)\}$ and $\mu_n(\cU_n^c)\to_p0$, it follows that
    \begin{equation*}
        Z_n^u
        =
        \frac{\int_{\cU_n}g_n\dd\boldsymbol{u}}{1-\mu_n(\cU_n^c)}
        \to_p
        C
        ,\quad
        \int_{\cU_n^c}g_n\dd\boldsymbol{u}
        =
        Z_n^u\mu_n(\cU_n^c)
        \to_p
        0
        .
    \end{equation*}
    Consequently, $\int_{\RR^p}|g_n-\exp(-\psi)|\dd\boldsymbol{u}\to_p0$, and
    \begin{equation*}
        \begin{split}
            2\|\mu_n-C^{-1}\exp(-\psi)\dd\boldsymbol{u}\|_{\TV}
            &=
            \int_{\RR^p}\Bigl|
                \frac{g_n}{Z_n^u}-\frac{\exp(-\psi)}{C}
            \Bigr|\dd\boldsymbol{u}
            \\
            &\le
            \frac{1}{Z_n^u}\int_{\RR^p}|g_n-\exp(-\psi)|\dd\boldsymbol{u}
            +\frac{|C-Z_n^u|}{Z_n^u}
            \\
            &\to_p
            0
            .
        \end{split}
    \end{equation*}
    Taking suprema over the index set in each bound above, using the common envelope $b$, the uniformly small tails of $\exp(-\psi)$, and the uniform bounds on $C$, yields the uniform assertion.
\end{proof}

\subsection{Baseline and posterior normal limits}
This subsection proves Propositions~\ref{prop:baseline-clt} and~\ref{prop:bvm-baseline}.

\begin{proof}[Proof of Proposition~\ref{prop:baseline-clt}]
    Since $\check{\boldsymbol{\theta}}_n\to_p\boldsymbol{\theta}^\lambda$ by Assumption~\ref{ass:baseline}, the line segment joining $\boldsymbol{\theta}^\lambda$ and $\check{\boldsymbol{\theta}}_n$ lies in a fixed compact subset of $\cN$ with probability tending to one.
    On this event, the fundamental theorem of calculus and Assumption~\ref{ass:regularity}~\ref{ass:regularity:population}--\ref{ass:regularity:empirical} give
    \begin{equation*}
        \begin{gathered}
            \boldsymbol{F}_n(\check{\boldsymbol{\theta}}_n)
            -\boldsymbol{F}_n(\boldsymbol{\theta}^\lambda)
            =
            \boldsymbol{B}_n(\check{\boldsymbol{\theta}}_n-\boldsymbol{\theta}^\lambda)
            ,\\
            \boldsymbol{B}_n
            =
            \int_0^1
            [\boldsymbol{J}_n+\lambda_n\boldsymbol{H}_\rho]
            \{\boldsymbol{\theta}^\lambda+t(\check{\boldsymbol{\theta}}_n-\boldsymbol{\theta}^\lambda)\}\dd t
            \to_p
            \boldsymbol{J}_\lambda^\ast
            .
        \end{gathered}
    \end{equation*}
    Since $\boldsymbol{J}_\lambda^\ast\succ0$ by Assumption~\ref{ass:regularity}~\ref{ass:regularity:clt}, $\boldsymbol{B}_n^{-1}\to_p(\boldsymbol{J}_\lambda^\ast)^{-1}$.
    Also, by the population equation \eqref{eq:pop-penalized} and Assumption~\ref{ass:regularity}~\ref{ass:regularity:population}, \ref{ass:regularity:clt},
    \begin{equation*}
        \begin{split}
            \sqrt{s_n}\boldsymbol{F}_n(\boldsymbol{\theta}^\lambda)
            &=
            \sqrt{s_n}\{\boldsymbol{U}_n(\boldsymbol{\theta}^\lambda)
            -\Psi(\boldsymbol{\theta}^\lambda)\}
            +\sqrt{s_n}(\lambda_n-\lambda)\nabla\rho(\boldsymbol{\theta}^\lambda)
            \\
            &\to_d
            N(\boldsymbol{0},\boldsymbol{K}^\ast)
            .
        \end{split}
    \end{equation*}
    Rearranging the segment identity and applying $\sqrt{s_n}\boldsymbol{F}_n(\check{\boldsymbol{\theta}}_n)=o_p(1)$ from Assumption~\ref{ass:baseline}, Slutsky's lemma \citep[Lemma~2.8]{vandervaart1998asymptotic}, and \eqref{eq:target-objects} therefore yield
    \begin{equation*}
        \begin{split}
            \sqrt{s_n}(\check{\boldsymbol{\theta}}_n-\boldsymbol{\theta}^\lambda)
            &=
            \boldsymbol{B}_n^{-1}
            \{\sqrt{s_n}\boldsymbol{F}_n(\check{\boldsymbol{\theta}}_n)
            -\sqrt{s_n}\boldsymbol{F}_n(\boldsymbol{\theta}^\lambda)\}
            \\
            &=
            -\boldsymbol{B}_n^{-1}\sqrt{s_n}\boldsymbol{F}_n(\boldsymbol{\theta}^\lambda)
            +o_p(1)
            \\
            &=
            -(\boldsymbol{J}_\lambda^\ast)^{-1}
            \sqrt{s_n}\boldsymbol{F}_n(\boldsymbol{\theta}^\lambda)+o_p(1)
            \\
            &\to_d
            N(\boldsymbol{0},\boldsymbol{V}_{\mathrm{target}}^\ast)
            .
        \end{split}
    \end{equation*}
    Any two baseline estimators have the same leading term, so subtracting their expansions gives a difference of order $o_p(s_n^{-1/2})$.
    For any admissible center $\tilde{\boldsymbol{\theta}}_n$, Definition~\ref{def:admissible} gives
    \begin{equation*}
        \sqrt{s_n}(\tilde{\boldsymbol{\theta}}_n-\boldsymbol{\theta}^\lambda)
        =
        \sqrt{s_n}(\check{\boldsymbol{\theta}}_n-\boldsymbol{\theta}^\lambda)+o_p(1)
        ,
    \end{equation*}
    which proves the remaining assertion.
\end{proof}

\begin{proof}[Proof of Proposition~\ref{prop:bvm-baseline}]
    Use $\boldsymbol{u}=\sqrt{s_n}(\boldsymbol{\theta}-\check{\boldsymbol{\theta}}_n)$.
    For $\check{\boldsymbol{\theta}}_n+s_n^{-1/2}\boldsymbol{u}\in\Theta$, define $g_n^\eta$ by
    \begin{equation*}
        \log g_n^\eta(\boldsymbol{u})
        =
        -\eta s_n\{G_n(\check{\boldsymbol{\theta}}_n+s_n^{-1/2}\boldsymbol{u})
        -G_n(\check{\boldsymbol{\theta}}_n)\}
        +r_n(\check{\boldsymbol{\theta}}_n+s_n^{-1/2}\boldsymbol{u})
        -r_n(\check{\boldsymbol{\theta}}_n)
        ,
    \end{equation*}
    and set $g_n^\eta(\boldsymbol{u})=0$ otherwise.
    Assumptions~\ref{ass:posterior}~\ref{ass:posterior:proper} and~\ref{ass:baseline} ensure that this definition is valid with probability tending to one, and normalizing $g_n^\eta$ yields the density of $Q_n^{\check{\boldsymbol{\theta}}_n,\eta}$.

    Recall $\boldsymbol{A}_n=\boldsymbol{J}_n+\lambda_n\boldsymbol{H}_\rho$ from \eqref{supp-eq:onestep-def}.
    By Assumption~\ref{ass:regularity}~\ref{ass:regularity:population}--\ref{ass:regularity:clt}, choose $\delta,c>0$ with $\overline{B}_{2\delta}\subset\cN\cap B_{\delta_0}$ such that, with probability tending to one,
    \begin{equation}
    \label{eq:local-curvature}
        \boldsymbol{A}_n(\boldsymbol{\theta})
        \succeq
        c\boldsymbol{I}_p
        \quad
        (\boldsymbol{\theta}\in\overline{B}_{2\delta})
        ,\quad
        \sup_{\boldsymbol{\theta}\in\overline{B}_{2\delta}}
        \|\boldsymbol{A}_n(\boldsymbol{\theta})\|
        =
        O_p(1)
        ,
    \end{equation}
    where $\boldsymbol{A}\succeq\boldsymbol{B}$ means that $\boldsymbol{A}-\boldsymbol{B}$ is positive semidefinite.
    For each fixed $M<\infty$, since $\check{\boldsymbol{\theta}}_n\to_p\boldsymbol{\theta}^\lambda$ by Assumption~\ref{ass:baseline}, the line segment from $\check{\boldsymbol{\theta}}_n$ to $\check{\boldsymbol{\theta}}_n+s_n^{-1/2}\boldsymbol{u}$ lies in $B_{2\delta}$ for all $\|\boldsymbol{u}\|\le M$ with probability tending to one.
    On this event, Taylor expansion along the segment gives
    \begin{equation*}
        \begin{split}
            s_n\{G_n(\check{\boldsymbol{\theta}}_n+s_n^{-1/2}\boldsymbol{u})
            -G_n(\check{\boldsymbol{\theta}}_n)\}
            &=
            \sqrt{s_n}\boldsymbol{F}_n(\check{\boldsymbol{\theta}}_n)^\top\boldsymbol{u}
            +\frac{1}{2}\boldsymbol{u}^\top\boldsymbol{B}_n(\boldsymbol{u})\boldsymbol{u}
            ,\\
            \boldsymbol{B}_n(\boldsymbol{u})
            &=
            2\int_0^1(1-t)\boldsymbol{A}_n
            (\check{\boldsymbol{\theta}}_n+ts_n^{-1/2}\boldsymbol{u})\dd t
            .
        \end{split}
    \end{equation*}
    Uniform Hessian convergence and continuity in Assumption~\ref{ass:regularity}~\ref{ass:regularity:population}--\ref{ass:regularity:empirical} imply
    \begin{equation*}
        \begin{split}
            \sup_{\|\boldsymbol{u}\|\le M}
            \|\boldsymbol{B}_n(\boldsymbol{u})-\boldsymbol{J}_\lambda^\ast\|
            &\le
            \sup_{\|\boldsymbol{u}\|\le M,\,0\le t\le1}
            \|\boldsymbol{A}_n(\check{\boldsymbol{\theta}}_n+ts_n^{-1/2}\boldsymbol{u})
            -\boldsymbol{J}_\lambda^\ast\|
            \\
            &\to_p
            0
            .
        \end{split}
    \end{equation*}
    Consequently, Assumptions~\ref{ass:posterior}~\ref{ass:posterior:local} and~\ref{ass:baseline} give
    \begin{equation*}
        \begin{split}
            &\sup_{\eta\in\cE}\sup_{\|\boldsymbol{u}\|\le M}
            \Bigl|
                \log g_n^\eta(\boldsymbol{u})
                +\frac{\eta}{2}\boldsymbol{u}^\top\boldsymbol{J}_\lambda^\ast\boldsymbol{u}
            \Bigr|
            \\
            &\quad\le
            \eta_U M\sqrt{s_n}\|\boldsymbol{F}_n(\check{\boldsymbol{\theta}}_n)\|
            +\frac{\eta_U M^2}{2}
            \sup_{\|\boldsymbol{u}\|\le M}
            \|\boldsymbol{B}_n(\boldsymbol{u})-\boldsymbol{J}_\lambda^\ast\|
            +\frac{L_nM}{\sqrt{s_n}}
            \\
            &\quad\to_p
            0
            .
        \end{split}
    \end{equation*}

    Set $\cU_n=\{\boldsymbol{u}:\check{\boldsymbol{\theta}}_n+s_n^{-1/2}\boldsymbol{u}\in B_\delta\}$.
    By Assumptions~\ref{ass:posterior}~\ref{ass:posterior:local} and~\ref{ass:baseline},
    \begin{equation*}
        \epsilon_n
        \coloneq
        \eta_U\sqrt{s_n}\|\boldsymbol{F}_n(\check{\boldsymbol{\theta}}_n)\|
        +L_n/\sqrt{s_n}
        \to_p
        0
        .
    \end{equation*}
    Consistency of $\check{\boldsymbol{\theta}}_n$ in Assumption~\ref{ass:baseline} ensures that every fixed ball $\{\|\boldsymbol{u}\|\le M\}$ lies in $\cU_n$ with probability tending to one.
    Moreover, when $\check{\boldsymbol{\theta}}_n\in B_\delta$, convexity of $B_\delta$ ensures that the line segment from $\check{\boldsymbol{\theta}}_n$ to $\check{\boldsymbol{\theta}}_n+s_n^{-1/2}\boldsymbol{u}$ lies in $B_\delta\subset B_{2\delta}$ for all $\boldsymbol{u}\in\cU_n$.
    On the event where $\check{\boldsymbol{\theta}}_n\in B_\delta$, \eqref{eq:local-curvature} holds, and $\epsilon_n^2\le\eta_Lc$, Taylor expansion and Young's inequality give, for all $\boldsymbol{u}\in\cU_n$ and $\eta\in\cE$,
    \begin{equation*}
        \log g_n^\eta(\boldsymbol{u})
        \le
        -\frac{\eta_Lc}{2}\|\boldsymbol{u}\|^2
        +\epsilon_n\|\boldsymbol{u}\|
        \le
        -\frac{\eta_Lc}{4}\|\boldsymbol{u}\|^2
        +\frac{\epsilon_n^2}{\eta_Lc}
        \le
        1-\frac{\eta_Lc}{4}\|\boldsymbol{u}\|^2
        .
    \end{equation*}
    Thus $b(\boldsymbol{u})=\exp(1-\eta_Lc\|\boldsymbol{u}\|^2/4)$ is a common integrable envelope on events of probability tending to one, and Lemma~\ref{lem:concentration} gives
    \begin{equation*}
        \sup_{\eta\in\cE}Q_n^{\check{\boldsymbol{\theta}}_n,\eta}(\cU_n^c)
        =
        \sup_{\eta\in\cE}\Pi_n^\eta(B_\delta^c\mid\cD_n)
        \to_p
        0
        .
    \end{equation*}

    The limiting kernel $\exp(-\eta\boldsymbol{u}^\top\boldsymbol{J}_\lambda^\ast\boldsymbol{u}/2)$ is bounded above by its value at $\eta_L$, and its normalizer $C_\eta=(2\pi)^{p/2}\det(\eta\boldsymbol{J}_\lambda^\ast)^{-1/2}$ satisfies $0<\inf_{\eta\in\cE}C_\eta\le\sup_{\eta\in\cE}C_\eta<\infty$ by Assumption~\ref{ass:regularity}~\ref{ass:regularity:clt}.
    Writing $\nu_\eta=N\{\boldsymbol{0},(\eta\boldsymbol{J}_\lambda^\ast)^{-1}\}$, the uniform form of Lemma~\ref{supp-lem:localization} gives
    \begin{equation*}
        \sup_{\eta\in\cE}
        \|Q_n^{\check{\boldsymbol{\theta}}_n,\eta}-\nu_\eta\|_{\TV}
        \to_p
        0
        .
    \end{equation*}
    Finally, for any admissible center $\tilde{\boldsymbol{\theta}}_n$, Definition~\ref{def:admissible} gives $\sqrt{s_n}(\tilde{\boldsymbol{\theta}}_n-\check{\boldsymbol{\theta}}_n)=o_p(1)$.
    Since $Q_n^{\tilde{\boldsymbol{\theta}}_n,\eta}$ is the translation of $Q_n^{\check{\boldsymbol{\theta}}_n,\eta}$ by $-\sqrt{s_n}(\tilde{\boldsymbol{\theta}}_n-\check{\boldsymbol{\theta}}_n)$, Lemmas~\ref{supp-lem:tv-tools} and~\ref{supp-lem:gaussian} imply
    \begin{equation*}
        \begin{split}
            &\sup_{\eta\in\cE}
            \|Q_n^{\tilde{\boldsymbol{\theta}}_n,\eta}-\nu_\eta\|_{\TV}
            \\
            &\quad\le
            \sup_{\eta\in\cE}
            \bigl\|
                Q_n^{\tilde{\boldsymbol{\theta}}_n,\eta}
                -
                N\{-\sqrt{s_n}(\tilde{\boldsymbol{\theta}}_n-\check{\boldsymbol{\theta}}_n),(\eta\boldsymbol{J}_\lambda^\ast)^{-1}\}
            \bigr\|_{\TV}
            \\
            &\qquad+
            \sup_{\eta\in\cE}
            \bigl\|
                N\{-\sqrt{s_n}(\tilde{\boldsymbol{\theta}}_n-\check{\boldsymbol{\theta}}_n),(\eta\boldsymbol{J}_\lambda^\ast)^{-1}\}
                -
                \nu_\eta
            \bigr\|_{\TV}
            \\
            &\quad\le
            \sup_{\eta\in\cE}
            \|Q_n^{\check{\boldsymbol{\theta}}_n,\eta}-\nu_\eta\|_{\TV}
            +
            \frac{\sqrt{\eta_U}\|(\boldsymbol{J}_\lambda^\ast)^{1/2}\|
            \sqrt{s_n}\|\tilde{\boldsymbol{\theta}}_n-\check{\boldsymbol{\theta}}_n\|}{\sqrt{2\pi}}
            \\
            &\quad\to_p
            0
            .
        \end{split}
    \end{equation*}
\end{proof}

\subsection{Calibration and coverage}
This subsection proves Theorem~\ref{thm:ls-calib} and Corollaries~\ref{cor:coverage} and~\ref{cor:functionals}, and extends Corollary~\ref{cor:coverage} to regions of a general fixed shape.

\begin{proof}[Proof of Theorem~\ref{thm:ls-calib}]
    For each $\eta\in\cE$, define $T_\eta(\boldsymbol{u})=\boldsymbol{\Omega}\boldsymbol{u}$ and $\nu_\eta=N(\boldsymbol{0},\boldsymbol{H}_0^{-1})$.
    By \eqref{eq:def-lsc-map},
    \begin{equation*}
        \begin{gathered}
            \sqrt{s_n}(\boldsymbol{\theta}_{\mathrm{calib}}^{(\eta)}-\tilde{\boldsymbol{\theta}}_n)
            =
            \boldsymbol{\Omega}\sqrt{s_n}(\boldsymbol{\theta}-\boldsymbol{a}_n)
            ,\\
            \boldsymbol{\Omega}\boldsymbol{H}_0^{-1}\boldsymbol{\Omega}^{\top}
            =
            (\boldsymbol{V}_{\mathrm{target}}^\ast)^{1/2}
            \boldsymbol{H}_0^{1/2}\boldsymbol{H}_0^{-1}\boldsymbol{H}_0^{1/2}
            (\boldsymbol{V}_{\mathrm{target}}^\ast)^{1/2}
            =
            \boldsymbol{V}_{\mathrm{target}}^\ast
            .
        \end{gathered}
    \end{equation*}
    Thus $\operatorname{Law}\{\sqrt{s_n}(\boldsymbol{\theta}_{\mathrm{calib}}^{(\eta)}-\tilde{\boldsymbol{\theta}}_n)\mid\cD_n\}=(T_\eta)_\#Q_n^{\boldsymbol{a}_n,\eta}$ and $(T_\eta)_\#\nu_\eta=N(\boldsymbol{0},\boldsymbol{V}_{\mathrm{target}}^\ast)$.
    Since $T_\eta$ is nonsingular, Lemma~\ref{supp-lem:tv-tools} and Proposition~\ref{prop:bvm-baseline} give
    \begin{equation*}
        \begin{split}
            \sup_{\eta\in\cE}
            \Bigl\|
                \operatorname{Law}\{\sqrt{s_n}(\boldsymbol{\theta}_{\mathrm{calib}}^{(\eta)}-\tilde{\boldsymbol{\theta}}_n)\mid\cD_n\}
                -
                N(\boldsymbol{0},\boldsymbol{V}_{\mathrm{target}}^\ast)
            \Bigr\|_{\TV}
            &=
            \sup_{\eta\in\cE}
            \|(T_\eta)_\#Q_n^{\boldsymbol{a}_n,\eta}-(T_\eta)_\#\nu_\eta\|_{\TV}
            \\
            &=
            \sup_{\eta\in\cE}
            \|Q_n^{\boldsymbol{a}_n,\eta}-\nu_\eta\|_{\TV}
            \to_p
            0
            .
        \end{split}
    \end{equation*}
\end{proof}

\begin{proof}[Proof of Corollary~\ref{cor:coverage}]
    Let $\mu_{n,\eta}=\operatorname{Law}\{\sqrt{s_n}(\boldsymbol{\theta}_{\mathrm{calib}}^{(\eta)}-\tilde{\boldsymbol{\theta}}_n)\mid\cD_n\}$ and $\nu=N(\boldsymbol{0},\boldsymbol{V}_{\mathrm{target}}^\ast)$.
    On the event $\{\hat{\boldsymbol{V}}_n\succ0\}$, whose probability tends to one since $\hat{\boldsymbol{V}}_n\to_p\boldsymbol{V}_{\mathrm{target}}^\ast\succ0$, put $A_n=\{\boldsymbol{x}\in\RR^p:\boldsymbol{x}^\top\hat{\boldsymbol{V}}_n^{-1}\boldsymbol{x}\le\chi^2_{p,1-\alpha}\}$.
    Then $\mu_{n,\eta}(A_n)$ is the calibrated posterior probability of $C_n(\alpha)$ and $N(\boldsymbol{0},\hat{\boldsymbol{V}}_n)(A_n)=1-\alpha$.
    Because total variation bounds probability differences uniformly over all Borel sets, Theorem~\ref{thm:ls-calib} and Lemma~\ref{supp-lem:gaussian} give
    \begin{equation*}
        \begin{split}
            \sup_{\eta\in\cE}|\mu_{n,\eta}(A_n)-(1-\alpha)|
            &\le
            \sup_{\eta\in\cE}\|\mu_{n,\eta}-\nu\|_{\TV}
            +
            \bigl|
                \nu(A_n)-N(\boldsymbol{0},\hat{\boldsymbol{V}}_n)(A_n)
            \bigr|
            \\
            &\le
            \sup_{\eta\in\cE}\|\mu_{n,\eta}-\nu\|_{\TV}
            +\|\nu-N(\boldsymbol{0},\hat{\boldsymbol{V}}_n)\|_{\TV}
            \to_p
            0
            .
        \end{split}
    \end{equation*}

    For frequentist coverage, Proposition~\ref{prop:baseline-clt} and Definition~\ref{def:admissible} give $\sqrt{s_n}(\tilde{\boldsymbol{\theta}}_n-\boldsymbol{\theta}^\lambda)\to_d N(\boldsymbol{0},\boldsymbol{V}_{\mathrm{target}}^\ast)$.
    Letting $\chi^2_p$ denote a chi-square distribution with $p$ degrees of freedom, Slutsky's lemma, the continuous mapping theorem, and continuity of the $\chi^2_p$ distribution function at $\chi^2_{p,1-\alpha}$ yield
    \begin{equation*}
        \begin{gathered}
            s_n(\tilde{\boldsymbol{\theta}}_n-\boldsymbol{\theta}^\lambda)^\top
            \hat{\boldsymbol{V}}_n^{-1}
            (\tilde{\boldsymbol{\theta}}_n-\boldsymbol{\theta}^\lambda)
            \to_d
            \chi^2_p
            ,\\
            P^\ast\{\boldsymbol{\theta}^\lambda\in C_n(\alpha)\}
            =
            P^\ast\bigl\{
                s_n(\tilde{\boldsymbol{\theta}}_n-\boldsymbol{\theta}^\lambda)^\top
                \hat{\boldsymbol{V}}_n^{-1}
                (\tilde{\boldsymbol{\theta}}_n-\boldsymbol{\theta}^\lambda)
                \le\chi^2_{p,1-\alpha}
            \bigr\}
            \longrightarrow
            1-\alpha
            .
        \end{gathered}
    \end{equation*}
\end{proof}

The same argument extends to regions of a general fixed shape.

\begin{corollary}[Regions of fixed shape]
\label{cor:fixed-shape}
    Suppose Assumptions~\ref{ass:regularity}, \ref{ass:posterior}, and~\ref{ass:baseline} hold.
    For any fixed Borel set $A$ whose boundary has zero probability under $N(\boldsymbol{0},\boldsymbol{V}_{\mathrm{target}}^{\ast})$, the region $\tilde{\boldsymbol{\theta}}_n+s_n^{-1/2}A$ has limiting calibrated posterior probability and coverage equal to $N(\boldsymbol{0},\boldsymbol{V}_{\mathrm{target}}^{\ast})(A)$.
\end{corollary}

\begin{proof}[Proof of Corollary~\ref{cor:fixed-shape}]
    Let $\mu_{n,\eta}=\operatorname{Law}\{\sqrt{s_n}(\boldsymbol{\theta}_{\mathrm{calib}}^{(\eta)}-\tilde{\boldsymbol{\theta}}_n)\mid\cD_n\}$ and $\nu=N(\boldsymbol{0},\boldsymbol{V}_{\mathrm{target}}^\ast)$.
    By Theorem~\ref{thm:ls-calib}, the calibrated posterior probability of $\tilde{\boldsymbol{\theta}}_n+s_n^{-1/2}A$ satisfies
    \begin{equation*}
        \sup_{\eta\in\cE}|\mu_{n,\eta}(A)-\nu(A)|
        \le
        \sup_{\eta\in\cE}\|\mu_{n,\eta}-\nu\|_{\TV}
        \to_p
        0
        .
    \end{equation*}
    Also, Proposition~\ref{prop:baseline-clt} and Definition~\ref{def:admissible} imply $-\sqrt{s_n}(\tilde{\boldsymbol{\theta}}_n-\boldsymbol{\theta}^\lambda)\to_d\nu$, since $\nu$ is symmetric about zero.
    Because the boundary of $A$ has $\nu$-measure zero, the portmanteau theorem \citep[Theorem~2.2]{vandervaart1998asymptotic} gives
    \begin{equation*}
        P^\ast\{\boldsymbol{\theta}^\lambda\in\tilde{\boldsymbol{\theta}}_n+s_n^{-1/2}A\}
        =
        P^\ast\{-\sqrt{s_n}(\tilde{\boldsymbol{\theta}}_n-\boldsymbol{\theta}^\lambda)\in A\}
        \longrightarrow
        \nu(A)
        .
    \end{equation*}
\end{proof}
The assertion on coverage uses symmetry of the limiting Gaussian and does not require symmetry of $A$.
Coverage of an arbitrary data-dependent region would additionally require control of its joint sampling law with the center.

\begin{proof}[Proof of Corollary~\ref{cor:functionals}]
    Let $\mu_{n,\eta}=\operatorname{Law}\{\sqrt{s_n}(\boldsymbol{\theta}_{\mathrm{calib}}^{(\eta)}-\tilde{\boldsymbol{\theta}}_n)\mid\cD_n\}$ and $\nu=N(\boldsymbol{0},\boldsymbol{V}_{\mathrm{target}}^\ast)$, so $\sup_{\eta\in\cE}\|\mu_{n,\eta}-\nu\|_{\TV}\to_p0$ by Theorem~\ref{thm:ls-calib}.
    Define
    \begin{equation*}
        T_n(\boldsymbol{x})
        =
        \sqrt{s_n}\{g(\tilde{\boldsymbol{\theta}}_n+s_n^{-1/2}\boldsymbol{x})
        -g(\tilde{\boldsymbol{\theta}}_n)\}
        ,
    \end{equation*}
    so $(T_n)_\#\mu_{n,\eta}=\operatorname{Law}[\sqrt{s_n}\{g(\boldsymbol{\theta}_{\mathrm{calib}}^{(\eta)})-g(\tilde{\boldsymbol{\theta}}_n)\}\mid\cD_n]$, while $\boldsymbol{\Gamma}^\ast\boldsymbol{x}\sim N(\boldsymbol{0},\boldsymbol{W}^\ast)$ under $\boldsymbol{x}\sim\nu$.
    Since $\tilde{\boldsymbol{\theta}}_n\to_p\boldsymbol{\theta}^\lambda$ by Assumption~\ref{ass:baseline} and Definition~\ref{def:admissible}, and $\nabla g$ is continuous near $\boldsymbol{\theta}^\lambda$, for each fixed $M<\infty$ the line segment from $\tilde{\boldsymbol{\theta}}_n$ to $\tilde{\boldsymbol{\theta}}_n+s_n^{-1/2}\boldsymbol{x}$ lies in the $C^1$ neighborhood of $g$ for all $\|\boldsymbol{x}\|\le M$ with probability tending to one, and
    \begin{equation*}
        \begin{split}
            T_n(\boldsymbol{x})-\boldsymbol{\Gamma}^\ast\boldsymbol{x}
            &=
            \int_0^1
            \{\nabla g(\tilde{\boldsymbol{\theta}}_n+ts_n^{-1/2}\boldsymbol{x})
            -\boldsymbol{\Gamma}^\ast\}\boldsymbol{x}\dd t
            ,\\
            \sup_{\|\boldsymbol{x}\|\le M}
            \|T_n(\boldsymbol{x})-\boldsymbol{\Gamma}^\ast\boldsymbol{x}\|
            &\le
            M\sup_{\|\boldsymbol{x}\|\le M,\,0\le t\le1}
            \|\nabla g(\tilde{\boldsymbol{\theta}}_n+ts_n^{-1/2}\boldsymbol{x})
            -\boldsymbol{\Gamma}^\ast\|
            \to_p
            0
            .
        \end{split}
    \end{equation*}
    For every $h$ with $\|h\|_\infty\le1$ and $\operatorname{Lip}(h)\le1$,
    \begin{equation*}
        \begin{split}
            &\sup_{\eta\in\cE}
            \Biggl|
                \int h\{T_n(\boldsymbol{x})\}\mu_{n,\eta}(\dd\boldsymbol{x})
                -
                \int h(\boldsymbol{\Gamma}^\ast\boldsymbol{x})\nu(\dd\boldsymbol{x})
            \Biggr|
            \\
            &\quad\le
            2\sup_{\eta\in\cE}\|\mu_{n,\eta}-\nu\|_{\TV}
            +
            \int_{\|\boldsymbol{x}\|\le M}
            |h\{T_n(\boldsymbol{x})\}-h(\boldsymbol{\Gamma}^\ast\boldsymbol{x})|
            \nu(\dd\boldsymbol{x})
            \\
            &\qquad+
            \int_{\|\boldsymbol{x}\|>M}
            |h\{T_n(\boldsymbol{x})\}-h(\boldsymbol{\Gamma}^\ast\boldsymbol{x})|
            \nu(\dd\boldsymbol{x})
            \\
            &\quad\le
            2\sup_{\eta\in\cE}\|\mu_{n,\eta}-\nu\|_{\TV}
            +
            \sup_{\|\boldsymbol{x}\|\le M}\|T_n(\boldsymbol{x})-\boldsymbol{\Gamma}^\ast\boldsymbol{x}\|
            +
            2\nu(\|\boldsymbol{x}\|>M)
            .
        \end{split}
    \end{equation*}
    Taking the supremum over $h$, followed by $n\to\infty$ and $M\to\infty$, proves
    \begin{equation*}
        \sup_{\eta\in\cE}
        d_{\mathrm{BL}}
        \{(T_n)_\#\mu_{n,\eta},N(\boldsymbol{0},\boldsymbol{W}^\ast)\}
        \to_p
        0
        .
    \end{equation*}

    For the functional ellipsoid, since $\nabla g(\tilde{\boldsymbol{\theta}}_n)\to_p\boldsymbol{\Gamma}^\ast$, $\boldsymbol{\Gamma}^\ast$ has full row rank, and $\hat{\boldsymbol{V}}_n\to_p\boldsymbol{V}_{\mathrm{target}}^\ast\succ0$, we have $\hat{\boldsymbol{W}}_n\to_p\boldsymbol{W}^\ast\succ0$ and $\hat{\boldsymbol{W}}_n^{-1/2}\to_p(\boldsymbol{W}^\ast)^{-1/2}$.
    By Proposition~\ref{prop:baseline-clt}, Definition~\ref{def:admissible}, and the delta method \citep[Theorem~3.1]{vandervaart1998asymptotic},
    \begin{equation*}
        \sqrt{s_n}\{g(\tilde{\boldsymbol{\theta}}_n)-g(\boldsymbol{\theta}^\lambda)\}
        \to_d
        N(\boldsymbol{0},\boldsymbol{W}^\ast)
        ,\quad
        s_n\{g(\tilde{\boldsymbol{\theta}}_n)-g(\boldsymbol{\theta}^\lambda)\}^\top
        \hat{\boldsymbol{W}}_n^{-1}
        \{g(\tilde{\boldsymbol{\theta}}_n)-g(\boldsymbol{\theta}^\lambda)\}
        \to_d
        \chi^2_k
        ,
    \end{equation*}
    which gives frequentist coverage $1-\alpha$.
    Furthermore, since $\hat{\boldsymbol{W}}_n^{-1/2}\to_p(\boldsymbol{W}^\ast)^{-1/2}$, the laws of $\hat{\boldsymbol{W}}_n^{-1/2}T_n(\boldsymbol{x})$ under $\boldsymbol{x}\sim\mu_{n,\eta}$ converge in bounded-Lipschitz distance to $N(\boldsymbol{0},\boldsymbol{I}_k)$ uniformly over $\eta\in\cE$.
    Bracketing the indicator $\1\{\|\boldsymbol{w}\|^2\le\chi^2_{k,1-\alpha}\}$ by bounded Lipschitz functions whose $N(\boldsymbol{0},\boldsymbol{I}_k)$-expectations are arbitrarily close to $1-\alpha$ yields
    \begin{equation*}
        \sup_{\eta\in\cE}
        \bigl|
            \mu_{n,\eta}\{T_n(\boldsymbol{x})^\top\hat{\boldsymbol{W}}_n^{-1}T_n(\boldsymbol{x})
            \le\chi^2_{k,1-\alpha}\}
            -(1-\alpha)
        \bigr|
        \to_p
        0
        .
    \end{equation*}

    For $k=1$, let $F_{n,\eta}$ be the conditional distribution function of $\sqrt{s_n}\{g(\boldsymbol{\theta}_{\mathrm{calib}}^{(\eta)})-g(\tilde{\boldsymbol{\theta}}_n)\}$ given $\cD_n$, and let $F$ be that of $N(0,\boldsymbol{W}^\ast)$.
    Write $q_t=F^{-1}(t)$ and $q_{n,\eta,t}=\inf\{x:F_{n,\eta}(x)\ge t\}$ for $0<t<1$.
    Uniform conditional weak convergence and strict increase of $F$ imply, for every $\varepsilon>0$,
    \begin{equation*}
        \begin{gathered}
            F(q_t-\varepsilon)<t<F(q_t+\varepsilon)
            ,\\
            P^\ast\Bigl\{
                \sup_{\eta\in\cE}F_{n,\eta}(q_t-\varepsilon)<t
                <\inf_{\eta\in\cE}F_{n,\eta}(q_t+\varepsilon)
            \Bigr\}
            \longrightarrow
            1
            .
        \end{gathered}
    \end{equation*}
    These inequalities bracket every $q_{n,\eta,t}$ between $q_t-\varepsilon$ and $q_t+\varepsilon$, so $\sup_{\eta\in\cE}|q_{n,\eta,t}-q_t|\to_p0$, and in particular $q_{n,\eta_n,t}\to_p q_t$ for any measurable $\eta_n\in\cE$.
    Since the equal-tailed calibrated interval is $[g(\tilde{\boldsymbol{\theta}}_n)+s_n^{-1/2}q_{n,\eta_n,\alpha/2},g(\tilde{\boldsymbol{\theta}}_n)+s_n^{-1/2}q_{n,\eta_n,1-\alpha/2}]$, Slutsky's lemma and symmetry of $N(0,\boldsymbol{W}^\ast)$ give
    \begin{equation*}
        P^\ast\bigl\{
            q_{n,\eta_n,\alpha/2}\le-\sqrt{s_n}\{g(\tilde{\boldsymbol{\theta}}_n)-g(\boldsymbol{\theta}^\lambda)\}\le q_{n,\eta_n,1-\alpha/2}
        \bigr\}
        \longrightarrow
        F(q_{1-\alpha/2})-F(q_{\alpha/2})
        =
        1-\alpha
        .
    \end{equation*}
\end{proof}

\section{Proofs for Section~\ref{sec:plugin}}
\label{supp:pf-plugin}

This section proves the results of Section~\ref{sec:plugin} and gives sufficient conditions for Assumptions~\ref{ass:UI} and~\ref{ass:MC}.

\subsection{Posterior moments and their tail condition}
This subsection proves Lemma~\ref{lem:postvar-consistency} and gives a sufficient condition for Assumption~\ref{ass:UI}.

\begin{proof}[Proof of Lemma~\ref{lem:postvar-consistency}]
    Write $\boldsymbol{X}_n=\sqrt{s_n}(\boldsymbol{\theta}-\check{\boldsymbol{\theta}}_n)$, $\mu_n=\operatorname{Law}(\boldsymbol{X}_n\mid\cD_n)$, and $\nu_n=N(\boldsymbol{0},\boldsymbol{H}_0^{-1})$, where $\boldsymbol{H}_0=\eta_n\boldsymbol{J}_\lambda^\ast$.
    Since $\eta_n\in\cE$, Proposition~\ref{prop:bvm-baseline} gives $\|\mu_n-\nu_n\|_{\TV}\to_p0$.
    For $M\ge1$ and any coordinate function $f(\boldsymbol{x})=x_j$ or $f(\boldsymbol{x})=x_jx_k$, since $|f(\boldsymbol{x})|\le M^2$ on $\{\|\boldsymbol{x}\|\le M\}$ and $|f(\boldsymbol{x})|\le\|\boldsymbol{x}\|^2$ on $\{\|\boldsymbol{x}\|>M\}$, splitting the integral gives
    \begin{equation*}
        \begin{split}
            \Biggl|
                \int f\dd(\mu_n-\nu_n)
            \Biggr|
            &\le
            \Biggl|
                \int_{\|\boldsymbol{x}\|\le M}f\dd(\mu_n-\nu_n)
            \Biggr|
            +
            \int_{\|\boldsymbol{x}\|>M}|f|\dd(\mu_n+\nu_n)
            \\
            &\le
            2M^2\|\mu_n-\nu_n\|_{\TV}
            +
            \int_{\|\boldsymbol{x}\|>M}\|\boldsymbol{x}\|^2\mu_n(\dd\boldsymbol{x})
            +
            \int_{\|\boldsymbol{x}\|>M}\|\boldsymbol{x}\|^2\nu_n(\dd\boldsymbol{x})
            .
        \end{split}
    \end{equation*}
    Assumption~\ref{ass:UI} controls the tail integral under $\mu_n$.
    Let $\boldsymbol{Z}\sim N(\boldsymbol{0},\boldsymbol{I}_p)$ be independent of the data and choose a fixed $C_0>0$ such that $\|\boldsymbol{H}_0^{-1/2}\|\le C_0$ for all $\eta_n\in\cE$.
    Then
    \begin{equation*}
        \begin{split}
            \int_{\|\boldsymbol{x}\|>M}\|\boldsymbol{x}\|^2\nu_n(\dd\boldsymbol{x})
            &=
            \EE\bigl[
                \|\boldsymbol{H}_0^{-1/2}\boldsymbol{Z}\|^2
                \1\{\|\boldsymbol{H}_0^{-1/2}\boldsymbol{Z}\|>M\}\mid\cD_n
            \bigr]
            \\
            &\le
            C_0^2\EE\bigl[
                \|\boldsymbol{Z}\|^2\1\{\|\boldsymbol{Z}\|>M/C_0\}
            \bigr]
            \longrightarrow
            0
            \quad
            (M\to\infty)
            ,
        \end{split}
    \end{equation*}
    uniformly in $n$.
    Taking $n\to\infty$ and then $M\to\infty$ in the truncation bound therefore yields
    \begin{equation*}
        \EE_{\Pi_n^\eta}[\boldsymbol{X}_n\mid\cD_n]
        \to_p
        \boldsymbol{0}
        ,\quad
        \bigl\|
            \EE_{\Pi_n^\eta}[\boldsymbol{X}_n\boldsymbol{X}_n^\top\mid\cD_n]-\boldsymbol{H}_0^{-1}
        \bigr\|
        \to_p
        0
        .
    \end{equation*}
    Consequently,
    \begin{equation*}
        \begin{split}
            \sqrt{s_n}(\boldsymbol{\theta}_{\mathrm{GB}}-\check{\boldsymbol{\theta}}_n)
            &=
            \EE_{\Pi_n^\eta}[\boldsymbol{X}_n\mid\cD_n]
            \\
            &=
            o_p(1)
            ,\\
            s_n\boldsymbol{\Sigma}_{\mathrm{post},n}
            &=
            \EE_{\Pi_n^\eta}[\boldsymbol{X}_n\boldsymbol{X}_n^\top\mid\cD_n]
            -
            \EE_{\Pi_n^\eta}[\boldsymbol{X}_n\mid\cD_n]\EE_{\Pi_n^\eta}[\boldsymbol{X}_n\mid\cD_n]^\top
            \\
            &=
            \boldsymbol{H}_0^{-1}+o_p(1)
            .
        \end{split}
    \end{equation*}
    Under Assumption~\ref{ass:MC}, adding the Monte Carlo errors gives
    \begin{equation*}
        \begin{split}
            \sqrt{s_n}(\hat{\boldsymbol{\theta}}_{\mathrm{GB}}-\check{\boldsymbol{\theta}}_n)
            &=
            \sqrt{s_n}(\hat{\boldsymbol{\theta}}_{\mathrm{GB}}-\boldsymbol{\theta}_{\mathrm{GB}})
            +
            \sqrt{s_n}(\boldsymbol{\theta}_{\mathrm{GB}}-\check{\boldsymbol{\theta}}_n)
            =
            o_p(1)
            ,\\
            \|s_n\hat{\boldsymbol{\Sigma}}_{\mathrm{post}}-\boldsymbol{H}_0^{-1}\|
            &\le
            s_n\|\hat{\boldsymbol{\Sigma}}_{\mathrm{post}}-\boldsymbol{\Sigma}_{\mathrm{post},n}\|
            +
            \|s_n\boldsymbol{\Sigma}_{\mathrm{post},n}-\boldsymbol{H}_0^{-1}\|
            \to_p
            0
            .
        \end{split}
    \end{equation*}
\end{proof}

We next give a sufficient condition for Assumption~\ref{ass:UI}.
For a sufficiently small fixed $\delta>0$, consider the weighted-tail condition
\begin{equation}
\label{eq:weighted-tail}
    \int_{B_\delta^c}s_n\|\boldsymbol{\theta}-\check{\boldsymbol{\theta}}_n\|^2
    \Pi_n^\eta(\dd\boldsymbol{\theta}\mid\cD_n)
    \to_p
    0
    .
\end{equation}
Then we obtain the following result.
\begin{lemma}[A sufficient weighted-tail condition]
\label{lem:weighted-tail}
    Under Assumptions~\ref{ass:regularity}--\ref{ass:baseline}, if the selected posterior has finite second moments almost surely and \eqref{eq:weighted-tail} holds, then Assumption~\ref{ass:UI} holds.
\end{lemma}
Condition~\eqref{eq:weighted-tail} holds if
\begin{equation*}
    \int_{\Theta\setminus B_\delta}s_n\|\boldsymbol{\theta}-\check{\boldsymbol{\theta}}_n\|^2
    \exp[-\eta s_n\{G_n(\boldsymbol{\theta})-G_n(\boldsymbol{\theta}^\lambda)\}
    +r_n(\boldsymbol{\theta})-r_n(\boldsymbol{\theta}^\lambda)]
    \dd\boldsymbol{\theta}
    =
    O_p\{\exp(-a s_n)\}
\end{equation*}
for some $a>0$.
On a bounded parameter space, Assumption~\ref{ass:posterior} \ref{ass:posterior:separation} already implies \eqref{eq:weighted-tail}.

\begin{proof}[Proof of Lemma~\ref{lem:weighted-tail}]
    Use $g_n=g_n^{\eta_n}$ and $\cU_n=\{\boldsymbol{u}:\check{\boldsymbol{\theta}}_n+s_n^{-1/2}\boldsymbol{u}\in B_\delta\}$ from the proof of Proposition~\ref{prop:bvm-baseline}, with $\delta$ small enough for \eqref{eq:local-curvature}.
    Put $c_0=\eta_Lc/4>0$ and $\kappa_n=\sup_{\|\boldsymbol{u}\|\le1}|\log g_n(\boldsymbol{u})|=O_p(1)$, and let $c_p$ be the volume of the unit ball in $\RR^p$.
    By the proof of Proposition~\ref{prop:bvm-baseline} and restricting the integral in $Z_n^u\coloneq\int_{\RR^p}g_n(\boldsymbol{u})\dd\boldsymbol{u}$ to $\{\|\boldsymbol{u}\|\le1\}$, on events of probability tending to one we have
    \begin{equation*}
        g_n(\boldsymbol{u})
        \le
        \exp(1-c_0\|\boldsymbol{u}\|^2)
        \quad
        (\boldsymbol{u}\in\cU_n)
        ,\quad
        Z_n^u
        \ge
        c_p\exp(-\kappa_n)
        .
    \end{equation*}
    For $\boldsymbol{X}_n=\sqrt{s_n}(\boldsymbol{\theta}-\check{\boldsymbol{\theta}}_n)$, since $\boldsymbol{X}_n\in\cU_n^c$ is equivalent to $\boldsymbol{\theta}\in B_\delta^c$, splitting the expectation over $\cU_n$ and $\cU_n^c$ gives
    \begin{equation*}
        \begin{split}
            &\EE_{\Pi_n^\eta}\bigl[
                \|\boldsymbol{X}_n\|^2\1\{\|\boldsymbol{X}_n\|>M\}\mid\cD_n
            \bigr]
            \\
            &\quad\le
            \frac{1}{Z_n^u}
            \int_{\cU_n\cap\{\|\boldsymbol{u}\|>M\}}
            \|\boldsymbol{u}\|^2g_n(\boldsymbol{u})\dd\boldsymbol{u}
            +
            \int_{B_\delta^c}s_n\|\boldsymbol{\theta}-\check{\boldsymbol{\theta}}_n\|^2
            \Pi_n^\eta(\dd\boldsymbol{\theta}\mid\cD_n)
            \\
            &\quad\le
            c_p^{-1}\exp(1+\kappa_n)
            \int_{\|\boldsymbol{u}\|>M}
            \|\boldsymbol{u}\|^2\exp(-c_0\|\boldsymbol{u}\|^2)\dd\boldsymbol{u}
            +
            \int_{B_\delta^c}s_n\|\boldsymbol{\theta}-\check{\boldsymbol{\theta}}_n\|^2
            \Pi_n^\eta(\dd\boldsymbol{\theta}\mid\cD_n)
            .
        \end{split}
    \end{equation*}
    Since $c_p^{-1}\exp(1+\kappa_n)=O_p(1)$, $\int_{\|\boldsymbol{u}\|>M}\|\boldsymbol{u}\|^2\exp(-c_0\|\boldsymbol{u}\|^2)\dd\boldsymbol{u}\to0$ as $M\to\infty$, and the second term is $o_p(1)$ by \eqref{eq:weighted-tail}, letting $n\to\infty$ and then $M\to\infty$ proves Assumption~\ref{ass:UI}.

    For the sufficient condition stated after the lemma, dividing the unnormalized tail integral by $\hat{Z}_n(\eta)$ in \eqref{eq:small-ball} gives
    \begin{equation*}
        \int_{B_\delta^c}s_n\|\boldsymbol{\theta}-\check{\boldsymbol{\theta}}_n\|^2
        \Pi_n^\eta(\dd\boldsymbol{\theta}\mid\cD_n)
        \le
        s_n^{p/2}\exp(\eta_U A_n+B_n)
        O_p\{\exp(-a s_n)\}
        =
        o_p(1)
        .
    \end{equation*}
    If $\Theta$ is bounded, choose $D_\Theta<\infty$ with $\|\boldsymbol{a}-\boldsymbol{b}\|\le D_\Theta$ for all $\boldsymbol{a},\boldsymbol{b}\in\Theta$.
    Assumption~\ref{ass:posterior}~\ref{ass:posterior:separation} and \eqref{eq:concentration-bound} then yield
    \begin{equation*}
        \begin{split}
            \int_{B_\delta^c}s_n\|\boldsymbol{\theta}-\check{\boldsymbol{\theta}}_n\|^2
            \Pi_n^\eta(\dd\boldsymbol{\theta}\mid\cD_n)
            &\le
            D_\Theta^2s_n\Pi_n^\eta(B_\delta^c\mid\cD_n)
            \\
            &\le
            D_\Theta^2s_n^{1+p/2}\exp(\eta_U A_n+B_n)
            O_p\{\exp(-a_\delta s_n)\}
            =
            o_p(1)
            .
        \end{split}
    \end{equation*}
    Boundedness of $\Theta$ also ensures finite posterior second moments.
\end{proof}

Assumption~\ref{ass:UI} cannot be omitted.
For example, the probability laws $(1-n^{-1})N(0,1)+n^{-1}N(n,1)$ approach $N(0,1)$ in total variation, but their means equal one and their second moments diverge.

\subsection{Monte Carlo accuracy}
This subsection gives a sufficient condition for Assumption~\ref{ass:MC}.
Write $\mu_{4,n}=\EE_{\Pi_n^\eta}[s_n^2\|\boldsymbol{\theta}-\boldsymbol{\theta}_{\mathrm{GB}}\|^4\mid\cD_n]$ for the scaled posterior fourth central moment.
Then we obtain the following result.

\begin{lemma}[A sufficient Monte Carlo condition]
\label{lem:mc-sufficient}
    Suppose Assumptions~\ref{ass:regularity}, \ref{ass:posterior}, \ref{ass:baseline}, and~\ref{ass:UI} hold, and $\mu_{4,n}$ is finite almost surely and $O_p(1)$.
    Conditionally on the data, suppose the retained chain is stationary and reversible, and its Markov operator on mean-zero $L^2(\Pi_n^\eta)$ has norm at most $\varrho_n<1$.
    If $D_n(1-\varrho_n)\to_p\infty$, Assumption~\ref{ass:MC} holds.
\end{lemma}
Independent draws correspond to $\varrho_n=0$.
The condition expresses growth of the effective number of draws, accounting for their dependence, and a separate requirement such as $D_n/s_n\to\infty$ is unnecessary.

\begin{proof}[Proof of Lemma~\ref{lem:mc-sufficient}]
    Let $P_n$ be the conditional transition operator given $\cD_n$, and let $f$ be a square-integrable function centered under $\Pi_n^\eta$.
    Write $\langle f,g\rangle_{L^2(\Pi_n^\eta)}=\EE_{\Pi_n^\eta}[f(\boldsymbol{\theta})g(\boldsymbol{\theta})\mid\cD_n]$ and $\|f\|_2^2=\langle f,f\rangle_{L^2(\Pi_n^\eta)}$.
    Stationarity, the Markov property, the Cauchy--Schwarz inequality, and the operator-norm bound give
    \begin{equation*}
        \bigl|
            \cov\{f(\boldsymbol{\theta}^{(d)}),f(\boldsymbol{\theta}^{(d+k)})\mid\cD_n\}
        \bigr|
        =
        |\langle f,P_n^kf\rangle_{L^2(\Pi_n^\eta)}|
        \le
        \|f\|_2\|P_n^kf\|_2
        \le
        \varrho_n^k\|f\|_2^2
        .
    \end{equation*}
    Hence
    \begin{equation}
    \label{eq:mc-variance}
        \begin{split}
            \var\Biggl\{
                D_n^{-1}\sum_{d=1}^{D_n}f(\boldsymbol{\theta}^{(d)})\mid\cD_n
            \Biggr\}
            &=
            D_n^{-2}\Biggl\{
                D_n\|f\|_2^2
                +
                2\sum_{k=1}^{D_n-1}(D_n-k)\langle f,P_n^kf\rangle_{L^2(\Pi_n^\eta)}
            \Biggr\}
            \\
            &\le
            \frac{\|f\|_2^2}{D_n}
            \Biggl(1+2\sum_{k=1}^{\infty}\varrho_n^k\Biggr)
            \\
            &=
            \frac{1+\varrho_n}{D_n(1-\varrho_n)}
            \var_{\Pi_n^\eta}(f\mid\cD_n)
            .
        \end{split}
    \end{equation}
    Put $\tau_n=\{D_n(1-\varrho_n)\}^{-1}\to_p0$, write subscripted nonbold letters for vector and matrix entries, and let $\tr$ denote the matrix trace.
    Applying \eqref{eq:mc-variance} to $\theta_j-\theta_{\mathrm{GB},j}$ for each coordinate $j=1,\ldots,p$ and using $1+\varrho_n\le2$ and Lemma~\ref{lem:postvar-consistency} yield
    \begin{equation*}
        \begin{split}
            \EE\bigl[
                s_n\|\hat{\boldsymbol{\theta}}_{\mathrm{GB}}-\boldsymbol{\theta}_{\mathrm{GB}}\|^2
                \mid
                \cD_n
            \bigr]
            &=
            s_n\sum_{j=1}^p
            \var\Biggl\{
                D_n^{-1}\sum_{d=1}^{D_n}\theta_j^{(d)}
                \mid
                \cD_n
            \Biggr\}
            \\
            &\le
            2\tau_ns_n\tr(\boldsymbol{\Sigma}_{\mathrm{post},n})
            =
            o_p(1)
            .
        \end{split}
    \end{equation*}
    For every $\varepsilon>0$, conditional Markov's inequality and dominated convergence therefore give
    \begin{equation*}
        P^\ast\bigl(
            s_n\|\hat{\boldsymbol{\theta}}_{\mathrm{GB}}-\boldsymbol{\theta}_{\mathrm{GB}}\|^2>\varepsilon
        \bigr)
        \le
        \EE\Bigl[
            \min\Bigl\{
                1,\varepsilon^{-1}\EE\bigl[
                    s_n\|\hat{\boldsymbol{\theta}}_{\mathrm{GB}}-\boldsymbol{\theta}_{\mathrm{GB}}\|^2\mid\cD_n
                \bigr]
            \Bigr\}
        \Bigr]
        \longrightarrow
        0
        ,
    \end{equation*}
    which proves $\sqrt{s_n}\|\hat{\boldsymbol{\theta}}_{\mathrm{GB}}-\boldsymbol{\theta}_{\mathrm{GB}}\|\to_p0$.

    For the covariance, define
    \begin{equation*}
        \tilde{\boldsymbol{\Sigma}}_n
        =
        D_n^{-1}\sum_{d=1}^{D_n}
        (\boldsymbol{\theta}^{(d)}-\boldsymbol{\theta}_{\mathrm{GB}})
        (\boldsymbol{\theta}^{(d)}-\boldsymbol{\theta}_{\mathrm{GB}})^\top
        ,
    \end{equation*}
    and set $f_{jk}(\boldsymbol{\theta})=(\theta_j-\theta_{\mathrm{GB},j})(\theta_k-\theta_{\mathrm{GB},k})$.
    Since
    \begin{equation*}
        \var_{\Pi_n^\eta}(f_{jk}\mid\cD_n)
        \le
        \EE_{\Pi_n^\eta}\bigl[
            \|\boldsymbol{\theta}-\boldsymbol{\theta}_{\mathrm{GB}}\|^4
            \mid
            \cD_n
        \bigr]
        =
        s_n^{-2}\mu_{4,n}
        ,
    \end{equation*}
    bounding the spectral norm by the Frobenius norm and applying \eqref{eq:mc-variance} to $f_{jk}-\Sigma_{\mathrm{post},n,jk}$ give
    \begin{equation*}
        \begin{split}
            \EE\bigl[s_n^2\|\tilde{\boldsymbol{\Sigma}}_n-\boldsymbol{\Sigma}_{\mathrm{post},n}\|^2\mid\cD_n\bigr]
            &\le
            \sum_{j,k=1}^p
            \EE\bigl[
                s_n^2(\tilde{\Sigma}_{n,jk}-\Sigma_{\mathrm{post},n,jk})^2
                \mid
                \cD_n
            \bigr]
            \\
            &\le
            2p^2\tau_n\mu_{4,n}
            =
            o_p(1)
            .
        \end{split}
    \end{equation*}
    The same conditional Markov argument gives $s_n\|\tilde{\boldsymbol{\Sigma}}_n-\boldsymbol{\Sigma}_{\mathrm{post},n}\|\to_p0$.
    Finally, since $\hat{\boldsymbol{\Sigma}}_{\mathrm{post}}=\tilde{\boldsymbol{\Sigma}}_n-(\hat{\boldsymbol{\theta}}_{\mathrm{GB}}-\boldsymbol{\theta}_{\mathrm{GB}})(\hat{\boldsymbol{\theta}}_{\mathrm{GB}}-\boldsymbol{\theta}_{\mathrm{GB}})^\top$, we obtain
    \begin{equation*}
        \begin{split}
            s_n\|\hat{\boldsymbol{\Sigma}}_{\mathrm{post}}-\boldsymbol{\Sigma}_{\mathrm{post},n}\|
            &\le
            s_n\|\tilde{\boldsymbol{\Sigma}}_n-\boldsymbol{\Sigma}_{\mathrm{post},n}\|
            +
            s_n\|\hat{\boldsymbol{\theta}}_{\mathrm{GB}}-\boldsymbol{\theta}_{\mathrm{GB}}\|^2
            =
            o_p(1)
            ,
        \end{split}
    \end{equation*}
    proving the covariance condition.
\end{proof}

Lemma~\ref{lem:mc-sufficient} extends to two other simulation schemes.
For a fixed number of conditionally independent stationary chains with this operator-norm bound, the cross-chain covariances vanish, and summing the within-chain bounds gives \eqref{eq:mc-variance} with $D_n$ equal to the total retained length.
The same conclusion therefore holds.
If the retained path is not stationary but its conditional path law is within $o_p(1)$ in total variation of the stationary law, the probabilities of all error events in the proof change by at most this distance.
The conclusion therefore holds for that path as well.
For such nonstationary simulation, the path-law condition must be justified in addition to the number of retained iterations.

\subsection{Target matrices and centering the scores}
This subsection states and proves the consistency of the estimators in \eqref{eq:target-plugin} and~\eqref{eq:K-centered}.
The first lemma concerns the target matrices.

\begin{lemma}[Plug-in consistency for the target]
\label{lem:target-plugin}
    Under Assumption~\ref{ass:regularity}, if $\bar{\boldsymbol{a}}_n\to_p\boldsymbol{\theta}^{\lambda}$ and a symmetric estimator satisfies $\hat{\boldsymbol{K}}\to_p\boldsymbol{K}^{\ast}$, then $\hat{\boldsymbol{J}}_{\lambda}\to_p\boldsymbol{J}_{\lambda}^{\ast}$ and $\hat{\boldsymbol{V}}_{\mathrm{target}}\to_p\boldsymbol{V}_{\mathrm{target}}^{\ast}$.
\end{lemma}

\begin{proof}[Proof of Lemma~\ref{lem:target-plugin}]
    Choose a compact ball $\cK\subset\cN$ centered at $\boldsymbol{\theta}^\lambda$.
    On the event $\{\bar{\boldsymbol{a}}_n\in\cK\}$, whose probability tends to one since $\bar{\boldsymbol{a}}_n\to_p\boldsymbol{\theta}^\lambda$, the triangle inequality, uniform convergence of $\boldsymbol{J}_n$ on $\cK$, continuity of $\boldsymbol{J}$ and $\boldsymbol{H}_\rho$ at $\boldsymbol{\theta}^\lambda$, and $\lambda_n\to\lambda$ in Assumption~\ref{ass:regularity}~\ref{ass:regularity:population}--\ref{ass:regularity:empirical} give
    \begin{equation*}
        \begin{split}
            \|\hat{\boldsymbol{J}}_\lambda-\boldsymbol{J}_\lambda^\ast\|
            &\le
            \sup_{\boldsymbol{\theta}\in\cK}
            \|\boldsymbol{J}_n(\boldsymbol{\theta})-\boldsymbol{J}(\boldsymbol{\theta})\|
            +
            \|\boldsymbol{J}(\bar{\boldsymbol{a}}_n)-\boldsymbol{J}(\boldsymbol{\theta}^\lambda)\|
            \\
            &\quad
            +
            \lambda_n\|\boldsymbol{H}_\rho(\bar{\boldsymbol{a}}_n)-\boldsymbol{H}_\rho(\boldsymbol{\theta}^\lambda)\|
            +
            |\lambda_n-\lambda|\|\boldsymbol{H}_\rho(\boldsymbol{\theta}^\lambda)\|
            \\
            &=
            o_p(1)
            .
        \end{split}
    \end{equation*}
    Assumption~\ref{ass:regularity}~\ref{ass:regularity:clt} gives $\boldsymbol{J}_\lambda^\ast\succ0$, so the symmetric matrix $\hat{\boldsymbol{J}}_\lambda$ is positive definite with probability tending to one.
    Since $\hat{\boldsymbol{K}}\to_p\boldsymbol{K}^\ast$, continuity of matrix inversion and multiplication yields
    \begin{equation*}
        \begin{split}
            \hat{\boldsymbol{J}}_\lambda^{-1}
            &\to_p
            (\boldsymbol{J}_\lambda^\ast)^{-1}
            ,\\
            \hat{\boldsymbol{V}}_{\mathrm{target}}
            =
            \hat{\boldsymbol{J}}_\lambda^{-1}\hat{\boldsymbol{K}}\hat{\boldsymbol{J}}_\lambda^{-1}
            &\to_p
            (\boldsymbol{J}_\lambda^\ast)^{-1}\boldsymbol{K}^\ast(\boldsymbol{J}_\lambda^\ast)^{-1}
            =
            \boldsymbol{V}_{\mathrm{target}}^\ast
            ,
        \end{split}
    \end{equation*}
    by the continuous mapping theorem \citep[Theorem~2.3]{vandervaart1998asymptotic}.
\end{proof}

The second lemma concerns the centered score covariance in the i.i.d.\ setting of Section~\ref{sec:plugin}, where $\boldsymbol{\varphi}$ denotes a generic summand.

\begin{lemma}[Centered score covariance]\label{lem:K-consistency}
    Suppose $\boldsymbol{\varphi}$ is almost surely continuous at $\boldsymbol{\theta}^{\lambda}$ and bounded in norm on a fixed neighborhood by a random envelope $B$ with $\EE[B^2]<\infty$.
    For every $\bar{\boldsymbol{a}}_n\to_p\boldsymbol{\theta}^{\lambda}$, the estimator in \eqref{eq:K-centered} satisfies $\hat{\boldsymbol{K}}\to_p\var\{\boldsymbol{\varphi}(\boldsymbol{\theta}^{\lambda})\}=\boldsymbol{K}^{\ast}$.
\end{lemma}
The envelope condition holds, for example, if $\boldsymbol{\varphi}$ is Lipschitz on a ball of radius $\delta_0$ about $\boldsymbol{\theta}^{\lambda}$ with a random constant $L$ such that $\EE[L^2]<\infty$ and $\EE[\|\boldsymbol{\varphi}(\boldsymbol{\theta}^{\lambda})\|^2]<\infty$, because $B=\|\boldsymbol{\varphi}(\boldsymbol{\theta}^\lambda)\|+\delta_0L$ then has finite second moment.

\begin{proof}[Proof of Lemma~\ref{lem:K-consistency}]
    Write $N=N_n$, and let $B_i$ be the envelope of $\boldsymbol{\varphi}_i$ on $\{\boldsymbol{a}:\|\boldsymbol{a}-\boldsymbol{\theta}^\lambda\|\le\delta_1\}$ with $\EE[B_i^2]<\infty$.
    For $0<\delta\le\delta_1$, define
    \begin{equation*}
        \omega_i(\delta)
        =
        \sup_{\|\boldsymbol{a}-\boldsymbol{\theta}^\lambda\|\le\delta}
        \|\boldsymbol{\varphi}_i(\boldsymbol{a})-\boldsymbol{\varphi}_i(\boldsymbol{\theta}^\lambda)\|
        .
    \end{equation*}
    Since $0\le\omega_i(\delta)\le2B_i$ and $\EE[(1+2B_i)B_i]<\infty$, almost sure continuity of $\boldsymbol{\varphi}_i$ at $\boldsymbol{\theta}^\lambda$ and dominated convergence give, as $\delta\downarrow0$,
    \begin{equation*}
        \omega_i(\delta)
        \longrightarrow
        0
        \quad\text{almost surely}
        ,\quad
        \EE[(1+2B_i)\omega_i(\delta)]
        \longrightarrow
        0
        .
    \end{equation*}
    On $\{\|\bar{\boldsymbol{a}}_n-\boldsymbol{\theta}^\lambda\|\le\delta\}$, the triangle inequality gives
    \begin{equation*}
        \Biggl\|
            N^{-1}\sum_{i=1}^N
            \{\boldsymbol{\varphi}_i(\bar{\boldsymbol{a}}_n)-\boldsymbol{\varphi}_i(\boldsymbol{\theta}^\lambda)\}
        \Biggr\|
        \le
        N^{-1}\sum_{i=1}^N\omega_i(\delta)
        ,
    \end{equation*}
    \begin{equation*}
        \begin{split}
            &\Biggl\|
                N^{-1}\sum_{i=1}^N
                \{\boldsymbol{\varphi}_i(\bar{\boldsymbol{a}}_n)\boldsymbol{\varphi}_i(\bar{\boldsymbol{a}}_n)^\top
                -\boldsymbol{\varphi}_i(\boldsymbol{\theta}^\lambda)\boldsymbol{\varphi}_i(\boldsymbol{\theta}^\lambda)^\top\}
            \Biggr\|
            \\
            &\quad\le
            N^{-1}\sum_{i=1}^N
            \{\|\boldsymbol{\varphi}_i(\bar{\boldsymbol{a}}_n)\|+\|\boldsymbol{\varphi}_i(\boldsymbol{\theta}^\lambda)\|\}
            \|\boldsymbol{\varphi}_i(\bar{\boldsymbol{a}}_n)-\boldsymbol{\varphi}_i(\boldsymbol{\theta}^\lambda)\|
            \\
            &\quad\le
            N^{-1}\sum_{i=1}^N2B_i\omega_i(\delta)
            .
        \end{split}
    \end{equation*}
    Hence, for every $\varepsilon>0$, Markov's inequality bounds the probability that the sum of these two empirical-moment differences exceeds $\varepsilon$ by
    \begin{equation*}
        P^\ast\{\|\bar{\boldsymbol{a}}_n-\boldsymbol{\theta}^\lambda\|>\delta\}
        +
        \varepsilon^{-1}\EE[(1+2B_1)\omega_1(\delta)]
        .
    \end{equation*}
    Taking $n\to\infty$ and then $\delta\downarrow0$ proves that both differences are $o_p(1)$.
    Writing $\bar{\boldsymbol{\varphi}}=N^{-1}\sum_{i=1}^N\boldsymbol{\varphi}_i(\bar{\boldsymbol{a}}_n)$, the law of large numbers at $\boldsymbol{\theta}^\lambda$ therefore gives
    \begin{equation*}
        \begin{split}
            \hat{\boldsymbol{K}}
            &=
            N^{-1}\sum_{i=1}^N
            \boldsymbol{\varphi}_i(\bar{\boldsymbol{a}}_n)\boldsymbol{\varphi}_i(\bar{\boldsymbol{a}}_n)^\top
            -
            \bar{\boldsymbol{\varphi}}\bar{\boldsymbol{\varphi}}^\top
            \\
            &=
            N^{-1}\sum_{i=1}^N
            \boldsymbol{\varphi}_i(\boldsymbol{\theta}^\lambda)\boldsymbol{\varphi}_i(\boldsymbol{\theta}^\lambda)^\top
            -
            \Biggl\{N^{-1}\sum_{i=1}^N\boldsymbol{\varphi}_i(\boldsymbol{\theta}^\lambda)\Biggr\}
            \Biggl\{N^{-1}\sum_{i=1}^N\boldsymbol{\varphi}_i(\boldsymbol{\theta}^\lambda)\Biggr\}^\top
            +
            o_p(1)
            \\
            &\to_p
            \EE[\boldsymbol{\varphi}(\boldsymbol{\theta}^\lambda)
                \boldsymbol{\varphi}(\boldsymbol{\theta}^\lambda)^\top]
            -
            \EE[\boldsymbol{\varphi}(\boldsymbol{\theta}^\lambda)]
            \EE[\boldsymbol{\varphi}(\boldsymbol{\theta}^\lambda)]^\top
            \\
            &=
            \var\{\boldsymbol{\varphi}(\boldsymbol{\theta}^\lambda)\}
            =
            \boldsymbol{K}^\ast
            ,
        \end{split}
    \end{equation*}
    where the final equality follows from the i.i.d.\ score central limit theorem.
\end{proof}

\subsection{Feasible calibration and reuse of the simulated draws}
This subsection proves Proposition~\ref{prop:estimated-calib} and shows that the pilot draws themselves can be calibrated.

\begin{proof}[Proof of Proposition~\ref{prop:estimated-calib}]
    Put $\boldsymbol{H}_0=\eta_n\boldsymbol{J}_\lambda^\ast$ and $\boldsymbol{\Omega}_n=(\boldsymbol{V}_{\mathrm{target}}^\ast)^{1/2}\boldsymbol{H}_0^{1/2}$.
    Lemmas~\ref{lem:postvar-consistency} and~\ref{lem:target-plugin} give
    \begin{equation*}
        \begin{gathered}
            \boldsymbol{\delta}_n
            \coloneq
            \sqrt{s_n}(\hat{\boldsymbol{\theta}}_{\mathrm{GB}}-\check{\boldsymbol{\theta}}_n)
            =
            o_p(1)
            ,\\
            \|s_n\hat{\boldsymbol{\Sigma}}_{\mathrm{post}}-\boldsymbol{H}_0^{-1}\|
            \to_p
            0
            ,\quad
            \hat{\boldsymbol{V}}_{\mathrm{target}}
            \to_p
            \boldsymbol{V}_{\mathrm{target}}^\ast
            .
        \end{gathered}
    \end{equation*}
    Since $\eta_n\in\cE$, the matrices $\boldsymbol{H}_0^{-1}$ range over a compact subset of the positive-definite cone, and $\hat{\boldsymbol{V}}_{\mathrm{target}}$ and $s_n\hat{\boldsymbol{\Sigma}}_{\mathrm{post}}$ are positive definite with probability tending to one.
    On this event, uniform continuity of matrix inversion and the symmetric square root on compact subsets of the positive-definite cone gives
    \begin{equation*}
        \hat{\boldsymbol{\Omega}}
        =
        \hat{\boldsymbol{V}}_{\mathrm{target}}^{1/2}
        (s_n\hat{\boldsymbol{\Sigma}}_{\mathrm{post}})^{-1/2}
        =
        \boldsymbol{\Omega}_n+o_p(1)
        .
    \end{equation*}
    Since $\boldsymbol{\Omega}_n\boldsymbol{H}_0^{-1}\boldsymbol{\Omega}_n^\top=\boldsymbol{V}_{\mathrm{target}}^\ast$ and $\|\boldsymbol{H}_0^{-1}\|$ and $\|\boldsymbol{\Omega}_n\|$ are bounded for $\eta_n\in\cE$, we have
    \begin{equation*}
        \begin{split}
            \|\hat{\boldsymbol{\Omega}}\boldsymbol{H}_0^{-1}\hat{\boldsymbol{\Omega}}^\top
                -\boldsymbol{V}_{\mathrm{target}}^\ast\|
            &=
            \|(\hat{\boldsymbol{\Omega}}-\boldsymbol{\Omega}_n)\boldsymbol{H}_0^{-1}\hat{\boldsymbol{\Omega}}^\top
            +\boldsymbol{\Omega}_n\boldsymbol{H}_0^{-1}(\hat{\boldsymbol{\Omega}}-\boldsymbol{\Omega}_n)^\top\|
            \\
            &\le
            \|\hat{\boldsymbol{\Omega}}-\boldsymbol{\Omega}_n\|
            \|\boldsymbol{H}_0^{-1}\|
            \{\|\hat{\boldsymbol{\Omega}}\|+\|\boldsymbol{\Omega}_n\|\}
            =
            o_p(1)
            .
        \end{split}
    \end{equation*}

    Conditional independence of the new draw implies $\operatorname{Law}\{\sqrt{s_n}(\boldsymbol{\theta}-\check{\boldsymbol{\theta}}_n)\mid\cW_n\}=Q_n^{\check{\boldsymbol{\theta}}_n,\eta_n}$.
    Since $\sqrt{s_n}(\hat{\boldsymbol{\theta}}_{\mathrm{calib}}-\bar{\boldsymbol{\theta}}_n)=\hat{\boldsymbol{\Omega}}\sqrt{s_n}(\boldsymbol{\theta}-\hat{\boldsymbol{\theta}}_{\mathrm{GB}})$ and $\sqrt{s_n}(\boldsymbol{\theta}-\hat{\boldsymbol{\theta}}_{\mathrm{GB}})=\sqrt{s_n}(\boldsymbol{\theta}-\check{\boldsymbol{\theta}}_n)-\boldsymbol{\delta}_n$, Lemmas~\ref{supp-lem:tv-tools}--\ref{supp-lem:gaussian} and Proposition~\ref{prop:bvm-baseline} yield, on the event where $\hat{\boldsymbol{\Omega}}$ is nonsingular,
    \begin{equation*}
        \begin{split}
            &\Bigl\|
                \operatorname{Law}\{\sqrt{s_n}(\hat{\boldsymbol{\theta}}_{\mathrm{calib}}-\bar{\boldsymbol{\theta}}_n)\mid\cW_n\}
                -N(\boldsymbol{0},\boldsymbol{V}_{\mathrm{target}}^\ast)
            \Bigr\|_{\TV}
            \\
            &\quad\le
            \Bigl\|
                \operatorname{Law}\{\sqrt{s_n}(\boldsymbol{\theta}-\hat{\boldsymbol{\theta}}_{\mathrm{GB}})\mid\cW_n\}
                -N(\boldsymbol{0},\boldsymbol{H}_0^{-1})
            \Bigr\|_{\TV}
            \\
            &\qquad+
            \bigl\|
                N(\boldsymbol{0},\hat{\boldsymbol{\Omega}}\boldsymbol{H}_0^{-1}\hat{\boldsymbol{\Omega}}^\top)
                -N(\boldsymbol{0},\boldsymbol{V}_{\mathrm{target}}^\ast)
            \bigr\|_{\TV}
            \\
            &\quad\le
            \|Q_n^{\check{\boldsymbol{\theta}}_n,\eta_n}
                -N(\boldsymbol{0},\boldsymbol{H}_0^{-1})\|_{\TV}
            +
            \frac{\|\boldsymbol{H}_0^{1/2}\boldsymbol{\delta}_n\|}{\sqrt{2\pi}}
            +
            \bigl\|
                N(\boldsymbol{0},\hat{\boldsymbol{\Omega}}\boldsymbol{H}_0^{-1}\hat{\boldsymbol{\Omega}}^\top)
                -N(\boldsymbol{0},\boldsymbol{V}_{\mathrm{target}}^\ast)
            \bigr\|_{\TV}
            \\
            &\quad\to_p
            0
            .
        \end{split}
    \end{equation*}
\end{proof}

The coverage statement following Proposition~\ref{prop:estimated-calib} holds because the proofs of Corollaries~\ref{cor:coverage} and~\ref{cor:functionals} apply with this conditional law, the consistent target covariance, and an admissible output center.
The sampling expansion for such a center still holds under the joint law, since its difference from the data-based baseline is $o_p(s_n^{-1/2})$.

We next show that the pilot draws themselves can be transformed.
This conclusion uses a different mode of convergence from Proposition~\ref{prop:estimated-calib}, because a finite empirical distribution has total-variation distance one from a nonsingular Gaussian.
Let $\delta_{\boldsymbol{x}}$ denote unit point mass at $\boldsymbol{x}$ and, for the draws in \eqref{eq:empirical-lsc}, define
\begin{equation*}
    \hat{P}_n^{\mathrm{calib}}
    =
    D_n^{-1}\sum_{d=1}^{D_n}
    \delta_{\sqrt{s_n}(\hat{\boldsymbol{\theta}}_{\mathrm{calib}}^{(d)}
    -\bar{\boldsymbol{\theta}}_n)}
    .
\end{equation*}
For the next lemma, consider one or a fixed number of conditionally independent stationary reversible pilot chains, with mean-zero $L^2(\Pi_n^\eta)$ operator norms bounded by $\varrho_n<1$ and total retained length $D_n$ satisfying $D_n(1-\varrho_n)\to_p\infty$.
A nonstationary pilot path is also allowed if its joint conditional law is within $o_p(1)$ in total variation of such stationary paths.

\begin{lemma}[Reusing the simulated draws]
\label{lem:empirical-output}
    Suppose Assumptions~\ref{ass:regularity}--\ref{ass:baseline}, \ref{ass:UI}, and~\ref{ass:MC} hold, the pilot paths satisfy the conditions above, $\bar{\boldsymbol{a}}_n\to_p\boldsymbol{\theta}^\lambda$, the symmetric estimator $\hat{\boldsymbol{K}}\to_p\boldsymbol{K}^\ast$, and $\bar{\boldsymbol{\theta}}_n$ is admissible.
    Then the bounded-Lipschitz distance from $\hat{P}_n^{\mathrm{calib}}$ to $N(\boldsymbol{0},\boldsymbol{V}_{\mathrm{target}}^\ast)$ tends to zero in probability.
\end{lemma}

\begin{proof}[Proof of Lemma~\ref{lem:empirical-output}]
    Put $\boldsymbol{X}_{n,d}=\sqrt{s_n}(\boldsymbol{\theta}^{(d)}-\check{\boldsymbol{\theta}}_n)$, $\hat{Q}_n=D_n^{-1}\sum_{d=1}^{D_n}\delta_{\boldsymbol{X}_{n,d}}$, and $\nu_n=N(\boldsymbol{0},\boldsymbol{H}_0^{-1})$, where $\boldsymbol{H}_0=\eta_n\boldsymbol{J}_\lambda^\ast$.
    For one or a fixed number of conditionally independent stationary pilot chains, since cross-chain covariances vanish, applying \eqref{eq:mc-variance} to $\boldsymbol{\theta}\mapsto f\{\sqrt{s_n}(\boldsymbol{\theta}-\check{\boldsymbol{\theta}}_n)\}$ gives, for every deterministic bounded measurable $f$,
    \begin{equation*}
        \EE\Biggl[
            \Biggl\{\int f\dd\hat{Q}_n-\int f\dd Q_n^{\check{\boldsymbol{\theta}}_n,\eta_n}\Biggr\}^2
            \Bigm|\cD_n
        \Biggr]
        \le
        \frac{2\|f\|_\infty^2}{D_n(1-\varrho_n)}
        \to_p
        0
        .
    \end{equation*}
    Conditional Markov's inequality and Proposition~\ref{prop:bvm-baseline} therefore imply
    \begin{equation*}
        \begin{split}
            \Biggl|
                \int f\dd(\hat{Q}_n-\nu_n)
            \Biggr|
            &\le
            \Biggl|
                \int f\dd\bigl(
                    \hat{Q}_n-Q_n^{\check{\boldsymbol{\theta}}_n,\eta_n}
                \bigr)
            \Biggr|
            +
            2\|f\|_\infty
            \|Q_n^{\check{\boldsymbol{\theta}}_n,\eta_n}-\nu_n\|_{\TV}
            \\
            &=
            o_p(1)
            .
        \end{split}
    \end{equation*}
    Under the stated nonstationary approximation, the conditional probabilities of these error events change by at most the $o_p(1)$ path-law distance, so the same convergence holds.

    Let $\cH=\{f:\|f\|_\infty\le1,\operatorname{Lip}(f)\le1\}$.
    For fixed $M<\infty$ and $\varepsilon>0$, because functions in $\cH$ are uniformly bounded and equicontinuous on the compact ball $\{\|\boldsymbol{x}\|\le M\}$, there exist $f_1,\ldots,f_J\in\cH$ such that
    \begin{equation*}
        \sup_{f\in\cH}\min_{1\le j\le J}
        \sup_{\|\boldsymbol{x}\|\le M}|f(\boldsymbol{x})-f_j(\boldsymbol{x})|
        \le
        \varepsilon
        .
    \end{equation*}
    Splitting each integral over $\{\|\boldsymbol{x}\|\le M\}$ and $\{\|\boldsymbol{x}\|>M\}$ and applying the preceding convergence to $f_1,\ldots,f_J$ and to $\1\{\|\boldsymbol{x}\|>M\}$ give
    \begin{equation*}
        \begin{split}
            d_{\mathrm{BL}}(\hat{Q}_n,\nu_n)
            &\le
            \max_{1\le j\le J}\Biggl|
                \int f_j\dd(\hat{Q}_n-\nu_n)
            \Biggr|
            +
            2\varepsilon+2\hat{Q}_n(\|\boldsymbol{x}\|>M)+2\nu_n(\|\boldsymbol{x}\|>M)
            \\
            &\le
            2\varepsilon+4\nu_n(\|\boldsymbol{x}\|>M)+o_p(1)
            .
        \end{split}
    \end{equation*}
    Since $\sup_n\nu_n(\|\boldsymbol{x}\|>M)\to0$ as $M\to\infty$, first taking $n\to\infty$, then $M\to\infty$ and $\varepsilon\downarrow0$, proves $d_{\mathrm{BL}}(\hat{Q}_n,\nu_n)\to_p0$.

    For the estimated map, put $\boldsymbol{\delta}_n=\sqrt{s_n}(\hat{\boldsymbol{\theta}}_{\mathrm{GB}}-\check{\boldsymbol{\theta}}_n)$ and $\boldsymbol{\Omega}_n=(\boldsymbol{V}_{\mathrm{target}}^\ast)^{1/2}\boldsymbol{H}_0^{1/2}$.
    Lemma~\ref{lem:postvar-consistency} and the proof of Proposition~\ref{prop:estimated-calib} give $\boldsymbol{\delta}_n=o_p(1)$ and $\|\hat{\boldsymbol{\Omega}}-\boldsymbol{\Omega}_n\|=o_p(1)$.
    The sample-moment identity and the Cauchy--Schwarz inequality yield
    \begin{equation*}
        \begin{split}
            D_n^{-1}\sum_{d=1}^{D_n}\|\boldsymbol{X}_{n,d}\|^2
            &=
            s_n\tr(\hat{\boldsymbol{\Sigma}}_{\mathrm{post}})+\|\boldsymbol{\delta}_n\|^2
            =
            O_p(1)
            ,\\
            D_n^{-1}\sum_{d=1}^{D_n}\|\boldsymbol{X}_{n,d}\|
            &\le
            \Biggl(
                D_n^{-1}\sum_{d=1}^{D_n}\|\boldsymbol{X}_{n,d}\|^2
            \Biggr)^{1/2}
            =
            O_p(1)
            .
        \end{split}
    \end{equation*}
    Choose a fixed $L\ge1$ with $\|\boldsymbol{\Omega}_n\|\le L$ for all $\eta_n\in\cE$.
    Since $\sqrt{s_n}(\hat{\boldsymbol{\theta}}_{\mathrm{calib}}^{(d)}-\bar{\boldsymbol{\theta}}_n)=\hat{\boldsymbol{\Omega}}(\boldsymbol{X}_{n,d}-\boldsymbol{\delta}_n)$, $\boldsymbol{x}\mapsto\boldsymbol{\Omega}_n\boldsymbol{x}$ maps $\nu_n$ to $N(\boldsymbol{0},\boldsymbol{V}_{\mathrm{target}}^\ast)$, and $f(\boldsymbol{\Omega}_n\boldsymbol{x})/L\in\cH$ for every $f\in\cH$, the triangle inequality and Lipschitz bound give
    \begin{equation*}
        \begin{split}
            &d_{\mathrm{BL}}\{\hat{P}_n^{\mathrm{calib}},N(\boldsymbol{0},\boldsymbol{V}_{\mathrm{target}}^\ast)\}
            \\
            &\quad\le
            d_{\mathrm{BL}}\Biggl\{
                \hat{P}_n^{\mathrm{calib}},
                D_n^{-1}\sum_{d=1}^{D_n}\delta_{\boldsymbol{\Omega}_n\boldsymbol{X}_{n,d}}
            \Biggr\}
            +
            d_{\mathrm{BL}}\Biggl\{
                D_n^{-1}\sum_{d=1}^{D_n}\delta_{\boldsymbol{\Omega}_n\boldsymbol{X}_{n,d}},
                N(\boldsymbol{0},\boldsymbol{V}_{\mathrm{target}}^\ast)
            \Biggr\}
            \\
            &\quad\le
            \|\hat{\boldsymbol{\Omega}}-\boldsymbol{\Omega}_n\|
            D_n^{-1}\sum_{d=1}^{D_n}\|\boldsymbol{X}_{n,d}\|
            +
            \|\hat{\boldsymbol{\Omega}}\|\|\boldsymbol{\delta}_n\|
            +
            Ld_{\mathrm{BL}}(\hat{Q}_n,\nu_n)
            \\
            &\quad=
            o_p(1)
            .
        \end{split}
    \end{equation*}
\end{proof}

Using Lemma~\ref{lem:empirical-output}, we obtain the following result for quantiles and intervals.
\begin{corollary}[Quantiles and intervals from the simulated draws]
\label{cor:empirical-quantiles}
    Under the conditions of Lemma~\ref{lem:empirical-output}, each scalar marginal empirical quantile of $\hat{P}_n^{\mathrm{calib}}$ at a fixed probability in $(0,1)$ converges to its normal limit, and the resulting marginal equal-tailed intervals have their nominal limiting coverage.
    The same conclusions hold for scalar $C^1$ functionals with nonzero derivative at the target, using their local scaling.
\end{corollary}

\begin{proof}[Proof of Corollary~\ref{cor:empirical-quantiles}]
    Fix a coordinate $j$, let $\hat{F}_n$ be the $j$th marginal distribution function of $\hat{P}_n^{\mathrm{calib}}$, and let $F(x)=\Phi(x/\sigma_j)$, where $\sigma_j^2=(\boldsymbol{V}_{\mathrm{target}}^\ast)_{jj}>0$.
    Since $F$ is continuous, Lemma~\ref{lem:empirical-output} gives $\hat{F}_n(x)\to_p F(x)$ for every fixed $x\in\RR$.
    For $0<t<1$, define $q_t=F^{-1}(t)$ and $\hat{q}_{n,t}=\inf\{x:\hat{F}_n(x)\ge t\}$.
    Since $F(q_t-\varepsilon)<t<F(q_t+\varepsilon)$, for every $\varepsilon>0$,
    \begin{equation*}
        P^\ast(|\hat{q}_{n,t}-q_t|>\varepsilon)
        \le
        P^\ast\{\hat{F}_n(q_t-\varepsilon)\ge t\}
        +
        P^\ast\{\hat{F}_n(q_t+\varepsilon)<t\}
        \longrightarrow
        0
        .
    \end{equation*}
    Write $\bar{\theta}_{n,j}$ and $\theta_j^\lambda$ for the $j$th coordinates of the output center and target.
    For $0<\alpha<1$, the empirical equal-tailed interval is
    \begin{equation*}
        C_{n,j}
        =
        \bigl[
            \bar{\theta}_{n,j}+s_n^{-1/2}\hat{q}_{n,\alpha/2},
            \bar{\theta}_{n,j}+s_n^{-1/2}\hat{q}_{n,1-\alpha/2}
        \bigr]
        .
    \end{equation*}
    Admissibility of $\bar{\boldsymbol{\theta}}_n$, Proposition~\ref{prop:baseline-clt}, and symmetry of $N(0,\sigma_j^2)$ imply $\sqrt{s_n}(\theta_j^\lambda-\bar{\theta}_{n,j})\to_d N(0,\sigma_j^2)$.
    Thus Slutsky's lemma and continuity of $F$ give
    \begin{equation*}
        \begin{split}
            P^\ast(\theta_j^\lambda\in C_{n,j})
            &=
            P^\ast\bigl\{
                \hat{q}_{n,\alpha/2}
                \le\sqrt{s_n}(\theta_j^\lambda-\bar{\theta}_{n,j})
                \le\hat{q}_{n,1-\alpha/2}
            \bigr\}
            \\
            &\longrightarrow
            F(q_{1-\alpha/2})-F(q_{\alpha/2})
            =
            1-\alpha
            .
        \end{split}
    \end{equation*}
    For a scalar $C^1$ functional $g$ as in Corollary~\ref{cor:functionals}, put $\boldsymbol{\Gamma}^\ast=\nabla g(\boldsymbol{\theta}^\lambda)$ and $T_n(\boldsymbol{x})=\sqrt{s_n}\{g(\bar{\boldsymbol{\theta}}_n+s_n^{-1/2}\boldsymbol{x})-g(\bar{\boldsymbol{\theta}}_n)\}$.
    For every fixed $M<\infty$, since $\bar{\boldsymbol{\theta}}_n\to_p\boldsymbol{\theta}^\lambda$, the ball $\{\boldsymbol{y}:\|\boldsymbol{y}-\bar{\boldsymbol{\theta}}_n\|\le M/\sqrt{s_n}\}$ lies in the neighborhood where $g$ is $C^1$ with probability tending to one.
    On this event, the fundamental theorem of calculus and continuity of $\nabla g$ at $\boldsymbol{\theta}^\lambda$ give
    \begin{equation*}
        \begin{split}
            \sup_{\|\boldsymbol{x}\|\le M}|T_n(\boldsymbol{x})-\boldsymbol{\Gamma}^\ast\boldsymbol{x}|
            &=
            \sup_{\|\boldsymbol{x}\|\le M}
            \Biggl|
                \int_0^1\{\nabla g(\bar{\boldsymbol{\theta}}_n+ts_n^{-1/2}\boldsymbol{x})-\boldsymbol{\Gamma}^\ast\}\boldsymbol{x}\dd t
            \Biggr|
            \\
            &\le
            M\sup_{\|\boldsymbol{y}-\bar{\boldsymbol{\theta}}_n\|\le M/\sqrt{s_n}}
            \|\nabla g(\boldsymbol{y})-\boldsymbol{\Gamma}^\ast\|
            \to_p
            0
            .
        \end{split}
    \end{equation*}
    Combining this bound with the tightness of $\hat{P}_n^{\mathrm{calib}}$ in Lemma~\ref{lem:empirical-output}, as in the proof of Corollary~\ref{cor:functionals}, yields
    \begin{equation*}
        d_{\mathrm{BL}}\bigl\{
            (T_n)_\#\hat{P}_n^{\mathrm{calib}},
            N(0,\boldsymbol{\Gamma}^\ast\boldsymbol{V}_{\mathrm{target}}^\ast(\boldsymbol{\Gamma}^\ast)^\top)
        \bigr\}
        \to_p
        0
        .
    \end{equation*}
    Since $\boldsymbol{\Gamma}^\ast\ne\boldsymbol{0}$, the limiting variance $\boldsymbol{\Gamma}^\ast\boldsymbol{V}_{\mathrm{target}}^\ast(\boldsymbol{\Gamma}^\ast)^\top$ is positive, so the same quantile-bracketing argument and the delta method for $g(\bar{\boldsymbol{\theta}}_n)$ prove the functional assertions.
\end{proof}

\section{Extension to non-smooth penalties}
\label{supp:non-smooth}

This section extends the calibration to non-smooth penalties.
A non-smooth penalty can produce different posterior concentration rates in different coordinates.
We consider penalties that are smooth in one block and positively homogeneous in a block whose target value is zero.
The first block has a Gaussian limit at scale $s_n^{-1/2}$, while the second concentrates at scale $s_n^{-1}$ under the strict complementarity condition below.
The resulting joint approximation justifies calibration of the smooth marginal while accounting for posterior mass in the other coordinates.
Non-smooth penalized estimators have related nonregular limits \citep{knight2000asymptotics}, whereas here the object of study is the full posterior under the penalty scaling in \eqref{eq:posterior}.

We first fix the notation for the two blocks of coordinates.
After permuting the coordinates, write
\begin{equation*}
    \boldsymbol{\theta}
    =
    (\boldsymbol{\alpha}^{\top},\boldsymbol{\gamma}^{\top})^{\top}
    ,\quad
    \boldsymbol{\theta}^{\lambda}
    =
    (\boldsymbol{\alpha}_0^{\top},\boldsymbol{0}^{\top})^{\top}
    ,\quad
    \boldsymbol{\alpha}
    \in
    \RR^q
    ,\quad
    \boldsymbol{\gamma}
    \in\RR^r
    ,\quad
    q+r
    =
    p
    .
\end{equation*}
We assume $q,r\ge1$, since the smooth theory applies if $r=0$.
We call the smooth block $\boldsymbol{\alpha}$ active and the zero block $\boldsymbol{\gamma}$ inactive, with index sets $A$ and $I=A^c$.
Subscripts $A$ and $I$ denote the corresponding vector subblocks, and $AA$, $AI$, and $II$ denote the matrix blocks.
We retain $G_n=\bar{M}_n+\lambda_n\rho$ and the posterior family \eqref{eq:posterior} on the fixed interval $\cE=[\eta_L,\eta_U]\subset(0,\infty)$.

We impose the following conditions on the penalty and the population score.
\begin{assumption}[Local penalty structure and strict complementarity]
\label{ns:penalty}
    \leavevmode
    \begin{enumerate}[label=(\roman*),leftmargin=*]
        \item\label{ns:penalty:structure}
        The target $\boldsymbol{\theta}^{\lambda}$ is interior to $\Theta$.
        The penalty is convex and finite on a product neighborhood $\cN_A\times\cN_I$ of $(\boldsymbol{\alpha}_0,\boldsymbol{0})$, and on this neighborhood
        \begin{equation}
        \label{ns:decomposition}
            \rho(\boldsymbol{\alpha},\boldsymbol{\gamma})
            =
            \rho_A(\boldsymbol{\alpha})
            +
            \rho_I(\boldsymbol{\gamma})
            .
        \end{equation}
        Here $\rho_A$ is convex and $C^2$, and $\rho_I:\RR^r\to[0,\infty)$ is continuous, convex, and positively homogeneous of degree one, that is, $\rho_I(t\boldsymbol{v})=t\rho_I(\boldsymbol{v})$ for $t\ge0$.

        \item\label{ns:penalty:population}
        The regularization parameters satisfy $\lambda_n\to\lambda>0$ and $s_n^{1/2}(\lambda_n-\lambda)\to0$.
        For the population score $\Psi$,
        \begin{equation}
        \label{ns:population}
            \Psi_A(\boldsymbol{\theta}^{\lambda})
            +
            \lambda\nabla\rho_A(\boldsymbol{\alpha}_0)
            =
            \boldsymbol{0}
            ,\quad
            \Psi_I(\boldsymbol{\theta}^{\lambda})^{\top}\boldsymbol{v}
            +
            \lambda\rho_I(\boldsymbol{v})
            \ge
            m_I\|\boldsymbol{v}\|
            \quad
            (\boldsymbol{v}\in\RR^r)
        \end{equation}
        for a constant $m_I>0$.

        \item\label{ns:penalty:curvature}
        \begin{equation*}
            \boldsymbol{J}_{\lambda,AA}^{\ast}
            \coloneq
            \boldsymbol{J}_{AA}^{\ast}
            +
            \lambda\nabla^2\rho_A(\boldsymbol{\alpha}_0)
            \succ
            0
            ,\quad
            \boldsymbol{J}^{\ast}
            \coloneq
            \nabla\Psi(\boldsymbol{\theta}^{\lambda})
            .
        \end{equation*}
    \end{enumerate}
\end{assumption}

The second inequality in \eqref{ns:population} defines strict complementarity here, under which the population criterion increases to first order in every nonzero inactive direction, with a common positive margin.
Continuity and homogeneity also give $\rho_I(\boldsymbol{v})\le C_I\|\boldsymbol{v}\|$, where $C_I=\max_{\|\boldsymbol{v}\|=1}\rho_I(\boldsymbol{v})<\infty$.
Separability permits ordinary differentiation in the active block while retaining the exact inactive penalty increment.

The next assumption concerns the loss, an active baseline estimator, and the posterior tails.
\begin{assumption}[Loss, active baseline, and posterior tails]
\label{ns:loss}
    There is a bounded open convex product neighborhood $\cN$ of $\boldsymbol{\theta}^{\lambda}$, with closure inside $\Theta$ and inside the neighborhood of \eqref{ns:decomposition}, on which the following hold.
    \begin{enumerate}[label=(\roman*),leftmargin=*]
        \item\label{ns:loss:derivatives}
        With probability tending to one, $\bar{M}_n$ is $C^2$ on $\cN$.
        The population score $\Psi$ is $C^1$, its derivative $\boldsymbol{J}(\boldsymbol{\theta})=\nabla\Psi(\boldsymbol{\theta})$ is symmetric, and, for every compact $\cK\subset\cN$,
        \begin{equation*}
            \sup_{\boldsymbol{\theta}\in\cK}\|\boldsymbol{U}_n(\boldsymbol{\theta})-\Psi(\boldsymbol{\theta})\|
            \to_p
            0
            ,\quad
            \sup_{\boldsymbol{\theta}\in\cK}\|\boldsymbol{J}_n(\boldsymbol{\theta})-\boldsymbol{J}(\boldsymbol{\theta})\|
            \to_p
            0
            .
        \end{equation*}

        \item\label{ns:loss:clt}
        The active score satisfies
        \begin{equation*}
            s_n^{1/2}\{\boldsymbol{U}_{n,A}(\boldsymbol{\theta}^{\lambda})-\Psi_A(\boldsymbol{\theta}^{\lambda})\}
            \to_d
            N(\boldsymbol{0},\boldsymbol{K}_{AA}^{\ast})
            ,\quad
            \boldsymbol{K}_{AA}^{\ast}
            \succ
            0
            .
        \end{equation*}

        \item\label{ns:loss:baseline}
        There is a data-measurable $\check{\boldsymbol{\alpha}}_n\to_p\boldsymbol{\alpha}_0$ such that
        \begin{equation*}
            \boldsymbol{F}_{n,A}(\check{\boldsymbol{\alpha}}_n)
            =
            o_p(s_n^{-1/2})
            ,\quad
            \boldsymbol{F}_{n,A}(\boldsymbol{a})
            \coloneq
            \boldsymbol{U}_{n,A}(\boldsymbol{a},\boldsymbol{0})
            +
            \lambda_n\nabla\rho_A(\boldsymbol{a})
            .
        \end{equation*}

        \item\label{ns:loss:posterior}
        The properness, local Lipschitz, and integrated separation conditions in Assumption~\ref{ass:posterior} \ref{ass:posterior:proper}--\ref{ass:posterior:separation} hold for this $G_n$, this target, and every $\eta\in\cE$.
    \end{enumerate}
\end{assumption}

These conditions concern a smooth empirical loss and the locally separable penalty \eqref{ns:decomposition}.
The active estimator has sandwich covariance
\begin{equation*}
    \boldsymbol{V}_{\mathrm{target},AA}^{\ast}
    =
    (\boldsymbol{J}_{\lambda,AA}^{\ast})^{-1}\boldsymbol{K}_{AA}^{\ast}(\boldsymbol{J}_{\lambda,AA}^{\ast})^{-1}
    .
\end{equation*}
An \emph{active admissible center} is a measurable sequence $\tilde{\boldsymbol{\alpha}}_n$ satisfying $s_n^{1/2}\|\tilde{\boldsymbol{\alpha}}_n-\check{\boldsymbol{\alpha}}_n\|\to_p0$, under the joint data and simulation law when the center uses Monte Carlo output.
The following lemma gives the sampling distribution of the active baseline estimator.
\begin{lemma}[Active baseline expansion]
\label{ns:estimator}
    Under Assumptions~\ref{ns:penalty} and~\ref{ns:loss},
    \begin{equation*}
        \begin{split}
            s_n^{1/2}(\check{\boldsymbol{\alpha}}_n-\boldsymbol{\alpha}_0)
            &=
            -(\boldsymbol{J}_{\lambda,AA}^{\ast})^{-1}s_n^{1/2}\{\boldsymbol{U}_{n,A}(\boldsymbol{\theta}^{\lambda})-\Psi_A(\boldsymbol{\theta}^{\lambda})\}
            +
            o_p(1)
            \\
            &\to_d
            N(\boldsymbol{0},\boldsymbol{V}_{\mathrm{target},AA}^{\ast})
            .
        \end{split}
    \end{equation*}
    Every active admissible center has the same expansion, and any two active baselines differ by $o_p(s_n^{-1/2})$.
\end{lemma}

\begin{proof}[Proof of Lemma~\ref{ns:estimator}]
    Since $\check{\boldsymbol{\alpha}}_n\to_p\boldsymbol{\alpha}_0$ by Assumption~\ref{ns:loss}~\ref{ns:loss:baseline}, the line segment joining $\boldsymbol{\alpha}_0$ and $\check{\boldsymbol{\alpha}}_n$ lies in a fixed compact set $\cK_A$ with $\cK_A\times\{\boldsymbol{0}\}\subset\cN$ with probability tending to one.
    On this event, the fundamental theorem of calculus and Assumptions~\ref{ns:penalty}~\ref{ns:penalty:structure}--\ref{ns:penalty:population} and~\ref{ns:loss}~\ref{ns:loss:derivatives} give
    \begin{equation*}
        \begin{gathered}
            \boldsymbol{F}_{n,A}(\check{\boldsymbol{\alpha}}_n)
            -
            \boldsymbol{F}_{n,A}(\boldsymbol{\alpha}_0)
            =
            \boldsymbol{A}_{n,A}(\check{\boldsymbol{\alpha}}_n-\boldsymbol{\alpha}_0)
            ,\\
            \boldsymbol{A}_{n,A}
            \coloneq
            \int_0^1\bigl[
                \boldsymbol{J}_{n,AA}\{\boldsymbol{\alpha}_0+t(\check{\boldsymbol{\alpha}}_n-\boldsymbol{\alpha}_0),\boldsymbol{0}\}
                +
                \lambda_n\nabla^2\rho_A\{\boldsymbol{\alpha}_0+t(\check{\boldsymbol{\alpha}}_n-\boldsymbol{\alpha}_0)\}
            \bigr]\dd t
            \to_p
            \boldsymbol{J}_{\lambda,AA}^{\ast}
            .
        \end{gathered}
    \end{equation*}
    Assumption~\ref{ns:penalty}~\ref{ns:penalty:curvature} consequently gives $\boldsymbol{A}_{n,A}^{-1}\to_p(\boldsymbol{J}_{\lambda,AA}^{\ast})^{-1}$.
    Moreover, \eqref{ns:population} and Assumptions~\ref{ns:penalty}~\ref{ns:penalty:population} and~\ref{ns:loss}~\ref{ns:loss:clt} imply
    \begin{equation*}
        \begin{split}
            \sqrt{s_n}\boldsymbol{F}_{n,A}(\boldsymbol{\alpha}_0)
            &=
            \sqrt{s_n}\{\boldsymbol{U}_{n,A}(\boldsymbol{\theta}^{\lambda})
            -\Psi_A(\boldsymbol{\theta}^{\lambda})\}
            +
            \sqrt{s_n}(\lambda_n-\lambda)\nabla\rho_A(\boldsymbol{\alpha}_0)
            \\
            &=
            \sqrt{s_n}\{\boldsymbol{U}_{n,A}(\boldsymbol{\theta}^{\lambda})
            -\Psi_A(\boldsymbol{\theta}^{\lambda})\}
            +
            o(1)
            \\
            &\to_d
            N(\boldsymbol{0},\boldsymbol{K}_{AA}^{\ast})
            .
        \end{split}
    \end{equation*}
    Rearranging the segment identity and applying Assumption~\ref{ns:loss}~\ref{ns:loss:baseline} and Slutsky's lemma yield
    \begin{equation*}
        \begin{split}
            \sqrt{s_n}(\check{\boldsymbol{\alpha}}_n-\boldsymbol{\alpha}_0)
            &=
            \boldsymbol{A}_{n,A}^{-1}\{
                \sqrt{s_n}\boldsymbol{F}_{n,A}(\check{\boldsymbol{\alpha}}_n)
                -
                \sqrt{s_n}\boldsymbol{F}_{n,A}(\boldsymbol{\alpha}_0)
            \}
            \\
            &=
            -(\boldsymbol{J}_{\lambda,AA}^{\ast})^{-1}
            \sqrt{s_n}\{\boldsymbol{U}_{n,A}(\boldsymbol{\theta}^{\lambda})
            -\Psi_A(\boldsymbol{\theta}^{\lambda})\}
            +
            o_p(1)
            \\
            &\to_d
            N(\boldsymbol{0},\boldsymbol{V}_{\mathrm{target},AA}^{\ast})
            .
        \end{split}
    \end{equation*}
    Any two active baselines have the same leading term, so their difference is $o_p(s_n^{-1/2})$.
    For any active admissible center $\tilde{\boldsymbol{\alpha}}_n$,
    \begin{equation*}
        \sqrt{s_n}(\tilde{\boldsymbol{\alpha}}_n-\boldsymbol{\alpha}_0)
        =
        \sqrt{s_n}(\check{\boldsymbol{\alpha}}_n-\boldsymbol{\alpha}_0)
        +
        \sqrt{s_n}(\tilde{\boldsymbol{\alpha}}_n-\check{\boldsymbol{\alpha}}_n)
        =
        \sqrt{s_n}(\check{\boldsymbol{\alpha}}_n-\boldsymbol{\alpha}_0)+o_p(1)
        ,
    \end{equation*}
    which proves the remaining assertion.
\end{proof}

We next state the joint limit of the active and inactive coordinates.
Use a common data-measurable active admissible center $\tilde{\boldsymbol{\alpha}}_n$ throughout $\cE$.
Let $P_n^{\tilde{\boldsymbol{\alpha}}_n,\eta}$ denote the conditional law, under $\Pi_n^\eta$, of $\{s_n^{1/2}(\boldsymbol{\alpha}-\tilde{\boldsymbol{\alpha}}_n),s_n\boldsymbol{\gamma}\}$.
Let $P_I^\eta$ have density
\begin{equation*}
    p_I^{\eta}(\boldsymbol{v})
    =
    \frac
    {\exp[-\eta\{\Psi_I(\boldsymbol{\theta}^{\lambda})^{\top}\boldsymbol{v}+\lambda\rho_I(\boldsymbol{v})\}]}
    {\displaystyle\int_{\RR^r}\exp[-\eta\{\Psi_I(\boldsymbol{\theta}^{\lambda})^{\top}\boldsymbol{w}+\lambda\rho_I(\boldsymbol{w})\}]\dd\boldsymbol{w}}
    .
\end{equation*}
Then we obtain the following result.
\begin{proposition}[Joint active--inactive posterior limit]
\label{ns:bvm}
    Under Assumptions~\ref{ns:penalty} and~\ref{ns:loss},
    \begin{equation}
    \label{ns:joint-tv}
        \sup_{\eta\in\cE}\Bigl\|
            P_n^{\tilde{\boldsymbol{\alpha}}_n,\eta}-N\bigl\{
                \boldsymbol{0},(\eta\boldsymbol{J}_{\lambda,AA}^{\ast})^{-1}
            \bigr\}\otimes P_I^{\eta}
        \Bigr\|_{\TV}
        \to_p
        0
        ,
    \end{equation}
    where $\otimes$ denotes the product measure.
\end{proposition}

\begin{proof}[Proof of Proposition~\ref{ns:bvm}]
    Write
    \begin{equation*}
        \begin{gathered}
            \phi_I(\boldsymbol{v})
            =
            \Psi_I(\boldsymbol{\theta}^{\lambda})^{\top}\boldsymbol{v}
            +
            \lambda\rho_I(\boldsymbol{v})
            ,\\
            Z_I(\eta)
            =
            \int_{\RR^r}\exp\{-\eta\phi_I(\boldsymbol{v})\}\dd\boldsymbol{v}
            ,\quad
            B_I
            =
            \|\Psi_I(\boldsymbol{\theta}^{\lambda})\|+\lambda C_I
            .
        \end{gathered}
    \end{equation*}
    Assumption~\ref{ns:penalty}~\ref{ns:penalty:population} and the bound $\rho_I(\boldsymbol{v})\le C_I\|\boldsymbol{v}\|$ give $m_I\|\boldsymbol{v}\|\le\phi_I(\boldsymbol{v})\le B_I\|\boldsymbol{v}\|$ for all $\boldsymbol{v}\in\RR^r$.
    Letting $c_d$ denote the volume of the unit ball in $\RR^d$ and restricting the integral in $Z_I(\eta)$ to the unit ball for the lower bound yield
    \begin{equation*}
        0
        <
        c_r\exp(-\eta_U B_I)
        \le
        \inf_{\eta\in\cE}Z_I(\eta)
        \le
        \sup_{\eta\in\cE}Z_I(\eta)
        \le
        \int_{\RR^r}\exp(-\eta_Lm_I\|\boldsymbol{v}\|)\dd\boldsymbol{v}
        <
        \infty
        .
    \end{equation*}

    We first prove concentration near $\boldsymbol{\theta}^{\lambda}$.
    Choose a compact product neighborhood $\cK\subset\cN\cap B_{\delta_0}$ of $\boldsymbol{\theta}^{\lambda}$.
    Since $\sqrt{s_n}\|\boldsymbol{F}_{n,A}(\boldsymbol{\alpha}_0)\|=O_p(1)$ by the proof of Lemma~\ref{ns:estimator}, on the event where $\bar{M}_n$ is $C^2$ on $\cN$, Taylor expansion in the active coordinates, integration of the inactive score, and Assumption~\ref{ns:loss}~\ref{ns:loss:derivatives} give, for all sufficiently large $n$,
    \begin{equation*}
        \begin{split}
            &\sup_{\|\boldsymbol{u}\|,\|\boldsymbol{v}\|\le1}
            \bigl|
                s_n\{G_n(\boldsymbol{\alpha}_0+s_n^{-1/2}\boldsymbol{u},s_n^{-1}\boldsymbol{v})
                -G_n(\boldsymbol{\theta}^{\lambda})\}
            \bigr|
            \\
            &\quad
            \le
            \sqrt{s_n}\|\boldsymbol{F}_{n,A}(\boldsymbol{\alpha}_0)\|
            +
            \frac{1}{2}\sup_{(\boldsymbol{a},\boldsymbol{0})\in\cK}
            \|\boldsymbol{J}_{n,AA}(\boldsymbol{a},\boldsymbol{0})
            +\lambda_n\nabla^2\rho_A(\boldsymbol{a})\|
            +
            \sup_{\boldsymbol{\theta}\in\cK}\|\boldsymbol{U}_{n,I}(\boldsymbol{\theta})\|
            +
            \lambda_n C_I
            \\
            &\quad\eqqcolon
            A_n
            =
            O_p(1)
            .
        \end{split}
    \end{equation*}
    For $\|\boldsymbol{u}\|,\|\boldsymbol{v}\|\le1$ and all sufficiently large $n$, Assumption~\ref{ns:loss}~\ref{ns:loss:posterior} also gives $|r_n(\boldsymbol{\alpha}_0+s_n^{-1/2}\boldsymbol{u},s_n^{-1}\boldsymbol{v})-r_n(\boldsymbol{\theta}^{\lambda})|\le B_n\coloneq2L_n/\sqrt{s_n}=o(1)$.
    Let $\hat{Z}_n(\eta)$ be the integral in \eqref{eq:small-ball}, evaluated at the present criterion and target.
    Changing variables to $(\boldsymbol{u},\boldsymbol{v})$ via $(\boldsymbol{\alpha},\boldsymbol{\gamma})=(\boldsymbol{\alpha}_0+s_n^{-1/2}\boldsymbol{u},s_n^{-1}\boldsymbol{v})$ with Jacobian $s_n^{-q/2-r}$ yields, simultaneously over $\eta\in\cE$ on one event of probability tending to one,
    \begin{equation}
    \label{ns:smallball}
        \begin{split}
            \hat{Z}_n(\eta)
            &\ge
            s_n^{-q/2-r}
            \int_{\|\boldsymbol{u}\|\le1}\int_{\|\boldsymbol{v}\|\le1}
            \exp(-\eta A_n-B_n)\dd\boldsymbol{v}\dd\boldsymbol{u}
            \\
            &=
            c_qc_r s_n^{-(q/2+r)}\exp(-\eta A_n-B_n)
            .
        \end{split}
    \end{equation}
    Since Assumption~\ref{ns:loss}~\ref{ns:loss:posterior} imposes Assumption~\ref{ass:posterior}~\ref{ass:posterior:separation}, for each fixed $\varepsilon>0$,
    \begin{equation}
    \label{ns:concentration}
        \begin{split}
            \sup_{\eta\in\cE}\Pi_n^\eta(B_\varepsilon^c\mid\cD_n)
            &=
            \sup_{\eta\in\cE}\frac{I_{n,\varepsilon}(\eta)}{\hat{Z}_n(\eta)}
            \\
            &\le
            (c_qc_r)^{-1}s_n^{q/2+r}\exp(\eta_U A_n+B_n)
            \sup_{\eta\in\cE}I_{n,\varepsilon}(\eta)
            \\
            &=
            O_p\{s_n^{q/2+r}\exp(-a_\varepsilon s_n)\}
            \to_p
            0
            .
        \end{split}
    \end{equation}

    Next put $\boldsymbol{a}_n(\boldsymbol{u})=\check{\boldsymbol{\alpha}}_n+s_n^{-1/2}\boldsymbol{u}$ and $\boldsymbol{\theta}_n(\boldsymbol{u},\boldsymbol{v})=(\boldsymbol{a}_n(\boldsymbol{u}),s_n^{-1}\boldsymbol{v})$.
    Define the unnormalized density in these coordinates by
    \begin{equation*}
        \log g_n^\eta(\boldsymbol{u},\boldsymbol{v})
        =
        -\eta s_n\{G_n(\boldsymbol{\theta}_n(\boldsymbol{u},\boldsymbol{v}))
        -G_n(\check{\boldsymbol{\alpha}}_n,\boldsymbol{0})\}
        +
        r_n(\boldsymbol{\theta}_n(\boldsymbol{u},\boldsymbol{v}))
        -r_n(\check{\boldsymbol{\alpha}}_n,\boldsymbol{0})
        ,
    \end{equation*}
    when $\boldsymbol{\theta}_n(\boldsymbol{u},\boldsymbol{v})\in\Theta$, and set $g_n^\eta(\boldsymbol{u},\boldsymbol{v})=0$ otherwise.
    Properness in Assumption~\ref{ns:loss}~\ref{ns:loss:posterior} ensures that normalizing $g_n^\eta$ gives the density of $P_n^{\check{\boldsymbol{\alpha}}_n,\eta}$.
    When $(\check{\boldsymbol{\alpha}}_n,\boldsymbol{0})$ and $\boldsymbol{\theta}_n(\boldsymbol{u},\boldsymbol{v})$ lie in the convex product neighborhood $\cN$ on which $\bar{M}_n$ is $C^2$, splitting the loss increment into its active and inactive parts and using $\rho_I(s_n^{-1}\boldsymbol{v})-\rho_I(\boldsymbol{0})=s_n^{-1}\rho_I(\boldsymbol{v})$ give
    \begin{equation}
    \label{ns:exact}
        \begin{split}
            &s_n\{G_n(\boldsymbol{\theta}_n(\boldsymbol{u},\boldsymbol{v}))
            -G_n(\check{\boldsymbol{\alpha}}_n,\boldsymbol{0})\}
            \\
            &\quad
            =
            \sqrt{s_n}\boldsymbol{F}_{n,A}(\check{\boldsymbol{\alpha}}_n)^{\top}\boldsymbol{u}
            +
            \frac{1}{2}\boldsymbol{u}^{\top}\boldsymbol{H}_{A,n}(\boldsymbol{u})\boldsymbol{u}
            +
            \int_0^1\boldsymbol{U}_{n,I}(\boldsymbol{a}_n(\boldsymbol{u}),t s_n^{-1}\boldsymbol{v})^{\top}\boldsymbol{v}\dd t
            +
            \lambda_n\rho_I(\boldsymbol{v})
            ,
        \end{split}
    \end{equation}
    where
    \begin{equation*}
        \boldsymbol{H}_{A,n}(\boldsymbol{u})
        =
        2\int_0^1(1-t)\bigl\{
            \boldsymbol{J}_{n,AA}(\check{\boldsymbol{\alpha}}_n+t s_n^{-1/2}\boldsymbol{u},\boldsymbol{0})
            +
            \lambda_n\nabla^2\rho_A(\check{\boldsymbol{\alpha}}_n+t s_n^{-1/2}\boldsymbol{u})
        \bigr\}\dd t
        .
    \end{equation*}
    Fix $M<\infty$.
    Assumptions~\ref{ns:penalty}~\ref{ns:penalty:structure}--\ref{ns:penalty:population} and~\ref{ns:loss}~\ref{ns:loss:derivatives}, \ref{ns:loss:baseline} imply
    \begin{equation*}
        \sqrt{s_n}\|\boldsymbol{F}_{n,A}(\check{\boldsymbol{\alpha}}_n)\|
        +
        \sup_{\|\boldsymbol{u}\|\le M}
        \|\boldsymbol{H}_{A,n}(\boldsymbol{u})-\boldsymbol{J}_{\lambda,AA}^{\ast}\|
        \to_p
        0
        .
    \end{equation*}
    Define $\psi_\eta(\boldsymbol{u},\boldsymbol{v})=\eta\{\boldsymbol{u}^{\top}\boldsymbol{J}_{\lambda,AA}^{\ast}\boldsymbol{u}/2+\phi_I(\boldsymbol{v})\}$.
    Combining these bounds with uniform convergence of $\boldsymbol{U}_{n,I}$, continuity of $\Psi_I$ at $\boldsymbol{\theta}^\lambda$, and the local Lipschitz bound for $r_n$ in Assumption~\ref{ns:loss}~\ref{ns:loss:posterior} gives
    \begin{equation}
    \label{ns:local-limit}
        \begin{split}
            &\sup_{\eta\in\cE}\sup_{\|\boldsymbol{u}\|,\|\boldsymbol{v}\|\le M}
            |\log g_n^\eta(\boldsymbol{u},\boldsymbol{v})+\psi_\eta(\boldsymbol{u},\boldsymbol{v})|
            \\
            &\quad
            \le
            \eta_U\Biggl\{
                M\sqrt{s_n}\|\boldsymbol{F}_{n,A}(\check{\boldsymbol{\alpha}}_n)\|
                +
                \frac{M^2}{2}\sup_{\|\boldsymbol{u}\|\le M}
                \|\boldsymbol{H}_{A,n}(\boldsymbol{u})-\boldsymbol{J}_{\lambda,AA}^{\ast}\|
                \\
                &\qquad
                +
                M\sup_{\substack{\|\boldsymbol{u}\|,\|\boldsymbol{v}\|\le M\\0\le t\le1}}
                \|\boldsymbol{U}_{n,I}(\boldsymbol{a}_n(\boldsymbol{u}),t s_n^{-1}\boldsymbol{v})
                -\Psi_I(\boldsymbol{\theta}^{\lambda})\|
                +
                |\lambda_n-\lambda|C_I M
            \Biggr\}
            \\
            &\qquad
            +
            L_nM(s_n^{-1/2}+s_n^{-1})
            \\
            &\quad
            \to_p
            0
            .
        \end{split}
    \end{equation}

    To control the local tails, choose $\delta>0$ so that the closed product ball $\{(\boldsymbol{a},\boldsymbol{g}):\|\boldsymbol{a}-\boldsymbol{\alpha}_0\|\le\delta,\ \|\boldsymbol{g}\|\le\delta\}$ is contained in $\cN\cap B_{\delta_0}$.
    By Assumptions~\ref{ns:penalty}~\ref{ns:penalty:population}--\ref{ns:penalty:curvature} and~\ref{ns:loss}~\ref{ns:loss:derivatives}, $\delta$ can be chosen so that, for some fixed $c_A>0$, with probability tending to one on this product ball,
    \begin{equation*}
        \begin{gathered}
            \boldsymbol{J}_{n,AA}(\boldsymbol{a},\boldsymbol{0})
            +
            \lambda_n\nabla^2\rho_A(\boldsymbol{a})
            \succeq
            c_A\boldsymbol{I}_q
            ,\\
            \|\boldsymbol{U}_{n,I}(\boldsymbol{a},\boldsymbol{g})-\Psi_I(\boldsymbol{\theta}^{\lambda})\|
            +
            |\lambda_n-\lambda|C_I
            \le
            m_I/2
            .
        \end{gathered}
    \end{equation*}
    The second bound and \eqref{ns:population} imply, for every $\boldsymbol{v}\in\RR^r$,
    \begin{equation*}
        \begin{split}
            \boldsymbol{U}_{n,I}(\boldsymbol{a},\boldsymbol{g})^{\top}\boldsymbol{v}
            +
            \lambda_n\rho_I(\boldsymbol{v})
            &\ge
            \phi_I(\boldsymbol{v})
            -
            \{\|\boldsymbol{U}_{n,I}(\boldsymbol{a},\boldsymbol{g})
            -\Psi_I(\boldsymbol{\theta}^{\lambda})\|+|\lambda_n-\lambda|C_I\}\|\boldsymbol{v}\|
            \\
            &\ge
            \frac{m_I}{2}\|\boldsymbol{v}\|
            .
        \end{split}
    \end{equation*}
    Define $\cU_n=\{(\boldsymbol{u},\boldsymbol{v}):\|\boldsymbol{a}_n(\boldsymbol{u})-\boldsymbol{\alpha}_0\|<\delta,\ \|s_n^{-1}\boldsymbol{v}\|<\delta\}$ and put $\epsilon_n=\eta_U\sqrt{s_n}\|\boldsymbol{F}_{n,A}(\check{\boldsymbol{\alpha}}_n)\|+L_n/\sqrt{s_n}=o_p(1)$.
    Since $\check{\boldsymbol{\alpha}}_n\to_p\boldsymbol{\alpha}_0$, every fixed ball $\{\|(\boldsymbol{u},\boldsymbol{v})\|\le M\}$ lies in $\cU_n$ with probability tending to one.
    On the event where these two bounds hold on the product ball, $\|\check{\boldsymbol{\alpha}}_n-\boldsymbol{\alpha}_0\|<\delta$, $\epsilon_n^2\le\eta_L c_A$, and $L_n/s_n\le\eta_Lm_I/4$, \eqref{ns:exact} and Young's inequality yield, for all $(\boldsymbol{u},\boldsymbol{v})\in\cU_n$ and $\eta\in\cE$,
    \begin{equation*}
        \begin{split}
            \log g_n^\eta(\boldsymbol{u},\boldsymbol{v})
            &\le
            -\frac{\eta_L c_A}{2}\|\boldsymbol{u}\|^2
            +
            \epsilon_n\|\boldsymbol{u}\|
            -
            \Bigl(\frac{\eta_Lm_I}{2}-\frac{L_n}{s_n}\Bigr)\|\boldsymbol{v}\|
            \\
            &\le
            -\frac{\eta_L c_A}{4}\|\boldsymbol{u}\|^2
            +
            \frac{\epsilon_n^2}{\eta_L c_A}
            -
            \Bigl(
                \frac{\eta_Lm_I}{2}-\frac{L_n}{s_n}
            \Bigr)\|\boldsymbol{v}\|
            \\
            &\le
            1-\frac{\eta_L c_A}{4}\|\boldsymbol{u}\|^2-\frac{\eta_Lm_I}{4}\|\boldsymbol{v}\|
            .
        \end{split}
    \end{equation*}
    Hence, with probability tending to one,
    \begin{equation}
    \label{ns:envelope}
        \sup_{\eta\in\cE}g_n^\eta(\boldsymbol{u},\boldsymbol{v})
        \le
        b(\boldsymbol{u},\boldsymbol{v})
        \coloneq
        \exp\{1-\eta_L c_A\|\boldsymbol{u}\|^2/4-\eta_Lm_I\|\boldsymbol{v}\|/4\}
        ,\quad
        (\boldsymbol{u},\boldsymbol{v})\in\cU_n
        .
    \end{equation}
    This envelope $b$ is integrable on $\RR^{q+r}$.
    Moreover, since $(\boldsymbol{u},\boldsymbol{v})\in\cU_n^c$ implies $\|\boldsymbol{\theta}_n(\boldsymbol{u},\boldsymbol{v})-\boldsymbol{\theta}^{\lambda}\|\ge\delta$, \eqref{ns:concentration} gives
    \begin{equation*}
        \sup_{\eta\in\cE}
        P_n^{\check{\boldsymbol{\alpha}}_n,\eta}(\cU_n^c)
        \le
        \sup_{\eta\in\cE}\Pi_n^\eta(B_\delta^c\mid\cD_n)
        \to_p
        0
        .
    \end{equation*}

    For the limiting kernels $\exp\{-\psi_\eta(\boldsymbol{u},\boldsymbol{v})\}$, put
    \begin{equation*}
        Z_\infty(\eta)
        \coloneq
        \int_{\RR^q}\int_{\RR^r}\exp\{-\psi_\eta(\boldsymbol{u},\boldsymbol{v})\}\dd\boldsymbol{v}\dd\boldsymbol{u}
        =
        (2\pi)^{q/2}\det(\eta\boldsymbol{J}_{\lambda,AA}^{\ast})^{-1/2}Z_I(\eta)
        .
    \end{equation*}
    Assumption~\ref{ns:penalty}~\ref{ns:penalty:population}--\ref{ns:penalty:curvature} and the bounds on $\phi_I$ and $Z_I(\eta)$ imply
    \begin{equation*}
        \begin{gathered}
            0
            <
            \inf_{\eta\in\cE}Z_\infty(\eta)
            \le
            \sup_{\eta\in\cE}Z_\infty(\eta)
            <
            \infty
            ,\\
            \exp\{-\psi_\eta(\boldsymbol{u},\boldsymbol{v})\}
            \le
            \exp\{-\eta_L\boldsymbol{u}^{\top}\boldsymbol{J}_{\lambda,AA}^{\ast}\boldsymbol{u}/2
            -\eta_Lm_I\|\boldsymbol{v}\|\}
            .
        \end{gathered}
    \end{equation*}
    Thus \eqref{ns:local-limit}, \eqref{ns:envelope}, and the bound on $P_n^{\check{\boldsymbol{\alpha}}_n,\eta}(\cU_n^c)$ verify the uniform hypotheses of Lemma~\ref{supp-lem:localization}.
    Writing $\nu_\eta=N\{\boldsymbol{0},(\eta\boldsymbol{J}_{\lambda,AA}^{\ast})^{-1}\}\otimes P_I^\eta$, Lemma~\ref{supp-lem:localization} and its proof give
    \begin{equation*}
        \begin{gathered}
            \sup_{\eta\in\cE}\Bigl|
                \int_{\RR^q}\int_{\RR^r}g_n^\eta(\boldsymbol{u},\boldsymbol{v})\dd\boldsymbol{v}\dd\boldsymbol{u}
                -Z_\infty(\eta)
            \Bigr|
            \to_p
            0
            ,\\
            \sup_{\eta\in\cE}
            \|P_n^{\check{\boldsymbol{\alpha}}_n,\eta}-\nu_\eta\|_{\TV}
            \to_p
            0
            .
        \end{gathered}
    \end{equation*}
    Finally, since $P_n^{\tilde{\boldsymbol{\alpha}}_n,\eta}$ is obtained from $P_n^{\check{\boldsymbol{\alpha}}_n,\eta}$ by shifting the active coordinate by $-\sqrt{s_n}(\tilde{\boldsymbol{\alpha}}_n-\check{\boldsymbol{\alpha}}_n)=o_p(1)$ and the common factor $P_I^\eta$ cancels in the product-measure difference, Lemmas~\ref{supp-lem:tv-tools} and~\ref{supp-lem:gaussian} give
    \begin{equation*}
        \begin{split}
            \sup_{\eta\in\cE}
            \|P_n^{\tilde{\boldsymbol{\alpha}}_n,\eta}-\nu_\eta\|_{\TV}
            &\le
            \sup_{\eta\in\cE}
            \|P_n^{\check{\boldsymbol{\alpha}}_n,\eta}-\nu_\eta\|_{\TV}
            +
            \frac{\sqrt{\eta_U}\|(\boldsymbol{J}_{\lambda,AA}^{\ast})^{1/2}\|
            \sqrt{s_n}\|\tilde{\boldsymbol{\alpha}}_n-\check{\boldsymbol{\alpha}}_n\|}{\sqrt{2\pi}}
            \\
            &\to_p
            0
            .
        \end{split}
    \end{equation*}
\end{proof}

Proposition~\ref{ns:bvm} gives the active marginal limit by integrating the full posterior over the inactive coordinates.
It therefore supplies the Gaussian approximation needed for inference on the active coordinates, while the inactive limiting law retains its dependence on $\eta$.
It also implies that the inactive coordinates are of order $s_n^{-1}$ under the posterior, in the sense that, for every $\varepsilon>0$,
\begin{equation*}
    \lim_{M\to\infty}\limsup_{n\to\infty}
    P^\ast\Biggl[
        \sup_{\eta\in\cE}\Pi_n^\eta\{s_n\|\boldsymbol{\gamma}\|>M\mid\cD_n\}>\varepsilon
    \Biggr]
    =
    0
    .
\end{equation*}
Indeed, projection onto the inactive block cannot increase total variation, and $\sup_{\eta\in\cE}P_I^\eta(\|\boldsymbol{v}\|>M)\to0$ as $M\to\infty$ by the common integrable tail bound for the densities $p_I^\eta$.
At scale $s_n^{-1/2}$, the inactive coordinates therefore vanish and the joint weak limit is supported on the active subspace.
The different scales in \eqref{ns:joint-tv} are essential for total variation, because an absolutely continuous posterior assigns that subspace probability zero, whereas the singular weak limit assigns it probability one.

We now calibrate the active marginal posterior.
Let $\Pi_{n,A}^\eta$ denote the active marginal posterior, and choose a data-measurable active admissible center $\tilde{\boldsymbol{\alpha}}_n$ common to all $\eta\in\cE$.
Define
\begin{equation*}
    \begin{gathered}
        \boldsymbol{H}_{0,A}(\eta)
        =
        \eta\boldsymbol{J}_{\lambda,AA}^{\ast}
        ,\quad
        \boldsymbol{\Omega}_A(\eta)
        =
        (\boldsymbol{V}_{\mathrm{target},AA}^{\ast})^{1/2}\boldsymbol{H}_{0,A}(\eta)^{1/2}
        ,\\
        \boldsymbol{\alpha}_{\mathrm{calib}}^{(\eta)}
        =
        \tilde{\boldsymbol{\alpha}}_n+\boldsymbol{\Omega}_A(\eta)(\boldsymbol{\alpha}-\tilde{\boldsymbol{\alpha}}_n)
        ,\quad
        \boldsymbol{\alpha}
        \sim
        \Pi_{n,A}^{\eta}(\cdot\mid\cD_n)
        .
    \end{gathered}
\end{equation*}
Then we obtain the following result.
\begin{corollary}[Active location--scale calibration]
\label{ns:calibration}
    Under Assumptions~\ref{ns:penalty} and~\ref{ns:loss},
    \begin{equation*}
        \sup_{\eta\in\cE}\Bigl\|
            \operatorname{Law}\Bigl\{
                s_n^{1/2}(\boldsymbol{\alpha}_{\mathrm{calib}}^{(\eta)}-\tilde{\boldsymbol{\alpha}}_n)\Bigm|\cD_n
            \Bigr\}
            -
            N(\boldsymbol{0},\boldsymbol{V}_{\mathrm{target},AA}^{\ast})
        \Bigr\|_{\TV}
        \to_p
        0
        .
    \end{equation*}
\end{corollary}

\begin{proof}[Proof of Corollary~\ref{ns:calibration}]
    Let $\nu_\eta=N\{\boldsymbol{0},\boldsymbol{H}_{0,A}(\eta)^{-1}\}\otimes P_I^\eta$ and define $T_\eta(\boldsymbol{u},\boldsymbol{v})=\boldsymbol{\Omega}_A(\eta)\boldsymbol{u}$.
    The covariance identity
    \begin{equation*}
        \boldsymbol{\Omega}_A(\eta)\boldsymbol{H}_{0,A}(\eta)^{-1}\boldsymbol{\Omega}_A(\eta)^{\top}
        =
        \boldsymbol{V}_{\mathrm{target},AA}^{\ast}
    \end{equation*}
    gives $(T_\eta)_\#\nu_\eta=N(\boldsymbol{0},\boldsymbol{V}_{\mathrm{target},AA}^{\ast})$.
    Since $\sqrt{s_n}(\boldsymbol{\alpha}_{\mathrm{calib}}^{(\eta)}-\tilde{\boldsymbol{\alpha}}_n)=\boldsymbol{\Omega}_A(\eta)\sqrt{s_n}(\boldsymbol{\alpha}-\tilde{\boldsymbol{\alpha}}_n)$ has conditional law $(T_\eta)_\#P_n^{\tilde{\boldsymbol{\alpha}}_n,\eta}$ given $\cD_n$, Lemma~\ref{supp-lem:tv-tools} and Proposition~\ref{ns:bvm} imply
    \begin{equation*}
        \begin{split}
            &\sup_{\eta\in\cE}\Bigl\|
                \operatorname{Law}\{\sqrt{s_n}(\boldsymbol{\alpha}_{\mathrm{calib}}^{(\eta)}
                -\tilde{\boldsymbol{\alpha}}_n)\mid\cD_n\}
                -N(\boldsymbol{0},\boldsymbol{V}_{\mathrm{target},AA}^{\ast})
            \Bigr\|_{\TV}
            \\
            &\quad
            =
            \sup_{\eta\in\cE}
            \|(T_\eta)_\#P_n^{\tilde{\boldsymbol{\alpha}}_n,\eta}-(T_\eta)_\#\nu_\eta\|_{\TV}
            \\
            &\quad
            \le
            \sup_{\eta\in\cE}
            \|P_n^{\tilde{\boldsymbol{\alpha}}_n,\eta}-\nu_\eta\|_{\TV}
            \to_p
            0
            .
        \end{split}
    \end{equation*}
\end{proof}

For feasible calibration, let $\eta_n\in\cE$ be data-measurable, and write $\eta=\eta_n$ in the following conditions and construction.
Let $\hat{\boldsymbol{m}}_{A,n}$ and $\hat{\boldsymbol{\Sigma}}_{A,n}$ be the mean and covariance estimates from pilot draws.
Let $\cW_n$ be the sigma-field generated by the data, the pilot draws, and any auxiliary randomness used to compute the estimates.
All estimated matrices and centers below are $\cW_n$-measurable.
We impose the following conditions.
\begin{assumption}[Active posterior moments and estimation]
\label{ns:feasible}
    \leavevmode
    \begin{enumerate}[label=(\roman*),leftmargin=*]
        \item\label{ns:feasible:moment}
        The selected posterior has a finite second moment on the active coordinates and, for every fixed sufficiently small $\delta>0$,
        \begin{equation}
        \label{ns:moment-tail}
            \int_{\{\|\boldsymbol{\theta}-\boldsymbol{\theta}^\lambda\|\ge\delta\}}s_n\|\boldsymbol{\alpha}-\check{\boldsymbol{\alpha}}_n\|^2\Pi_n^\eta(\dd\boldsymbol{\theta}\mid\cD_n)
            \to_p
            0
            .
        \end{equation}

        \item\label{ns:feasible:mc}
        The pilot estimates satisfy
        \begin{equation*}
            s_n^{1/2}\|\hat{\boldsymbol{m}}_{A,n}-\EE_{\Pi_n^\eta}[\boldsymbol{\alpha}\mid\cD_n]\|
            \to_p
            0
            ,\quad
            s_n\|\hat{\boldsymbol{\Sigma}}_{A,n}-\var_{\Pi_n^\eta}(\boldsymbol{\alpha}\mid\cD_n)\|
            \to_p
            0
            .
        \end{equation*}

        \item\label{ns:feasible:target}
        Symmetric data- or $\cW_n$-measurable estimators $\hat{\boldsymbol{J}}_{\lambda,AA}$ and $\hat{\boldsymbol{K}}_{AA}$ satisfy $\hat{\boldsymbol{J}}_{\lambda,AA}\to_p\boldsymbol{J}_{\lambda,AA}^{\ast}$ and $\hat{\boldsymbol{K}}_{AA}\to_p\boldsymbol{K}_{AA}^{\ast}$.
        The output center $\bar{\boldsymbol{\alpha}}_n$ is active admissible under the joint law of the data and pilot simulation.
    \end{enumerate}
\end{assumption}

For a new draw $\boldsymbol{\alpha}\sim\Pi_{n,A}^{\eta}(\cdot\mid\cD_n)$ conditionally independent of the pilot draws and auxiliary randomness given the data, define
\begin{equation*}
    \begin{gathered}
        \hat{\boldsymbol{V}}_{A}
        =
        \hat{\boldsymbol{J}}_{\lambda,AA}^{-1}\hat{\boldsymbol{K}}_{AA}\hat{\boldsymbol{J}}_{\lambda,AA}^{-1}
        ,\quad
        \hat{\boldsymbol{\Omega}}_{A}
        =
        \hat{\boldsymbol{V}}_{A}^{1/2}(s_n\hat{\boldsymbol{\Sigma}}_{A,n})^{-1/2}
        ,\\
        \hat{\boldsymbol{\alpha}}_{\mathrm{calib}}
        =
        \bar{\boldsymbol{\alpha}}_n
        +
        \hat{\boldsymbol{\Omega}}_A(\boldsymbol{\alpha}-\hat{\boldsymbol{m}}_{A,n})
        ,
    \end{gathered}
\end{equation*}
using an arbitrary fixed map when an estimated covariance matrix is not positive definite.
Then we obtain the following result.
\begin{corollary}[Feasible active calibration]
\label{ns:feasible-calibration}
    Under Assumptions~\ref{ns:penalty}, \ref{ns:loss}, and~\ref{ns:feasible}, the conditional law of $s_n^{1/2}(\hat{\boldsymbol{\alpha}}_{\mathrm{calib}}-\bar{\boldsymbol{\alpha}}_n)$ given $\cW_n$ converges in total variation in probability to $N(\boldsymbol{0},\boldsymbol{V}_{\mathrm{target},AA}^{\ast})$.
\end{corollary}
Corollary~\ref{ns:feasible-calibration} gives coverage for the corresponding ellipsoid
\begin{equation*}
    C_{n,A}(a)
    =
    \bigl\{
        \boldsymbol{b}:s_n(\boldsymbol{b}-\bar{\boldsymbol{\alpha}}_n)^{\top}\hat{\boldsymbol{V}}_{A}^{-1}(\boldsymbol{b}-\bar{\boldsymbol{\alpha}}_n)\le\chi^2_{q,1-a}
    \bigr\}
    ,\quad
    0<a<1
    ,
\end{equation*}
where $\chi^2_{q,1-a}$ is the $(1-a)$-quantile of the chi-square distribution with $q$ degrees of freedom.
By the argument for Corollary~\ref{cor:coverage}, using the consistency of $\hat{\boldsymbol{V}}_A$ and Lemma~\ref{ns:estimator} in place of Proposition~\ref{prop:baseline-clt}, $C_{n,A}(a)$ has calibrated posterior probability tending to $1-a$ in probability and satisfies $P^{\ast}\{\boldsymbol{\alpha}_0\in C_{n,A}(a)\}\to1-a$.
These limits hold for every selected rate in $\cE$, with the same active target and limiting covariance.
The moment and Monte Carlo conditions concern the selected posterior and its pilot simulation.

\begin{proof}[Proof of Corollary~\ref{ns:feasible-calibration}]
    Put $\boldsymbol{H}_{0,A}=\eta_n\boldsymbol{J}_{\lambda,AA}^{\ast}$ and $\boldsymbol{\Omega}_{A,n}=(\boldsymbol{V}_{\mathrm{target},AA}^{\ast})^{1/2}\boldsymbol{H}_{0,A}^{1/2}$.
    Use $g_n^\eta$, $\cU_n$, $b$, and $Z_\infty(\eta)$ from the proof of Proposition~\ref{ns:bvm} at the selected rate $\eta=\eta_n$.
    Since $\sup_{\eta\in\cE}|\int g_n^\eta\dd\boldsymbol{u}\dd\boldsymbol{v}-Z_\infty(\eta)|\to_p0$ and $c_{\ast}\coloneq\frac{1}{2}\inf_{\eta\in\cE}Z_\infty(\eta)>0$, we have $\int g_n^\eta\dd\boldsymbol{u}\dd\boldsymbol{v}\ge c_{\ast}$ with probability tending to one.
    On the event where this lower bound and \eqref{ns:envelope} hold, since $\boldsymbol{\theta}_n(\cU_n^c)\subset B_\delta^c$, splitting the expectation over $\cU_n$ and $\cU_n^c$ gives
    \begin{equation*}
        \begin{split}
            &\EE_{\Pi_n^\eta}\bigl[
                s_n\|\boldsymbol{\alpha}-\check{\boldsymbol{\alpha}}_n\|^2
                \1\{\sqrt{s_n}\|\boldsymbol{\alpha}-\check{\boldsymbol{\alpha}}_n\|>M\}\mid\cD_n
            \bigr]
            \\
            &\quad
            \le
            c_{\ast}^{-1}\int_{\|\boldsymbol{u}\|>M}\int_{\RR^r}
            \|\boldsymbol{u}\|^2b(\boldsymbol{u},\boldsymbol{v})\dd\boldsymbol{v}\dd\boldsymbol{u}
            +
            \int_{B_\delta^c}s_n\|\boldsymbol{\alpha}-\check{\boldsymbol{\alpha}}_n\|^2
            \Pi_n^\eta(\dd\boldsymbol{\theta}\mid\cD_n)
            .
        \end{split}
    \end{equation*}
    The first term tends to zero as $M\to\infty$, and the second is $o_p(1)$ by Assumption~\ref{ns:feasible}~\ref{ns:feasible:moment}.
    Thus the active second moments are uniformly integrable in probability, as are the second moments of $N(\boldsymbol{0},\boldsymbol{H}_{0,A}^{-1})$ because $\eta_n\in\cE$.
    Applying the truncation argument of Lemma~\ref{lem:postvar-consistency} to the active marginal in Proposition~\ref{ns:bvm} gives
    \begin{equation*}
        \begin{gathered}
            \sqrt{s_n}\{\EE_{\Pi_n^\eta}[\boldsymbol{\alpha}\mid\cD_n]
            -\check{\boldsymbol{\alpha}}_n\}
            \to_p
            \boldsymbol{0}
            ,\\
            \|s_n\var_{\Pi_n^\eta}(\boldsymbol{\alpha}\mid\cD_n)-\boldsymbol{H}_{0,A}^{-1}\|
            \to_p
            0
            .
        \end{gathered}
    \end{equation*}
    Adding the Monte Carlo errors from Assumption~\ref{ns:feasible} \ref{ns:feasible:mc}, we obtain
    \begin{equation*}
        \begin{split}
            \sqrt{s_n}(\hat{\boldsymbol{m}}_{A,n}-\check{\boldsymbol{\alpha}}_n)
            &=
            \sqrt{s_n}\{\hat{\boldsymbol{m}}_{A,n}-\EE_{\Pi_n^\eta}[\boldsymbol{\alpha}\mid\cD_n]\}
            +
            \sqrt{s_n}\{\EE_{\Pi_n^\eta}[\boldsymbol{\alpha}\mid\cD_n]-\check{\boldsymbol{\alpha}}_n\}
            \to_p
            \boldsymbol{0}
            ,\\
            \|s_n\hat{\boldsymbol{\Sigma}}_{A,n}-\boldsymbol{H}_{0,A}^{-1}\|
            &\le
            s_n\|\hat{\boldsymbol{\Sigma}}_{A,n}-\var_{\Pi_n^\eta}(\boldsymbol{\alpha}\mid\cD_n)\|
            +
            \|s_n\var_{\Pi_n^\eta}(\boldsymbol{\alpha}\mid\cD_n)-\boldsymbol{H}_{0,A}^{-1}\|
            \to_p
            0
            .
        \end{split}
    \end{equation*}
    Assumption~\ref{ns:feasible}~\ref{ns:feasible:target} gives $\hat{\boldsymbol{V}}_A\to_p\boldsymbol{V}_{\mathrm{target},AA}^{\ast}$, so $\hat{\boldsymbol{V}}_A$ and $s_n\hat{\boldsymbol{\Sigma}}_{A,n}$ are positive definite with probability tending to one.
    Since $\eta_n\in\cE$, uniform continuity of matrix inversion and the symmetric square root on compact subsets of the positive-definite cone implies $\|\hat{\boldsymbol{\Omega}}_A-\boldsymbol{\Omega}_{A,n}\|\to_p0$.
    Since $\boldsymbol{\Omega}_{A,n}\boldsymbol{H}_{0,A}^{-1}\boldsymbol{\Omega}_{A,n}^\top=\boldsymbol{V}_{\mathrm{target},AA}^\ast$ and $\|\boldsymbol{H}_{0,A}^{-1}\|$ and $\|\boldsymbol{\Omega}_{A,n}\|$ are bounded for $\eta_n\in\cE$, we have
    \begin{equation*}
        \begin{split}
            \|\hat{\boldsymbol{\Omega}}_A\boldsymbol{H}_{0,A}^{-1}\hat{\boldsymbol{\Omega}}_A^\top
            -\boldsymbol{V}_{\mathrm{target},AA}^{\ast}\|
            &\le
            \|\hat{\boldsymbol{\Omega}}_A-\boldsymbol{\Omega}_{A,n}\|
            \|\boldsymbol{H}_{0,A}^{-1}\|
            \{\|\hat{\boldsymbol{\Omega}}_A\|+\|\boldsymbol{\Omega}_{A,n}\|\}
            \to_p
            0
            .
        \end{split}
    \end{equation*}

    Conditional independence of the new draw implies $\operatorname{Law}\{\sqrt{s_n}(\boldsymbol{\alpha}-\check{\boldsymbol{\alpha}}_n)\mid\cW_n\}=\operatorname{Law}\{\sqrt{s_n}(\boldsymbol{\alpha}-\check{\boldsymbol{\alpha}}_n)\mid\cD_n\}$.
    Since $\sqrt{s_n}(\hat{\boldsymbol{\alpha}}_{\mathrm{calib}}-\bar{\boldsymbol{\alpha}}_n)=\hat{\boldsymbol{\Omega}}_A\sqrt{s_n}(\boldsymbol{\alpha}-\hat{\boldsymbol{m}}_{A,n})$ and $\sqrt{s_n}(\boldsymbol{\alpha}-\hat{\boldsymbol{m}}_{A,n})=\sqrt{s_n}(\boldsymbol{\alpha}-\check{\boldsymbol{\alpha}}_n)-\sqrt{s_n}(\hat{\boldsymbol{m}}_{A,n}-\check{\boldsymbol{\alpha}}_n)$, Lemmas~\ref{supp-lem:tv-tools}--\ref{supp-lem:gaussian} and Proposition~\ref{ns:bvm} yield, on the event where $\hat{\boldsymbol{\Omega}}_A$ is nonsingular,
    \begin{equation*}
        \begin{split}
            &\Bigl\|
                \operatorname{Law}\{\sqrt{s_n}(\hat{\boldsymbol{\alpha}}_{\mathrm{calib}}-\bar{\boldsymbol{\alpha}}_n)\mid\cW_n\}
                -N(\boldsymbol{0},\boldsymbol{V}_{\mathrm{target},AA}^{\ast})
            \Bigr\|_{\TV}
            \\
            &\quad
            \le
            \Bigl\|
                \operatorname{Law}\{\sqrt{s_n}(\boldsymbol{\alpha}-\hat{\boldsymbol{m}}_{A,n})\mid\cW_n\}
                -N(\boldsymbol{0},\boldsymbol{H}_{0,A}^{-1})
            \Bigr\|_{\TV}
            \\
            &\qquad+
            \bigl\|
                N(\boldsymbol{0},\hat{\boldsymbol{\Omega}}_A\boldsymbol{H}_{0,A}^{-1}\hat{\boldsymbol{\Omega}}_A^\top)
                -N(\boldsymbol{0},\boldsymbol{V}_{\mathrm{target},AA}^{\ast})
            \bigr\|_{\TV}
            \\
            &\quad
            \le
            \Bigl\|
                \operatorname{Law}\{\sqrt{s_n}(\boldsymbol{\alpha}-\check{\boldsymbol{\alpha}}_n)\mid\cD_n\}
                -N(\boldsymbol{0},\boldsymbol{H}_{0,A}^{-1})
            \Bigr\|_{\TV}
            +
            \frac{\|\boldsymbol{H}_{0,A}^{1/2}\sqrt{s_n}(\hat{\boldsymbol{m}}_{A,n}-\check{\boldsymbol{\alpha}}_n)\|}{\sqrt{2\pi}}
            \\
            &\qquad
            +
            \bigl\|
                N(\boldsymbol{0},\hat{\boldsymbol{\Omega}}_A\boldsymbol{H}_{0,A}^{-1}\hat{\boldsymbol{\Omega}}_A^\top)
                -N(\boldsymbol{0},\boldsymbol{V}_{\mathrm{target},AA}^{\ast})
            \bigr\|_{\TV}
            \\
            &\quad\to_p
            0
            .
        \end{split}
    \end{equation*}
\end{proof}

The conditions in Assumption~\ref{ns:feasible} can be verified as follows.
If $\Theta$ is bounded, \eqref{ns:moment-tail} follows from \eqref{ns:concentration}, because the squared active distance is bounded and multiplication of the exponential tail bound by $s_n$ still gives a vanishing quantity.
On unbounded spaces, \eqref{ns:moment-tail} follows, for example, from a weighted version of the integrated separation bound, using the denominator in \eqref{ns:smallball}.
The sampler condition in Lemma~\ref{lem:mc-sufficient} applies to active coordinates when their scaled fourth moments are bounded in probability and the effective number of draws diverges.
For independent observations or independent identically distributed clusters, $\hat{\boldsymbol{K}}_{AA}$ must estimate the covariance of active scores, not their raw second moment, since $\Psi_A(\boldsymbol{\theta}^\lambda)=-\lambda\nabla\rho_A(\boldsymbol{\alpha}_0)$ need not vanish.

The following example illustrates these results for a weighted $\ell_1$ penalty.
\begin{example}[Weighted $\ell_1$ penalty]
\label{ns:lasso}
    Let $\rho(\boldsymbol{\theta})=\sum_{j=1}^pw_j|\theta_j|$, with $w_j>0$, and $A=\{j:\theta_j^\lambda\ne0\}$.
    On a sufficiently small neighborhood, $\rho_A$ is affine and $\nabla^2\rho_A=\boldsymbol{0}$.
    The inactive penalty is $\rho_I(\boldsymbol{v})=\sum_{j\in I}w_j|v_j|$.
    Strict complementarity holds exactly when
    \begin{equation*}
        |\Psi_j(\boldsymbol{\theta}^\lambda)|
        <
        \lambda w_j
        \quad
        (j\in I)
        .
    \end{equation*}
    Indeed, writing $\|\boldsymbol{v}\|_1=\sum_{j\in I}|v_j|$, the smallest positive margin bounds the left side of the second inequality in \eqref{ns:population} from below by a positive multiple of $\|\boldsymbol{v}\|_1$, hence of $\|\boldsymbol{v}\|$, and choosing $\boldsymbol{v}$ along each signed coordinate axis proves necessity.
    The inactive limiting density factorizes, its $j$th factor being
    \begin{equation*}
        p_{I,j}^{\eta}(v)
        =
        \frac{\eta\{(\lambda w_j)^2-\Psi_j(\boldsymbol{\theta}^\lambda)^2\}}{2\lambda w_j}\exp[-\eta\{\Psi_j(\boldsymbol{\theta}^\lambda)v+\lambda w_j|v|\}]
        ,\quad
        v\in\RR
        .
    \end{equation*}
    This is an asymmetric Laplace density.
    The active target covariance is
    \begin{equation*}
        (\boldsymbol{J}_{AA}^{\ast})^{-1}\boldsymbol{K}_{AA}^{\ast}(\boldsymbol{J}_{AA}^{\ast})^{-1}
        .
    \end{equation*}
    This example describes continuous concentration around the zero coordinates, and the finite-sample posterior has no atom on a coordinate hyperplane.
    Inference on the active coordinates uses the marginal posterior obtained by integrating out $\boldsymbol{\gamma}$.
\end{example}

We next relate the active baseline to the minimizer of the penalized loss.
For a convex homogeneous function, its subdifferential at zero is $\partial\rho_I(\boldsymbol{0})=\{\boldsymbol{z}:\boldsymbol{z}^{\top}\boldsymbol{v}\le\rho_I(\boldsymbol{v})\text{ for every }\boldsymbol{v}\}$.
If the empirical loss and penalty are globally convex, a consistent exact minimizer restricted to $\boldsymbol{\gamma}=\boldsymbol{0}$ is also a full minimizer with probability tending to one.
Its active score is zero, and uniform score convergence together with \eqref{ns:population} gives $-\boldsymbol{U}_{n,I}(\check{\boldsymbol{\alpha}}_n,\boldsymbol{0})\in\lambda_n\partial\rho_I(\boldsymbol{0})$, which verifies the remaining first-order condition.

The following remark concerns selection of the active set and boundary targets.
\begin{remark}[Selection and boundary targets]
    In the weighted $\ell_1$ case, suppose an estimator $\hat{\boldsymbol{\theta}}_n=(\hat{\boldsymbol{\alpha}}_n^{\top},\hat{\boldsymbol{\gamma}}_n^{\top})^{\top}$ satisfies $\hat{\boldsymbol{\alpha}}_n-\boldsymbol{\alpha}_0=O_p(s_n^{-1/2})$ and $\hat{\boldsymbol{\gamma}}_n=O_p(s_n^{-1})$.
    Thresholding its coordinates at $t_n\to0$ with $s_nt_n\to\infty$ consistently recovers the fixed set $A$, because nonzero target coordinates stay above the threshold and the inactive coordinates fall below it in probability.
    Using such a selected set preserves pointwise, unconditional coverage because the probability of a different set tends to zero.
    This is not a uniform post-selection guarantee for targets approaching the selection boundary, nor a coverage theorem conditional on every possible selected model.

    If $|\Psi_j(\boldsymbol{\theta}^\lambda)|=\lambda w_j$ for an inactive coordinate, the proposed limiting density has no decay along one half-line and is not integrable at the stated scale.
    The inactive rate and limiting law must then be rederived.
    Thus the active coverage result is pointwise for a fixed target satisfying strict complementarity.
\end{remark}

\section{Piecewise smooth losses and verification for penalized Huber regression}
\label{hb:theory}

This section verifies the conditions of Sections~\ref{sec:main} and~\ref{sec:plugin} for the penalized Huber regression of Section~\ref{hb:application}.
The Huber score is continuous and piecewise linear, with slope changes at residuals equal to $\pm c$.
This structure permits the local value and score expansions needed for Gaussian posterior limits, although the empirical loss is not $C^2$ on a fixed neighborhood.
We first give a convexity condition that controls posterior tails, then bound the contribution from observations near the slope changes.
Together these arguments verify the ridge-penalized clustered example under finite second design moments and no atoms at the two Huber kinks.

We retain $G_n=\bar{M}_n+\lambda_n\rho$ and $\boldsymbol{F}_n=\nabla G_n$.
The following conditions provide a global bound for the posterior tails.

\begin{assumption}[Strong convexity and prior regularity]
\label{hb:convex}
    \leavevmode
    \begin{enumerate}[label=(\roman*),leftmargin=*]
        \item\label{hb:convex:criterion}
        The parameter space is $\Theta=\RR^p$ and $\lambda_n>0$ with $\lambda_n\to\lambda>0$.
        The random function $\bar{M}_n$ is finite, convex, and continuously differentiable almost surely.
        The penalty $\rho$ is $C^2$ and satisfies $\nabla^2\rho(\boldsymbol{\theta})\succeq\kappa\boldsymbol{I}_p$ for every $\boldsymbol{\theta}$, for a constant $\kappa>0$.

        \item\label{hb:convex:local}
        There is a fixed $\boldsymbol{\theta}^\lambda$ such that
        \begin{equation*}
            \sqrt{s_n}\,\boldsymbol{F}_n(\boldsymbol{\theta}^\lambda)
            =
            O_p(1)
            ,\quad
            \boldsymbol{F}_n
            =
            \nabla G_n
            ,\quad
            G_n
            =
            \bar{M}_n+\lambda_n\rho
            .
        \end{equation*}
        On a fixed ball about $\boldsymbol{\theta}^\lambda$, $\boldsymbol{F}_n$ is Lipschitz with constant $B_n=O_p(1)$.

        \item\label{hb:convex:prior}
        The function $r_n$ is finite and globally Lipschitz with a finite constant $L_n=o_p(\sqrt{s_n})$.
    \end{enumerate}
\end{assumption}
Then we obtain the following result.

\begin{lemma}[Strong convexity and posterior tails]
\label{hb:convex-verification}
    Under Assumption~\ref{hb:convex}, $G_n$ has a unique measurable minimizer $\check{\boldsymbol{\theta}}_n$ with zero score and
    \begin{equation*}
        \|\check{\boldsymbol{\theta}}_n-\boldsymbol{\theta}^\lambda\|
        \le
        \frac{\|\boldsymbol{F}_n(\boldsymbol{\theta}^\lambda)\|}{\lambda_n\kappa}
        =
        O_p(s_n^{-1/2})
        .
    \end{equation*}
    The posterior is proper for every $\eta>0$, and Assumption~\ref{ass:posterior} \ref{ass:posterior:separation} holds uniformly over $\eta\ge\eta_L>0$.
    Uniformly over $\eta\in\cE$, the conditional densities of $\sqrt{s_n}(\boldsymbol{\theta}-\check{\boldsymbol{\theta}}_n)$ are bounded by $C_n\exp(-c_0\|\boldsymbol{u}\|^2)$ on events of probability tending to one, where $c_0>0$ is fixed and $C_n=O_p(1)$ is independent of $\eta$.
\end{lemma}
The density bound in Lemma~\ref{hb:convex-verification} controls the posterior moments.
For any fixed $k>0$, the integral of $\|\boldsymbol{u}\|^k$ over $\{\|\boldsymbol{u}\|>M\}$ under these conditional densities is at most $C_n$ times the corresponding integral of $\|\boldsymbol{u}\|^k\exp(-c_0\|\boldsymbol{u}\|^2)$, which tends to zero as $M\to\infty$.
Truncating $C_n$ at a fixed large constant then shows that the moments of every fixed positive order are uniformly integrable in probability, uniformly over $\cE$.
An admissible change of center translates $\boldsymbol{u}$ by $o_p(1)$ and preserves this conclusion, by the triangle inequality and the same Gaussian tail bound.
In particular, Assumption~\ref{ass:UI} holds under Assumption~\ref{hb:convex}.

\begin{proof}[Proof of Lemma~\ref{hb:convex-verification}]
    Put $\kappa_n=\lambda_n\kappa$.
    Assumption~\ref{hb:convex}~\ref{hb:convex:criterion} gives
    \begin{equation*}
        G_n(\boldsymbol{b})
        \ge
        G_n(\boldsymbol{a})
        +
        \boldsymbol{F}_n(\boldsymbol{a})^\top(\boldsymbol{b}-\boldsymbol{a})
        +
        \frac{\kappa_n}{2}\|\boldsymbol{b}-\boldsymbol{a}\|^2
        \longrightarrow
        \infty
        \quad
        (\|\boldsymbol{b}\|\to\infty)
        ,
    \end{equation*}
    for every fixed $\boldsymbol{a}$.
    Thus continuity and strong convexity give a unique minimizer $\check{\boldsymbol{\theta}}_n$ with $\boldsymbol{F}_n(\check{\boldsymbol{\theta}}_n)=\boldsymbol{0}$.
    Measurability follows by approximating the minimum on countable grids and using uniqueness.
    Adding this inequality for $(\boldsymbol{a},\boldsymbol{b})=(\boldsymbol{\theta}^\lambda,\check{\boldsymbol{\theta}}_n)$ and $(\check{\boldsymbol{\theta}}_n,\boldsymbol{\theta}^\lambda)$ and writing $\boldsymbol{h}_n=\check{\boldsymbol{\theta}}_n-\boldsymbol{\theta}^\lambda$ yield
    \begin{equation*}
        \kappa_n\|\boldsymbol{h}_n\|^2
        \le
        \{\boldsymbol{F}_n(\check{\boldsymbol{\theta}}_n)-\boldsymbol{F}_n(\boldsymbol{\theta}^\lambda)\}^\top\boldsymbol{h}_n
        =
        -\boldsymbol{F}_n(\boldsymbol{\theta}^\lambda)^\top\boldsymbol{h}_n
        \le
        \|\boldsymbol{F}_n(\boldsymbol{\theta}^\lambda)\|\|\boldsymbol{h}_n\|
        .
    \end{equation*}
    Dividing by $\|\boldsymbol{h}_n\|$ when it is nonzero and using Assumption~\ref{hb:convex}~\ref{hb:convex:local} and $\kappa_n\to\lambda\kappa>0$ give
    \begin{equation*}
        \|\boldsymbol{h}_n\|
        \le
        \kappa_n^{-1}\|\boldsymbol{F}_n(\boldsymbol{\theta}^\lambda)\|
        =
        O_p(s_n^{-1/2})
        .
    \end{equation*}

    For properness, write $\boldsymbol{h}=\boldsymbol{\theta}-\check{\boldsymbol{\theta}}_n$.
    Since $\boldsymbol{F}_n(\check{\boldsymbol{\theta}}_n)=\boldsymbol{0}$, strong convexity and Assumption~\ref{hb:convex}~\ref{hb:convex:prior} imply
    \begin{equation*}
        \begin{split}
            -\eta s_n\{G_n(\boldsymbol{\theta})-G_n(\check{\boldsymbol{\theta}}_n)\}
            +r_n(\boldsymbol{\theta})-r_n(\check{\boldsymbol{\theta}}_n)
            &\le
            -\frac{\eta s_n\kappa_n}{2}\|\boldsymbol{h}\|^2+L_n\|\boldsymbol{h}\|
            \\
            &\le
            -\frac{\eta s_n\kappa_n}{4}\|\boldsymbol{h}\|^2
            +\frac{L_n^2}{\eta s_n\kappa_n}
            .
        \end{split}
    \end{equation*}
    The exponential of the last expression is integrable on $\RR^p$ for every $\eta>0$.
    Since the posterior integrand is continuous and strictly positive, its normalizing constant is positive and finite.

    For separation, fix $\varepsilon>0$ and let $\kappa_0=\lambda\kappa/2$, so $\kappa_n\ge\kappa_0$ eventually.
    By Assumption~\ref{hb:convex}~\ref{hb:convex:local}--\ref{hb:convex:prior}, the event
    \begin{equation*}
        \|\boldsymbol{F}_n(\boldsymbol{\theta}^\lambda)\|
        \le
        \kappa_0\varepsilon/4
        ,\quad
        L_n/s_n
        \le
        \eta_L\kappa_0\varepsilon/8
    \end{equation*}
    has probability tending to one.
    On this event, for $\boldsymbol{h}=\boldsymbol{\theta}-\boldsymbol{\theta}^\lambda$ with $\|\boldsymbol{h}\|\ge\varepsilon$ and every $\eta\ge\eta_L$,
    \begin{equation*}
        \begin{gathered}
            G_n(\boldsymbol{\theta})-G_n(\boldsymbol{\theta}^\lambda)
            \ge
            -\|\boldsymbol{F}_n(\boldsymbol{\theta}^\lambda)\|\|\boldsymbol{h}\|
            +\frac{\kappa_0}{2}\|\boldsymbol{h}\|^2
            \ge
            \frac{\kappa_0}{4}\|\boldsymbol{h}\|^2
            ,\\
            -\eta s_n\{G_n(\boldsymbol{\theta})-G_n(\boldsymbol{\theta}^\lambda)\}
            +r_n(\boldsymbol{\theta})-r_n(\boldsymbol{\theta}^\lambda)
            \le
            -\frac{\eta_L\kappa_0s_n}{4}\|\boldsymbol{h}\|^2
            +L_n\|\boldsymbol{h}\|
            \le
            -\frac{\eta_L\kappa_0s_n}{8}\|\boldsymbol{h}\|^2
            .
        \end{gathered}
    \end{equation*}
    Therefore, with $a=\eta_L\kappa_0/8$,
    \begin{equation*}
        \begin{split}
            \sup_{\eta\ge\eta_L}I_{n,\varepsilon}(\eta)
            &\le
            \exp(-as_n\varepsilon^2/2)
            \int_{\RR^p}\exp(-as_n\|\boldsymbol{h}\|^2/2)\dd\boldsymbol{h}
            \\
            &=
            \Biggl(
                \frac{2\pi}{as_n}
            \Biggr)^{p/2}\exp(-as_n\varepsilon^2/2)
            ,
        \end{split}
    \end{equation*}
    which proves the required exponential bound.

    For the density bound, write $\boldsymbol{u}=\sqrt{s_n}(\boldsymbol{\theta}-\check{\boldsymbol{\theta}}_n)$ and $l_n=L_n/\sqrt{s_n}=o_p(1)$, and define
    \begin{equation*}
        g_n^\eta(\boldsymbol{u})
        =
        \exp\bigl[
            -\eta s_n\{G_n(\check{\boldsymbol{\theta}}_n+s_n^{-1/2}\boldsymbol{u})-G_n(\check{\boldsymbol{\theta}}_n)\}
            +r_n(\check{\boldsymbol{\theta}}_n+s_n^{-1/2}\boldsymbol{u})-r_n(\check{\boldsymbol{\theta}}_n)
        \bigr]
        .
    \end{equation*}
    When $\kappa_n\ge\kappa_0$, the preceding strong-convexity bound gives, uniformly over $\eta\in\cE$,
    \begin{equation*}
        \begin{split}
            g_n^\eta(\boldsymbol{u})
            &\le
            \exp\{-\eta_L\kappa_0\|\boldsymbol{u}\|^2/2+l_n\|\boldsymbol{u}\|\}
            \\
            &\le
            \exp\{l_n^2/(\eta_L\kappa_0)\}
            \exp\{-\eta_L\kappa_0\|\boldsymbol{u}\|^2/4\}
            .
        \end{split}
    \end{equation*}
    Since $\check{\boldsymbol{\theta}}_n\to_p\boldsymbol{\theta}^\lambda$, the ball $\{\check{\boldsymbol{\theta}}_n+s_n^{-1/2}\boldsymbol{u}:\|\boldsymbol{u}\|\le1\}$ lies in the ball of Assumption~\ref{hb:convex}~\ref{hb:convex:local} with probability tending to one.
    Let $c_p$ denote the volume of the unit ball in $\RR^p$.
    On this event, $\boldsymbol{F}_n(\check{\boldsymbol{\theta}}_n)=\boldsymbol{0}$, the Lipschitz bounds for $\boldsymbol{F}_n$ and $r_n$, and restricting the integral of $g_n^\eta$ to $\{\|\boldsymbol{u}\|\le1\}$ give
    \begin{equation*}
        \begin{split}
            &0
            \le
            s_n\{G_n(\check{\boldsymbol{\theta}}_n+s_n^{-1/2}\boldsymbol{u})-G_n(\check{\boldsymbol{\theta}}_n)\}
            \le
            B_n\|\boldsymbol{u}\|^2/2
            \quad
            (\|\boldsymbol{u}\|\le1)
            ,\\
            &\inf_{\eta\in\cE}\int_{\RR^p}g_n^\eta(\boldsymbol{u})\dd\boldsymbol{u}
            \ge
            c_p\exp(-\eta_U B_n/2-l_n)
            .
        \end{split}
    \end{equation*}
    Dividing the upper and lower bounds yields, simultaneously over $\eta\in\cE$,
    \begin{equation*}
        \begin{gathered}
            \frac{g_n^\eta(\boldsymbol{u})}{\int g_n^\eta(\boldsymbol{v})\dd\boldsymbol{v}}
            \le
            C_n\exp(-c_0\|\boldsymbol{u}\|^2)
            ,\quad
            c_0=\eta_L\kappa_0/4
            ,\\
            C_n
            =
            c_p^{-1}\exp\{l_n^2/(\eta_L\kappa_0)+\eta_U B_n/2+l_n\}
            =
            O_p(1)
            .
        \end{gathered}
    \end{equation*}
\end{proof}

As a consequence, we obtain the following corollary.
\begin{corollary}[Global MAP under strong convexity]
\label{hb:map}
    Under Assumption~\ref{hb:convex}, a global MAP estimator $\hat{\boldsymbol{\theta}}^{\mathrm{MAP}}_n$ exists for each $\eta\in\cE$ and satisfies
    \begin{equation*}
        \|\hat{\boldsymbol{\theta}}^{\mathrm{MAP}}_n-\check{\boldsymbol{\theta}}_n\|
        \le
        \frac{2L_n}{\eta s_n\lambda_n\kappa}
        .
    \end{equation*}
    Every measurable choice is admissible.
\end{corollary}

\begin{proof}[Proof of Corollary~\ref{hb:map}]
    Assumption~\ref{hb:convex}~\ref{hb:convex:criterion}, \ref{hb:convex:prior} gives, for $\boldsymbol{h}=\boldsymbol{\theta}-\check{\boldsymbol{\theta}}_n$,
    \begin{equation*}
        \eta s_n\{G_n(\boldsymbol{\theta})-G_n(\check{\boldsymbol{\theta}}_n)\}
        -r_n(\boldsymbol{\theta})+r_n(\check{\boldsymbol{\theta}}_n)
        \ge
        \frac{\eta s_n\lambda_n\kappa}{2}\|\boldsymbol{h}\|^2-L_n\|\boldsymbol{h}\|
        \longrightarrow
        \infty
    \end{equation*}
    as $\|\boldsymbol{h}\|\to\infty$.
    Hence the continuous objective $\eta s_nG_n(\boldsymbol{\theta})-r_n(\boldsymbol{\theta})$ attains a global minimum $\hat{\boldsymbol{\theta}}^{\mathrm{MAP}}_n$.
    Optimality of $\hat{\boldsymbol{\theta}}^{\mathrm{MAP}}_n$ relative to $\check{\boldsymbol{\theta}}_n$ gives
    \begin{equation*}
        \begin{split}
            \frac{\eta s_n\lambda_n\kappa}{2}
            \|\hat{\boldsymbol{\theta}}^{\mathrm{MAP}}_n-\check{\boldsymbol{\theta}}_n\|^2
            &\le
            \eta s_n\{G_n(\hat{\boldsymbol{\theta}}^{\mathrm{MAP}}_n)-G_n(\check{\boldsymbol{\theta}}_n)\}
            \\
            &\le
            r_n(\hat{\boldsymbol{\theta}}^{\mathrm{MAP}}_n)-r_n(\check{\boldsymbol{\theta}}_n)
            \\
            &\le
            L_n\|\hat{\boldsymbol{\theta}}^{\mathrm{MAP}}_n-\check{\boldsymbol{\theta}}_n\|
            .
        \end{split}
    \end{equation*}
    Thus $\|\hat{\boldsymbol{\theta}}^{\mathrm{MAP}}_n-\check{\boldsymbol{\theta}}_n\|\le2L_n/(\eta s_n\lambda_n\kappa)=o_p(s_n^{-1/2})$, which proves admissibility.
\end{proof}

For the Huber loss $\rho_c$ and score $\psi_c=\rho_c'$, let $w_c(r)=\1\{|r|<c\}$ denote the slope away from the kinks.
Then we obtain the following exact expansions.
\begin{lemma}[Exact Huber expansions]
\label{hb:local-expansion}
    For all real $r,d$,
    \begin{equation*}
        \begin{split}
            \rho_c(r+d)-\rho_c(r)-\psi_c(r)d
            &=
            \frac{1}{2}w_c(r)d^2+R_c(r,d)
            ,\\
            \psi_c(r+d)-\psi_c(r)-w_c(r)d
            &=
            T_c(r,d)
            ,
        \end{split}
    \end{equation*}
    where
    \begin{equation*}
        |R_c(r,d)|
        \le
        d^2\1\{\bigl|
            |r|-c
        \bigr|\le|d|\}
        ,\quad
        |T_c(r,d)|
        \le
        |d|\1\{\bigl|
            |r|-c
        \bigr|\le|d|\}
        .
    \end{equation*}
\end{lemma}

\begin{proof}[Proof of Lemma~\ref{hb:local-expansion}]
    Absolute continuity of $\psi_c$ and its almost-everywhere derivative $w_c$ give
    \begin{equation*}
        \begin{split}
            T_c(r,d)
            &=
            d\int_0^1\{w_c(r+td)-w_c(r)\}\dd t
            ,\\
            R_c(r,d)
            &=
            d\int_0^1\{\psi_c(r+td)-\psi_c(r)\}\dd t-\frac{1}{2}w_c(r)d^2
            \\
            &=
            d^2\int_0^1(1-t)\{w_c(r+td)-w_c(r)\}\dd t
            .
        \end{split}
    \end{equation*}
    If $||r|-c|>|d|$, then $r+td$ and $r$ lie on the same side of $\pm c$ for all $0\le t\le1$, so $w_c(r+td)=w_c(r)$.
    Since $w_c$ takes values in $\{0,1\}$, it follows that, for all $r,d\in\RR$ and $0\le t\le1$,
    \begin{equation*}
        |w_c(r+td)-w_c(r)|
        \le
        \1\{\bigl|
            |r|-c
        \bigr|\le|d|\}
        .
    \end{equation*}
    Substituting this bound into the integral representations gives
    \begin{equation*}
        |T_c(r,d)|
        \le
        |d|\1\{\bigl|
            |r|-c
        \bigr|\le|d|\}
        ,\quad
        |R_c(r,d)|
        \le
        \frac{d^2}{2}\1\{\bigl|
            |r|-c
        \bigr|\le|d|\}
        \le
        d^2\1\{\bigl|
            |r|-c
        \bigr|\le|d|\}
        .
    \end{equation*}
\end{proof}

We next apply these expansions to the empirical loss.
For observations $(y_k,\boldsymbol{x}_k)\in\RR\times\RR^p$ indexed by $k$ and residuals $r_k(\boldsymbol{\beta})=y_k-\boldsymbol{x}_k^\top\boldsymbol{\beta}$, write $\bar{M}_n(\boldsymbol{\beta})=s_n^{-1}\sum_k\rho_c\{r_k(\boldsymbol{\beta})\}$ and define
\begin{equation*}
    \boldsymbol{J}_n^\dagger(\boldsymbol{\beta})
    =
    s_n^{-1}\sum_k\boldsymbol{x}_k\boldsymbol{x}_k^\top w_c\{r_k(\boldsymbol{\beta})\}
    .
\end{equation*}
For a sequence $\boldsymbol{a}_n$, consider the condition
\begin{equation}
\label{hb:eq:kink-band}
    s_n^{-1}\sum_k\|\boldsymbol{x}_k\|^2\1\Bigl\{
        \bigl|
            |r_k(\boldsymbol{a}_n)|-c
        \bigr|\le M\|\boldsymbol{x}_k\|s_n^{-1/2}
    \Bigr\}
    \to_p
    0
    ,
\end{equation}
for every fixed $M<\infty$.
Under this condition, we obtain the following result.
\begin{corollary}[Local empirical Huber expansions]
\label{hb:empirical-expansion}
    If \eqref{hb:eq:kink-band} holds for every fixed $M<\infty$, then
    \begin{align}
        \label{hb:eq:value-LQ}
        &\sup_{\|\boldsymbol{u}\|\le M}\Biggl|
            s_n\{\bar{M}_n(\boldsymbol{a}_n+s_n^{-1/2}\boldsymbol{u})-\bar{M}_n(\boldsymbol{a}_n)\}-\sqrt{s_n}\boldsymbol{U}_n(\boldsymbol{a}_n)^\top\boldsymbol{u}-\frac{1}{2}\boldsymbol{u}^\top\boldsymbol{J}_n^\dagger(\boldsymbol{a}_n)\boldsymbol{u}
        \Biggr|
        \to_p
        0
        ,\\
        \label{hb:eq:score-LQ}
        &\sup_{\|\boldsymbol{u}\|\le M}\Bigl\|
            \sqrt{s_n}\{\boldsymbol{U}_n(\boldsymbol{a}_n+s_n^{-1/2}\boldsymbol{u})-\boldsymbol{U}_n(\boldsymbol{a}_n)\}-\boldsymbol{J}_n^\dagger(\boldsymbol{a}_n)\boldsymbol{u}
        \Bigr\|
        \to_p
        0
        .
    \end{align}
\end{corollary}

\begin{proof}[Proof of Corollary~\ref{hb:empirical-expansion}]
    Put $d_k=r_k(\boldsymbol{a}_n+s_n^{-1/2}\boldsymbol{u})-r_k(\boldsymbol{a}_n)=-s_n^{-1/2}\boldsymbol{x}_k^\top\boldsymbol{u}$.
    Applying Lemma~\ref{hb:local-expansion} with $r=r_k(\boldsymbol{a}_n)$ and $d=d_k$ and summing over $k$ give
    \begin{equation*}
        \begin{split}
            s_n\{\bar{M}_n(\boldsymbol{a}_n+s_n^{-1/2}\boldsymbol{u})-\bar{M}_n(\boldsymbol{a}_n)\}
            &=
            \sqrt{s_n}\boldsymbol{U}_n(\boldsymbol{a}_n)^\top\boldsymbol{u}
            +\frac{1}{2}\boldsymbol{u}^\top\boldsymbol{J}_n^\dagger(\boldsymbol{a}_n)\boldsymbol{u}
            +\sum_kR_c\{r_k(\boldsymbol{a}_n),d_k\}
            ,\\
            \sqrt{s_n}\{\boldsymbol{U}_n(\boldsymbol{a}_n+s_n^{-1/2}\boldsymbol{u})-\boldsymbol{U}_n(\boldsymbol{a}_n)\}
            &=
            \boldsymbol{J}_n^\dagger(\boldsymbol{a}_n)\boldsymbol{u}
            -s_n^{-1/2}\sum_k\boldsymbol{x}_kT_c\{r_k(\boldsymbol{a}_n),d_k\}
            .
        \end{split}
    \end{equation*}
    Since $|d_k|\le M\|\boldsymbol{x}_k\|s_n^{-1/2}$ for $\|\boldsymbol{u}\|\le M$, the remainder bounds in Lemma~\ref{hb:local-expansion} and \eqref{hb:eq:kink-band} yield
    \begin{equation*}
        \begin{split}
            \sup_{\|\boldsymbol{u}\|\le M}\Biggl|
                \sum_kR_c\{r_k(\boldsymbol{a}_n),d_k\}
            \Biggr|
            &\le
            \sup_{\|\boldsymbol{u}\|\le M}
            \sum_k d_k^2\1\bigl\{
                \bigl|
                    |r_k(\boldsymbol{a}_n)|-c
                \bigr|\le|d_k|
            \bigr\}
            \\
            &\le
            M^2s_n^{-1}\sum_k\|\boldsymbol{x}_k\|^2\1\Bigl\{
                \bigl|
                    |r_k(\boldsymbol{a}_n)|-c
                \bigr|\le M\|\boldsymbol{x}_k\|s_n^{-1/2}
            \Bigr\}
            \\
            &\to_p
            0
            ,
        \end{split}
    \end{equation*}
    and
    \begin{equation*}
        \begin{split}
            \sup_{\|\boldsymbol{u}\|\le M}\Biggl\|
                s_n^{-1/2}\sum_k\boldsymbol{x}_kT_c\{r_k(\boldsymbol{a}_n),d_k\}
            \Biggr\|
            &\le
            \sup_{\|\boldsymbol{u}\|\le M}
            s_n^{-1/2}\sum_k\|\boldsymbol{x}_k\||d_k|\1\bigl\{
                \bigl|
                    |r_k(\boldsymbol{a}_n)|-c
                \bigr|\le|d_k|
            \bigr\}
            \\
            &\le
            Ms_n^{-1}\sum_k\|\boldsymbol{x}_k\|^2\1\Bigl\{
                \bigl|
                    |r_k(\boldsymbol{a}_n)|-c
                \bigr|\le M\|\boldsymbol{x}_k\|s_n^{-1/2}
            \Bigr\}
            \\
            &\to_p
            0
            .
        \end{split}
    \end{equation*}
    Combining these bounds with the preceding identities yields \eqref{hb:eq:value-LQ} and \eqref{hb:eq:score-LQ}.
\end{proof}

Finally, we verify the clustered example.
Consider the fixed working whitening matrices and ridge penalty of Section~\ref{subsec:ri-lmm}, with $n=mG$, $\lambda>0$, and $\boldsymbol{Q}\succ0$.
Write $\tilde{\boldsymbol{X}}=\tilde{\boldsymbol{X}}_1$ and $\tilde{\boldsymbol{y}}=\tilde{\boldsymbol{y}}_1$ for a generic cluster.
Here $\|\boldsymbol{X}\|_F^2$ is the sum of the squared matrix entries, and $\|\boldsymbol{y}\|_1$ is the sum of the absolute vector entries.
Let $\boldsymbol{\beta}^\lambda$ denote the minimizer of the population Huber criterion plus the ridge penalty, whose existence and uniqueness follow from the following conditions.
\begin{assumption}[Clustered Huber regression]
\label{hb:clusters}
    \leavevmode
    \begin{enumerate}[label=(\roman*),leftmargin=*]
        \item\label{hb:clusters:moments}
        The clusters are independent and identically distributed, with fixed size $m$, and
        \begin{equation*}
            \EE\bigl[
                \|\tilde{\boldsymbol{X}}\|_F^2
            \bigr]
            <
            \infty
            ,\quad
            \EE\bigl[
                \|\tilde{\boldsymbol{y}}\|_1
            \bigr]
            <
            \infty
            ,\quad
            P^\ast\{|\tilde{r}_{1j}(\boldsymbol{\beta}^\lambda)|=c\}
            =
            0
            \quad
            (j=1,\ldots,m)
            ,
        \end{equation*}
        where $\boldsymbol{\beta}^\lambda$ is the population penalized target.

        \item\label{hb:clusters:penalty}
        The regularization parameters satisfy $\lambda_n>0$, $\lambda_n\to\lambda$, and $\sqrt{n}(\lambda_n-\lambda)\to0$.
        The log-prior factor is $r_n\equiv0$, and $\boldsymbol{K}^\ast=m^{-1}\var\{\boldsymbol{S}_1(\boldsymbol{\beta}^\lambda)\}\succ0$.
    \end{enumerate}
\end{assumption}

Define the population curvature by
\begin{equation*}
    \boldsymbol{J}^\ast
    =
    \frac{1}{m}\EE\bigl[
        \tilde{\boldsymbol{X}}_1^\top\boldsymbol{W}_1(\boldsymbol{\beta}^\lambda)\tilde{\boldsymbol{X}}_1
    \bigr]
    ,\quad
    \boldsymbol{W}_1(\boldsymbol{\beta})
    =
    \diag\{\1\{|\tilde{r}_{1j}(\boldsymbol{\beta})|<c\}\}_{j=1}^m
    .
\end{equation*}
Then we obtain the following result.
\begin{lemma}[Verification for fixed-size independent clusters]
\label{hb:huber-verification}
    Under Assumption~\ref{hb:clusters}, the population score is differentiable at $\boldsymbol{\beta}^\lambda$ with derivative $\boldsymbol{J}^\ast$, and its score central limit theorem has effective scale $n$.
    At every consistent evaluation point, $\boldsymbol{J}_n^\dagger$ and the centered cluster covariance in \eqref{hb:eq:cluster-variance} consistently estimate $\boldsymbol{J}^\ast$ and $\boldsymbol{K}^\ast$.
    The exact penalized estimator satisfies $\sqrt{n}(\check{\boldsymbol{\beta}}_n-\boldsymbol{\beta}^{\lambda})=O_p(1)$, and \eqref{hb:eq:value-LQ}--\eqref{hb:eq:score-LQ} hold whenever $\sqrt{n}(\boldsymbol{a}_n-\boldsymbol{\beta}^{\lambda})=O_p(1)$.
\end{lemma}
The proof also verifies Assumption~\ref{hb:convex}, so the conclusions of Lemma~\ref{hb:convex-verification} and Corollary~\ref{hb:map} hold for this model.

\begin{proof}[Proof of Lemma~\ref{hb:huber-verification}]
    For the population criterion $G(\boldsymbol{\beta})=m^{-1}\EE[\sum_j\rho_c\{\tilde{r}_{1j}(\boldsymbol{\beta})\}]+\lambda\rho(\boldsymbol{\beta})$, since $0\le\rho_c(u)\le c|u|$, the Cauchy--Schwarz inequality and Assumption~\ref{hb:clusters}~\ref{hb:clusters:moments} give
    \begin{equation*}
        0
        \le
        \EE\Biggl[
            \sum_j\rho_c\{\tilde{r}_{1j}(\boldsymbol{\beta})\}
        \Biggr]
        \le
        c\EE[\|\tilde{\boldsymbol{y}}_1\|_1]
        +c\sqrt{m}\|\boldsymbol{\beta}\|\EE[\|\tilde{\boldsymbol{X}}_1\|_F]
        <
        \infty
        .
    \end{equation*}
    The ridge penalty makes this finite convex criterion strongly convex and coercive, giving a unique minimizer $\boldsymbol{\beta}^\lambda$.
    Moreover, since $|\psi_c|\le c$, we have
    \begin{equation*}
        \sup_{\boldsymbol{\beta}}\|\boldsymbol{S}_i(\boldsymbol{\beta})\|
        \le
        c\sqrt{m}\|\tilde{\boldsymbol{X}}_i\|_F
        ,\quad
        \nabla G(\boldsymbol{\beta})
        =
        m^{-1}\EE[\boldsymbol{S}_1(\boldsymbol{\beta})]
        +\lambda\boldsymbol{Q}(\boldsymbol{\beta}-\boldsymbol{\mu})
        ,
    \end{equation*}
    where the integrable score bound justifies differentiation under the expectation.
    Thus
    \begin{equation*}
        \Psi(\boldsymbol{\beta}^\lambda)
        +
        \lambda\boldsymbol{Q}(\boldsymbol{\beta}^\lambda-\boldsymbol{\mu})
        =
        \boldsymbol{0}
        .
    \end{equation*}

    To control the kinks, define
    \begin{equation*}
        B_i(\delta)
        =
        \sum_{j=1}^m\|\tilde{\boldsymbol{x}}_{ij}\|^2
        \1\{\bigl|
            |\tilde{r}_{ij}(\boldsymbol{\beta}^\lambda)|-c
        \bigr|\le\delta\|\tilde{\boldsymbol{x}}_{ij}\|\}
        .
    \end{equation*}
    The no-atom and moment conditions in Assumption~\ref{hb:clusters}~\ref{hb:clusters:moments} and dominated convergence imply
    \begin{equation*}
        B_1(\delta)\to0\ \text{almost surely}
        ,\quad
        0\le B_1(\delta)\le\|\tilde{\boldsymbol{X}}_1\|_F^2
        ,\quad
        \EE[B_1(\delta)]\to0
        \quad
        (\delta\downarrow0)
        .
    \end{equation*}
    Applying the score expansion in Lemma~\ref{hb:local-expansion} to $\tilde{r}_{1j}(\boldsymbol{\beta}^\lambda+\boldsymbol{h})=\tilde{r}_{1j}(\boldsymbol{\beta}^\lambda)-\tilde{\boldsymbol{x}}_{1j}^\top\boldsymbol{h}$ for each $j=1,\ldots,m$ and taking expectations show that $\Psi$ is differentiable at $\boldsymbol{\beta}^\lambda$ with derivative $\boldsymbol{J}^\ast$, since
    \begin{equation*}
        \begin{split}
            \|\Psi(\boldsymbol{\beta}^\lambda+\boldsymbol{h})-\Psi(\boldsymbol{\beta}^\lambda)-\boldsymbol{J}^\ast\boldsymbol{h}\|
            &\le
            m^{-1}\|\boldsymbol{h}\|\EE[B_1(\|\boldsymbol{h}\|)]
            \\
            &=
            o(\|\boldsymbol{h}\|)
            .
        \end{split}
    \end{equation*}

    For any $\boldsymbol{a}_n\to_p\boldsymbol{\beta}^\lambda$ and $\delta>0$, on the event $\{\|\boldsymbol{a}_n-\boldsymbol{\beta}^\lambda\|\le\delta\}$, bounding the indicator differences $|w_c\{\tilde{r}_{ij}(\boldsymbol{a}_n)\}-w_c\{\tilde{r}_{ij}(\boldsymbol{\beta}^\lambda)\}|$ as in Lemma~\ref{hb:local-expansion} and applying the law of large numbers \citep[Chapter~2]{vandervaart1998asymptotic} give
    \begin{equation*}
        \begin{split}
            \|\boldsymbol{J}_n^\dagger(\boldsymbol{a}_n)-\boldsymbol{J}^\ast\|
            &\le
            \frac{1}{mG}\sum_{i=1}^GB_i(\delta)
            +
            \|\boldsymbol{J}_n^\dagger(\boldsymbol{\beta}^\lambda)-\boldsymbol{J}^\ast\|
            \\
            &=
            m^{-1}\EE[B_1(\delta)]+o_p(1)
            .
        \end{split}
    \end{equation*}
    Since $P^\ast(\|\boldsymbol{a}_n-\boldsymbol{\beta}^\lambda\|\le\delta)\to1$, letting $n\to\infty$ and then $\delta\downarrow0$ yields $\boldsymbol{J}_n^\dagger(\boldsymbol{a}_n)\to_p\boldsymbol{J}^\ast$.

    The score bound also gives $\EE[\|\boldsymbol{S}_1(\boldsymbol{\beta}^\lambda)\|^2]<\infty$.
    Since $n=mG$, the central limit theorem and Assumption~\ref{hb:clusters}~\ref{hb:clusters:penalty} give
    \begin{equation*}
        \begin{gathered}
            \sqrt{n}\{\boldsymbol{U}_n(\boldsymbol{\beta}^\lambda)-\Psi(\boldsymbol{\beta}^\lambda)\}
            =
            \frac{1}{\sqrt{mG}}\sum_{i=1}^G
            \{\boldsymbol{S}_i(\boldsymbol{\beta}^\lambda)-\EE[\boldsymbol{S}_i(\boldsymbol{\beta}^\lambda)]\}
            \to_d
            N(\boldsymbol{0},\boldsymbol{K}^\ast)
            ,\\
            \sqrt{n}\boldsymbol{F}_n(\boldsymbol{\beta}^\lambda)
            =
            \sqrt{n}\{\boldsymbol{U}_n(\boldsymbol{\beta}^\lambda)-\Psi(\boldsymbol{\beta}^\lambda)\}
            +\sqrt{n}(\lambda_n-\lambda)\boldsymbol{Q}(\boldsymbol{\beta}^\lambda-\boldsymbol{\mu})
            =
            O_p(1)
            .
        \end{gathered}
    \end{equation*}
    Since $\psi_c$ is $1$-Lipschitz,
    \begin{equation*}
        \|\boldsymbol{F}_n(\boldsymbol{\beta})-\boldsymbol{F}_n(\boldsymbol{\beta}')\|
        \le
        \Biggl\{
            n^{-1}\sum_i\|\tilde{\boldsymbol{X}}_i\|_F^2+\lambda_n\|\boldsymbol{Q}\|
        \Biggr\}
        \|\boldsymbol{\beta}-\boldsymbol{\beta}'\|
        ,
    \end{equation*}
    where $n^{-1}\sum_i\|\tilde{\boldsymbol{X}}_i\|_F^2+\lambda_n\|\boldsymbol{Q}\|=O_p(1)$ by the law of large numbers.
    Thus Assumption~\ref{hb:convex} holds with $r_n\equiv0$ and $\kappa$ equal to the smallest eigenvalue of $\boldsymbol{Q}$.
    Lemma~\ref{hb:convex-verification} gives
    \begin{equation*}
        \|\check{\boldsymbol{\beta}}_n-\boldsymbol{\beta}^\lambda\|
        \le
        \frac{\|\boldsymbol{F}_n(\boldsymbol{\beta}^\lambda)\|}{\lambda_n\kappa}
        =
        O_p(n^{-1/2})
        .
    \end{equation*}

    Now suppose $\sqrt{n}(\boldsymbol{a}_n-\boldsymbol{\beta}^\lambda)=O_p(1)$.
    On the event $\{\sqrt{n}\|\boldsymbol{a}_n-\boldsymbol{\beta}^\lambda\|\le T\}$, the triangle inequality gives $||\tilde{r}_{ij}(\boldsymbol{\beta}^\lambda)|-c|\le||\tilde{r}_{ij}(\boldsymbol{a}_n)|-c|+T\|\tilde{\boldsymbol{x}}_{ij}\|n^{-1/2}$, and hence
    \begin{equation*}
        \frac{1}{n}\sum_{i=1}^G\sum_{j=1}^m\|\tilde{\boldsymbol{x}}_{ij}\|^2\1\Bigl\{
            \bigl|
                |\tilde{r}_{ij}(\boldsymbol{a}_n)|-c
            \bigr|\le M\|\tilde{\boldsymbol{x}}_{ij}\|n^{-1/2}
        \Bigr\}
        \le
        \frac{1}{mG}\sum_{i=1}^GB_i\{(M+T)/\sqrt{n}\}
        .
    \end{equation*}
    By Markov's inequality, for every $\varepsilon>0$, the probability that the left side exceeds $\varepsilon$ is at most
    \begin{equation*}
        P^\ast\{\sqrt{n}\|\boldsymbol{a}_n-\boldsymbol{\beta}^\lambda\|>T\}
        +
        \frac{\EE[B_1\{(M+T)/\sqrt{n}\}]}{m\varepsilon}
        .
    \end{equation*}
    Letting $n\to\infty$ and then $T\to\infty$ proves \eqref{hb:eq:kink-band}, and Corollary~\ref{hb:empirical-expansion} gives \eqref{hb:eq:value-LQ} and \eqref{hb:eq:score-LQ}.

    Finally, for score covariance estimation, put
    \begin{equation*}
        A_i
        =
        c\sqrt{m}\|\tilde{\boldsymbol{X}}_i\|_F
        ,\quad
        D_i(\delta)
        =
        \min\{\delta\|\tilde{\boldsymbol{X}}_i\|_F^2,2A_i\}
        .
    \end{equation*}
    Since $\psi_c$ is $1$-Lipschitz and $\|\boldsymbol{S}_i(\boldsymbol{\beta})\|\le A_i$, uniformly over $\|\boldsymbol{\beta}-\boldsymbol{\beta}^\lambda\|\le\delta$,
    \begin{equation*}
        \begin{split}
            \|\boldsymbol{S}_i(\boldsymbol{\beta})-\boldsymbol{S}_i(\boldsymbol{\beta}^\lambda)\|
            &\le
            D_i(\delta)
            \le
            2A_i
            ,\\
            \|\boldsymbol{S}_i(\boldsymbol{\beta})\boldsymbol{S}_i(\boldsymbol{\beta})^\top
            -\boldsymbol{S}_i(\boldsymbol{\beta}^\lambda)\boldsymbol{S}_i(\boldsymbol{\beta}^\lambda)^\top\|
            &\le
            \{\|\boldsymbol{S}_i(\boldsymbol{\beta})\|+\|\boldsymbol{S}_i(\boldsymbol{\beta}^\lambda)\|\}
            \|\boldsymbol{S}_i(\boldsymbol{\beta})-\boldsymbol{S}_i(\boldsymbol{\beta}^\lambda)\|
            \\
            &\le
            2A_iD_i(\delta)
            \\
            &\le
            4A_i^2
            .
        \end{split}
    \end{equation*}
    Since $\EE[A_1^2]<\infty$, dominated convergence gives $\EE[D_1(\delta)]\to0$ and $\EE[2A_1D_1(\delta)]\to0$ as $\delta\downarrow0$.
    Applying the law of large numbers at fixed $\delta$ and then letting $\delta\downarrow0$ give, at every consistent $\boldsymbol{a}_n$,
    \begin{equation*}
        \begin{split}
            \bar{\boldsymbol{S}}
            =
            G^{-1}\sum_{i=1}^G\boldsymbol{S}_i(\boldsymbol{a}_n)
            &\to_p
            \EE[\boldsymbol{S}_1(\boldsymbol{\beta}^\lambda)]
            ,\\
            G^{-1}\sum_{i=1}^G\boldsymbol{S}_i(\boldsymbol{a}_n)\boldsymbol{S}_i(\boldsymbol{a}_n)^\top
            &\to_p
            \EE[\boldsymbol{S}_1(\boldsymbol{\beta}^\lambda)\boldsymbol{S}_1(\boldsymbol{\beta}^\lambda)^\top]
            .
        \end{split}
    \end{equation*}
    Evaluating $\hat{\boldsymbol{K}}$ in \eqref{hb:eq:cluster-variance} at $\boldsymbol{a}_n$, we obtain
    \begin{equation*}
        \begin{split}
            \hat{\boldsymbol{K}}
            &=
            \frac{G}{m(G-1)}
            \Biggl\{
                G^{-1}\sum_{i=1}^G\boldsymbol{S}_i(\boldsymbol{a}_n)\boldsymbol{S}_i(\boldsymbol{a}_n)^\top
                -
                \bar{\boldsymbol{S}}\bar{\boldsymbol{S}}^\top
            \Biggr\}
            \\
            &\to_p
            m^{-1}\var\{\boldsymbol{S}_1(\boldsymbol{\beta}^\lambda)\}
            =
            \boldsymbol{K}^\ast
            .
        \end{split}
    \end{equation*}
\end{proof}

Combining these results, we obtain the following corollary.
\begin{corollary}[Calibrated inference for Huber regression]
\label{hb:huber-limits}
    Under Assumption~\ref{hb:clusters}, the conclusions of Propositions~\ref{prop:baseline-clt} and~\ref{prop:bvm-baseline}, Theorem~\ref{thm:ls-calib}, and Corollaries~\ref{cor:coverage} and~\ref{cor:functionals} hold for this Huber criterion.
    The posterior mean and covariance conclusions of Lemma~\ref{lem:postvar-consistency} also hold, and their Monte Carlo versions additionally require Assumption~\ref{ass:MC}.
\end{corollary}

\begin{proof}[Proof of Corollary~\ref{hb:huber-limits}]
    By Lemma~\ref{hb:huber-verification}, $\sqrt{n}(\check{\boldsymbol{\beta}}_n-\boldsymbol{\beta}^\lambda)=O_p(1)$.
    Applying \eqref{hb:eq:score-LQ} at $\boldsymbol{a}_n=\boldsymbol{\beta}^\lambda$ with $\boldsymbol{u}=\sqrt{n}(\check{\boldsymbol{\beta}}_n-\boldsymbol{\beta}^\lambda)$ and adding the exact ridge score increment give
    \begin{equation*}
        \begin{split}
            \boldsymbol{0}
            &=
            \sqrt{n}\boldsymbol{F}_n(\check{\boldsymbol{\beta}}_n)
            \\
            &=
            \sqrt{n}\boldsymbol{F}_n(\boldsymbol{\beta}^\lambda)
            +\{\boldsymbol{J}_n^\dagger(\boldsymbol{\beta}^\lambda)+\lambda_n\boldsymbol{Q}\}
            \sqrt{n}(\check{\boldsymbol{\beta}}_n-\boldsymbol{\beta}^\lambda)
            +o_p(1)
            .
        \end{split}
    \end{equation*}
    Since $\boldsymbol{J}_n^\dagger(\boldsymbol{\beta}^\lambda)+\lambda_n\boldsymbol{Q}\to_p\boldsymbol{J}_\lambda^\ast\coloneq\boldsymbol{J}^\ast+\lambda\boldsymbol{Q}\succ0$ by Lemma~\ref{hb:huber-verification}, Slutsky's lemma gives
    \begin{equation*}
        \begin{split}
            \sqrt{n}(\check{\boldsymbol{\beta}}_n-\boldsymbol{\beta}^\lambda)
            &=
            -(\boldsymbol{J}_\lambda^\ast)^{-1}\sqrt{n}\boldsymbol{F}_n(\boldsymbol{\beta}^\lambda)+o_p(1)
            \\
            &\to_d
            N(\boldsymbol{0},\boldsymbol{V}_{\mathrm{target}}^\ast)
            .
        \end{split}
    \end{equation*}
    For any baseline estimator $\tilde{\boldsymbol{\beta}}_n$ satisfying Assumption~\ref{ass:baseline}, strong convexity and $\boldsymbol{F}_n(\check{\boldsymbol{\beta}}_n)=\boldsymbol{0}$ give $\lambda_n\kappa\|\tilde{\boldsymbol{\beta}}_n-\check{\boldsymbol{\beta}}_n\|\le\|\boldsymbol{F}_n(\tilde{\boldsymbol{\beta}}_n)\|=o_p(n^{-1/2})$, so every baseline estimator and admissible center has the same expansion.

    For each fixed $M<\infty$, applying \eqref{hb:eq:value-LQ} at $\boldsymbol{a}_n=\check{\boldsymbol{\beta}}_n$ and using $\boldsymbol{F}_n(\check{\boldsymbol{\beta}}_n)=\boldsymbol{0}$ and $\boldsymbol{J}_n^\dagger(\check{\boldsymbol{\beta}}_n)\to_p\boldsymbol{J}^\ast$ give, uniformly over $\|\boldsymbol{u}\|\le M$,
    \begin{equation*}
        \begin{split}
            n\{G_n(\check{\boldsymbol{\beta}}_n+n^{-1/2}\boldsymbol{u})-G_n(\check{\boldsymbol{\beta}}_n)\}
            &=
            \frac{1}{2}\boldsymbol{u}^\top
            \{\boldsymbol{J}_n^\dagger(\check{\boldsymbol{\beta}}_n)+\lambda_n\boldsymbol{Q}\}\boldsymbol{u}
            +o_p(1)
            \\
            &=
            \frac{1}{2}\boldsymbol{u}^\top\boldsymbol{J}_\lambda^\ast\boldsymbol{u}+o_p(1)
            .
        \end{split}
    \end{equation*}
    For the kernels $g_n^\eta(\boldsymbol{u})=\exp[-n\eta\{G_n(\check{\boldsymbol{\beta}}_n+n^{-1/2}\boldsymbol{u})-G_n(\check{\boldsymbol{\beta}}_n)\}]$ from the proof of Proposition~\ref{prop:bvm-baseline} with $r_n\equiv0$, boundedness of $\cE$ gives $\sup_{\eta\in\cE}\sup_{\|\boldsymbol{u}\|\le M}|\log g_n^\eta(\boldsymbol{u})+\eta\boldsymbol{u}^\top\boldsymbol{J}_\lambda^\ast\boldsymbol{u}/2|\to_p0$.
    Strong convexity as in the proof of Lemma~\ref{hb:convex-verification} also gives, for some fixed $c_0>0$ and all sufficiently large $n$,
    \begin{equation*}
        \sup_{\eta\in\cE}g_n^\eta(\boldsymbol{u})
        \le
        \exp(-c_0\|\boldsymbol{u}\|^2)
        \quad
        (\boldsymbol{u}\in\RR^p)
        .
    \end{equation*}
    Applying Lemma~\ref{supp-lem:localization} with $\cU_n=\RR^p$ therefore yields
    \begin{equation*}
        \sup_{\eta\in\cE}
        \bigl\|
            Q_n^{\check{\boldsymbol{\beta}}_n,\eta}
            -
            N\{\boldsymbol{0},(\eta\boldsymbol{J}_\lambda^\ast)^{-1}\}
        \bigr\|_{\TV}
        \to_p
        0
        .
    \end{equation*}
    Translation to admissible centers and the covariance identity in \eqref{eq:def-lsc-map} give calibration and coverage by the proofs in Section~\ref{supp:pf-main}.
    Finally, the moment bound following Lemma~\ref{hb:convex-verification} justifies the truncation argument in Lemma~\ref{lem:postvar-consistency}, giving
    \begin{equation*}
        \sqrt{n}(\boldsymbol{\beta}_{\mathrm{GB}}-\check{\boldsymbol{\beta}}_n)
        \to_p
        \boldsymbol{0}
        ,\quad
        \|n\boldsymbol{\Sigma}_{\mathrm{post},n}-(\eta\boldsymbol{J}_\lambda^\ast)^{-1}\|
        \to_p
        0
        ,
    \end{equation*}
    where $\boldsymbol{\beta}_{\mathrm{GB}}$ is the exact posterior mean.
    Adding the errors in Assumption~\ref{ass:MC} proves the Monte Carlo conclusions.
\end{proof}

\section{Gibbs sampling for the empirical illustration}
\label{supp:detail-exp}

This section describes the Gibbs sampler and the numerical settings of Section~\ref{hb:experiment}.

\subsection{An auxiliary-variable representation}
\label{supp-subsec:sample-algo}
We first introduce an auxiliary variable that yields tractable full conditional distributions.
Within this section, for the scalar ($p=1$) setting of Section~\ref{hb:experiment} with scalar ridge parameters $\mu$ and $Q>0$, index the $n=mG$ whitened pairs $(\tilde{y}_{ij},\tilde{x}_{ij})$ from Section~\ref{subsec:ri-lmm} by $k=1,\ldots,n$ as $(\tilde{y}_k,\tilde{x}_k)$, and write the penalized loss as
\begin{equation*}
    L_n(\beta)
    =
    \sum_{k=1}^{n}\rho_c(\tilde{y}_k-\tilde{x}_k\beta)+\frac{\lambda_n n Q}{2}(\beta-\mu)^2
    ,\quad
    \cE_n(\beta)
    =
    L_n(\beta)-L_n(\check{\beta}_n)
    ,
\end{equation*}
where $\check{\beta}_n$ is the unique minimizer of $L_n$ and $\lambda_n>0$.
The loss is continuously differentiable and strongly convex, so $\cE_n(\beta)\ge\lambda_n n Q(\beta-\check{\beta}_n)^2/2$.
Consequently, for every $\eta>0$, $C_{n,\eta}\coloneq\int_{\RR}\exp\{-\eta\cE_n(b)\}\dd b$ is finite and strictly positive.
The desired posterior has density $p_{n,\eta}(\beta)=C_{n,\eta}^{-1}\exp\{-\eta\cE_n(\beta)\}$, since subtracting the minimum changes only its normalizing constant.

Introduce a scalar auxiliary variable $T$ with joint density
\begin{equation}
\label{hb:eq:gibbs-joint}
    p_{n,\eta}(\beta,t)
    =
    \frac{\eta}{C_{n,\eta}}\exp(-\eta t)\1\{\cE_n(\beta)<t\}
    ,\quad
    (\beta,t)\in\RR\times(0,\infty)
    .
\end{equation}
This is the slice augmentation of \citet[Section~3]{neal2003slice}, written in terms of the loss difference $\cE_n$.
Its implementation here is a two-block Gibbs sampler with directly sampled full conditionals.
Integrating out the auxiliary variable gives
\begin{equation*}
    \int_0^\infty p_{n,\eta}(\beta,t)\dd t
    =
    \frac{1}{C_{n,\eta}}\int_{\cE_n(\beta)}^\infty\eta\exp(-\eta t)\dd t
    =
    p_{n,\eta}(\beta)
    .
\end{equation*}
Thus the augmentation preserves the entire Huber loss and ridge penalty.

We next derive the full conditional distributions.
For every $t>0$, strong convexity and coercivity imply
\begin{equation*}
    \{\beta:\cE_n(\beta)<t\}
    =
    (b_-(t),b_+(t))
    ,\quad
    b_-(t)
    <
    \check{\beta}_n
    <
    b_+(t)
    ,
\end{equation*}
where the two endpoints are the unique solutions of $\cE_n(b_-(t))=\cE_n(b_+(t))=t$ on their respective sides of $\check{\beta}_n$.
The full conditional densities of \eqref{hb:eq:gibbs-joint} are therefore
\begin{equation*}
    T\mid\beta
    \overset{d}{=}
    \cE_n(\beta)
    +\operatorname{Exp}(\eta)
    ,\quad
    \beta\mid T=t
    \sim
    \operatorname{Unif}\{b_-(t),b_+(t)\}
    ,
\end{equation*}
where $\operatorname{Exp}(\eta)$ denotes the exponential distribution with rate $\eta$, $\operatorname{Unif}(a,b)$ denotes the uniform distribution on $(a,b)$, and $\overset{d}{=}$ denotes equality in distribution.
Thus an iteration requires an exponential draw followed by a uniform draw.
Equivalently, with independent $U_0,U_1\sim\operatorname{Unif}(0,1)$, the update is
\begin{equation*}
    \begin{split}
        t^{(d+1)}
        &=
        \cE_n(\beta^{(d)})-\eta^{-1}\log U_0
        ,\\
        \beta^{(d+1)}
        &=
        b_-\{t^{(d+1)}\}+U_1\bigl[
            b_+\{t^{(d+1)}\}-b_-\{t^{(d+1)}\}
        \bigr]
        .
    \end{split}
\end{equation*}

\subsection{Analytical endpoints of the full conditional interval}
\label{hb:exact-sampler}
To calculate the interval endpoints $b_-(t)$ and $b_+(t)$, note that the quadratic expression of $L_n(\beta)$ changes only at
\begin{equation*}
    \Biggl\{
        \frac{\tilde{y}_k-c}{\tilde{x}_k},\frac{\tilde{y}_k+c}{\tilde{x}_k}:\tilde{x}_k\ne0,\ k=1,\ldots,n
    \Biggr\}
    .
\end{equation*}
Sort the distinct breakpoints.
On each resulting open interval $I_j$, the set $\cC_j$ of residuals satisfying $|\tilde{y}_k-\tilde{x}_k\beta|<c$ and the signs of the other residuals are fixed.
On this interval the loss is quadratic with curvature
\begin{equation*}
    a_j
    =
    \lambda_n n Q
    +
    \sum_{k\in\cC_j}\tilde{x}_k^2
    >
    0
    ,\quad
    \beta\in I_j
    .
\end{equation*}
The polynomial values and first derivatives agree at adjacent endpoints because the Huber loss is continuously differentiable.
The minimizer is consequently found from the zero of the continuous, strictly increasing piecewise-linear derivative.

On either side of the minimizer, move outwards and let $v$ be the inner endpoint of the quadratic interval containing the required root, taking $v=\check{\beta}_n$ for the interval containing $\check{\beta}_n$.
Write $\epsilon=-1$ on the left side and $\epsilon=1$ on the right side, let $a>0$ be the curvature on that interval, and put $g=\epsilon L_n'(v)\ge0$.
For a distance $h\ge0$ that remains within this interval,
\begin{equation*}
    \cE_n(v+\epsilon h)
    =
    \cE_n(v)+gh+\frac{1}{2}ah^2
    .
\end{equation*}
Given $t>\cE_n(v)$, the required distance is
\begin{equation*}
    h
    =
    \frac{2\{t-\cE_n(v)\}}{g+\sqrt{g^2+2a\{t-\cE_n(v)\}}}
    .
\end{equation*}
For equality $t=\cE_n(v)$, set $h=0$.
The endpoint is $v+\epsilon h$.
On each side of the minimizer, the breakpoint loss differences increase strictly as one moves outwards, so a binary search identifies the correct interval.

For $p>1$, the same augmentation permits successive Gibbs updates of $\beta_1,\ldots,\beta_p$ conditional on $T$ and the remaining coordinates.
Each coordinate conditional is uniform on a bounded interval because the restricted ridge-penalized loss is strongly convex and piecewise quadratic.
The breakpoints follow from the partial residuals, and the ridge contributes the appropriate quadratic and linear terms.

\subsection{Invariant distribution of the marginal chain}
\label{hb:gibbs-invariance}

This subsection verifies that the marginal chain for $\beta$ is reversible, irreducible, and aperiodic with invariant density $p_{n,\eta}$.
Put $w(t)=b_+(t)-b_-(t)>0$.
Integrating over the auxiliary update yields the transition density
\begin{equation*}
    k_{n,\eta}(b,b')
    =
    \int_{\max\{\cE_n(b),\cE_n(b')\}}^\infty\frac{\eta\exp[-\eta\{t-\cE_n(b)\}]}{w(t)}\dd t
    .
\end{equation*}
Multiplication by $p_{n,\eta}(b)$ cancels the factor $\exp\{\eta\cE_n(b)\}$, giving
\begin{equation*}
    p_{n,\eta}(b)k_{n,\eta}(b,b')
    =
    \frac{\eta}{C_{n,\eta}}\int_{\max\{\cE_n(b),\cE_n(b')\}}^\infty\frac{\exp(-\eta t)}{w(t)}\dd t
    .
\end{equation*}
The right side is symmetric in $b,b'$.
Integrating this balance identity with respect to $b$ proves invariance of $p_{n,\eta}$.
It also shows that the marginal chain is reversible.
For any two finite states $b,b'$, the integrand is strictly positive on an interval above their two loss differences, so $k_{n,\eta}(b,b')>0$.
Hence every measurable set of positive Lebesgue measure can be reached in one iteration.
A nontrivial cyclic partition is impossible, since each of its positive-measure sets would also have positive one-step return probability.
The marginal chain is therefore irreducible and aperiodic.

\subsection{Population target and convergence diagnostics}
\label{hb:population-numerics}
To evaluate the population target $\beta^\lambda$ in Section~\ref{hb:experiment}, put $\boldsymbol{A}=(2\boldsymbol{1}_m\boldsymbol{1}_m^\top+\boldsymbol{I}_m)^{-1/2}$, where $m=5$, and let $\boldsymbol{z}\in\{0,1\}^m$ denote a contamination pattern.
Its probability is $p(\boldsymbol{z})=0.1^{\sum_k z_k}0.9^{m-\sum_k z_k}$.
Conditional on this pattern, the whitened design and noise are independent centered Gaussian vectors.
For coordinate $j$, write
\begin{equation*}
    v_j
    =
    (\boldsymbol{A}\boldsymbol{A}^\top)_{jj}
    ,\quad
    w_j(\boldsymbol{z})
    =
    1+100\sum_{k=1}^m A_{jk}^2z_k
    ,\quad
    a
    =
    2-\beta
    .
\end{equation*}
The $j$th standardized residual $\tilde{r}_{1j}(\beta)$ of a generic cluster has variance $a^2v_j+w_j(\boldsymbol{z})$ and covariance $av_j$ with the $j$th whitened covariate $\tilde{x}_{1j}$.
For a centered jointly Gaussian pair $(X,R)$ with $\var(R)=s^2$, conditional expectation gives $\EE[X\mid R]=\cov(X,R)R/s^2$.
Integration by parts in the normal density yields $\EE[R\psi_c(R)]=s^2\EE[\psi_c'(R)]$, where the boundary term vanishes since $\psi_c$ is bounded, and its almost-everywhere derivative is $\1\{|R|<c\}$.
Applying these two identities to each residual and averaging over the contamination patterns gives
\begin{equation*}
    \Psi(\beta)
    =
    -\frac{1}{m}\sum_{\boldsymbol{z}\in\{0,1\}^m}p(\boldsymbol{z})\sum_{j=1}^m av_j\Biggl[
        2\Phi\Biggl\{
            \frac{c}{\sqrt{a^2v_j+w_j(\boldsymbol{z})}}
        \Biggr\}-1
    \Biggr]
    .
\end{equation*}
Solving $\Psi(\beta)+0.5\beta=0$ gave the target $\beta^\lambda\approx0.8163$.

Writing $\hat{J}_\lambda$ and $\hat{V}_{\mathrm{target}}$ for the scalar values of $\hat{\boldsymbol{J}}_\lambda$ and $\hat{\boldsymbol{V}}_{\mathrm{target}}$ in \eqref{hb:eq:cluster-variance}, the two chains for each data set and learning rate started at $\check{\beta}_n\pm4(\eta n\hat{J}_\lambda)^{-1/2}$.
The benchmark interval was $\check{\beta}_n\pm z_{0.975}\sqrt{\hat{V}_{\mathrm{target}}/n}$, where $z_{0.975}$ is the $0.975$ standard-normal quantile.

For the split-$\widehat{R}$ diagnostic, we divided each retained chain into halves of length $h=5000$.
Writing $W$ for the mean of the four within-half sample variances and $B$ for $h$ times the sample variance of their means, we calculated $\widehat{R}=\{[(h-1)W/h+B/h]/W\}^{1/2}$.
For each scalar summary $f$, we also formed $100$ consecutive nonoverlapping batches of length $b=100$ within each chain.
Letting $\bar{s}_{B,f}^2$ denote the mean of the two sample variances of batch means and $s_f^2$ the pooled sample variance of $f$, we estimated the effective sample size by $\widehat{\operatorname{ESS}}_f=\min\{D,Ds_f^2/(b\bar{s}_{B,f}^2)\}$.
Across all data sets and learning rates, the maximum split-$\widehat{R}$ for $\beta$ was $1.0005$, and the minimum effective sample sizes out of $D=20000$ retained draws were $14451$ for $\beta$ and $7015$ for $(\beta-\hat{\beta}_{\mathrm{GB}})^2$.

\end{document}